\documentclass[aps,prb,twocolumn,superscriptaddress,10pt]{revtex4-2}

\usepackage{amsfonts}
\usepackage{amsmath}
\usepackage{amssymb}
\usepackage{amsthm}
\usepackage{mathrsfs}
\usepackage{latexsym}
\usepackage[table]{xcolor} 
\usepackage{multirow} 
\usepackage{dirtytalk} 
\usepackage{hyperref}
\hypersetup{colorlinks=true, linkcolor=blue, citecolor=red, filecolor=blue, menucolor = green, runcolor = cyan, urlcolor=blue}
\usepackage{subfig} 
\usepackage{enumerate} 
\usepackage{geometry} 
\usepackage{graphicx} 
\usepackage{tikz} 
\usetikzlibrary{quantikz}
\usetikzlibrary{decorations.pathreplacing}
\usetikzlibrary{3d,calc}
\usepackage{tcolorbox} 
\usepackage{booktabs} 
\usepackage[font=footnotesize,labelfont=bf]{caption}
\usepackage{lipsum}
\usepackage{tikzpeople}
\usepackage{float} 
\usepackage[normalem]{ulem}
\usepackage{comment}
\usepackage{algorithm}
\usepackage{algpseudocode}
\usepackage{adjustbox}
\usepackage{framed}
\usepackage{soul} 

\renewcommand{\thetable}{\Roman{table}}

\newcommand{\R}{\mathbb{R}}

\newcommand{\N}{\mathbb{N}}

\newcommand{\Tr}{\textrm{Tr}}

\newtheorem{theorem}{Theorem}[section]

\begin{document}

\title{Learning out-of-time-ordered correlators with classical kernel methods}

\author{John Tanner}
\email{john.tanner@uwa.edu.au}
\affiliation{Department of Physics, University of Western Australia, 35 Stirling Hwy, Crawley WA, 6009, Australia}

\author{Jason Pye}
\affiliation{Department of Physics, University of Western Australia, 35 Stirling Hwy, Crawley WA, 6009, Australia}
\affiliation{Nordita, Stockholm University and KTH Royal Institute of Technology, Hannes Alfv\'ens v\"ag 12, SE-106 91 Stockholm, Sweden}

\author{Jingbo Wang}
\email{jingbo.wang@uwa.edu.au}
\affiliation{Department of Physics, University of Western Australia, 35 Stirling Hwy, Crawley WA, 6009, Australia}

\date{September 3, 2024}


\newgeometry{left=2cm,right=2cm}



\begin{abstract}
\noindent
Out-of-Time Ordered Correlators (OTOCs) are widely used to investigate information scrambling in quantum systems.
However, directly computing OTOCs with classical computers is an expensive procedure.
This is due to the need to classically simulate the dynamics of quantum many-body systems, which entails computational costs that scale rapidly with system size. 
Similarly, exact simulation of the dynamics with a quantum computer (QC) will either only be possible for short times with noisy intermediate-scale quantum (NISQ) devices, or will require a fault-tolerant QC which is currently beyond technological capabilities. 
This motivates a search for alternative approaches to determine OTOCs and related quantities. 
In this study, we explore four parameterised sets of Hamiltonians describing local one-dimensional quantum systems of interest in condensed matter physics. 
For each set, we investigate whether classical kernel methods (KMs) can accurately learn the XZ-OTOC and a particular sum of OTOCs, as functions of the Hamiltonian parameters. 
We frame the problem as a regression task, generating small batches of labelled data with classical tensor network methods for quantum many-body systems with up to 40 qubits.  
Using this data, we train a variety of standard kernel machines and observe that the Laplacian and radial basis function (RBF) kernels perform best, achieving a coefficient of determination (\(R^2\)) on the testing sets of at least 0.7167, with averages between 0.8112 and 0.9822 for the various sets of Hamiltonians, together with small root mean squared error and mean absolute error.
Hence, after training, the models can replace further uses of tensor networks for calculating an OTOC function of a system within the parameterised sets.
Accordingly, the proposed method can assist with extensive evaluations of an OTOC function.
\end{abstract}

\maketitle

\section{Introduction}

In recent years, there have been significant technical advances in quantum computing technologies. 
From photons, trapped ions and neutral atoms, to nuclear magnetic resonance and superconducting qubits~\cite{Madsen2022Photonic,Bruzewics2019Trapped,Evered2023Neutral,OBrien2022Molecular,Bravyi2022Superconducting}, these technologies work alongside theoretical research by providing opportunities to investigate a qualitatively rich variety of quantum many-body systems. 
Access to these experimental platforms has not only led to the emergence of novel ideas and new research directions, but has also facilitated the development of mathematical models and tools for characterising numerous aspects of quantum many-body dynamics. 
In particular, quantum information scrambling~\cite{Landsman2019Verified}
has become an active research topic, which investigates the propagation of local quantum information into non-local degrees of freedom, with the aim of characterising quantum many-body systems.

Initially considered in the context of black holes~\cite{Hayden2007Black,Sekino2008Fast,Shenker2014Black,Shenker2015Stringy}, scrambling is now studied in more general many-body systems~\cite{Richerme2014Propagation,Jurcevic2014Quasiparticle,Lewis2019Dynamics,Landsman2019Verified}, often through the lens of the so-called out-of-time-ordered correlator (OTOC)~\cite{Hashimoto2017OTOC,Xu2023Tutorial}. 
The OTOC, which first appeared over 50 years ago~\cite{Larkin1969Quasiclassical}, is a measure of the non-commutativity of Heisenberg operators separated in time. 
In quantum field theories, the commutator captures causal relationships between observables at different points in spacetime.
Accordingly, the OTOC can be viewed as a measure of whether a local operator can causally influence measurements of another spatially-distant operator after some time. 
The value of the OTOC thus provides a kind of information-theoretic ``light cone,'' illustrating how information propagates through a system. 
The light cone is often sharper than the general Lieb-Robinson bound~\cite{Lieb1972Velocity}, and can exhibit a range of interesting behaviours. 
For example, chaotic quantum systems exhibit a linear light cone~\cite{Roberts2016Butterfly}, many-body localised systems exhibit logarithmic light cones~\cite{Huang2016MBL,Chen2016MBL,Banuls2016MBL,Deng2017Logarithmic}, and other marginal systems exhibit light cones somewhere between linear and logarithmic~\cite{Slagle2017Marginal}.

OTOCs have also been used in discussions of quantum chaos~\cite{Hosur2016Chaos,Rozenbaum2017Lyapunov,GarciaMata2023Chaos}, including in the study of black holes~\cite{Oliviero2023Unscrambling}, thermalisation~\cite{Swingle2018Unscrambling,Balachandran2021Eigenstate}, and many-body scarred systems~\cite{Yuan2022Quantum}. 
The early-time exponential growth of the OTOC was initially thought to be a diagnosis of chaos, and was even used to define a quantum analogue of the Lyapunov exponent~\cite{Maldacena2016Bound} and the Lyapunov spectrum~\cite{Gharibyan2019Lyapunov}. 
However, more recent findings have revealed that the exponential growth of OTOCs is necessary but not sufficient to diagnose chaos in quantum systems~\cite{Dowling2023Scrambling}, suggesting that this connection requires further research.  
This, together with the significant role that the OTOC plays in studying quantum information scrambling, and other intriguing applications \cite{Zamani2022Floquet,Schuster2023Learning,Cotler2023Information}, motivates detailed investigations of the OTOC. 
In this work, we focus on methods to numerically approximate the OTOC, as a means to further explore its dynamics more efficiently.

Classical methods for calculating OTOCs require the simulation of a quantum many-body system, a task which generally involves computational costs that scale exponentially in the size of the system. 
This makes the implementation of these methods extremely difficult in practice. 
There are, however, a handful of specific quantum systems for which the OTOC can be determined efficiently with classical methods. 
This includes exactly solvable cases, such as the Sachdev-Ye-Kitaev model~\cite{Polchinski2016Spectrum,Maldacena2016Remarks}, 1-dimensional (1D) quantum Ising spin chain~\cite{Lin2018Ising}, and other integrable systems~\cite{McGinley2019Slow}. 
Similarly, for 1D systems evolving under local Hamiltonians (such as those considered in this work), tensor network methods can be used to reliably estimate OTOCs~\cite{Heyl2018Detecting,Xu2019Accessing,Paeckel2019Evolution}. 
However for quantum systems in higher dimensions, contracting the corresponding tensor networks is a \(\#P\)-complete problem~\cite{Schuch2007Computational,Garcia2012Exact}. 
This poses a significant challenge for the use of tensor network methods for quantum systems in two or more dimensions.

On the other hand, procedures making use of quantum devices to directly compute OTOCs have been proposed~\cite{Swingle2016Measuring,Garttner2017Measuring,Li2017Measuring,Pg2021Exponential,Blocher2022Measuring}. 
However many of these procedures require either specialised quantum systems to simulate specific Hamiltonians, or access to a universal fault-tolerant quantum computer (QC). 
In the latter case one would, in principle, be able to simulate arbitrary quantum many-body systems. 
However, the only devices currently in operation are noisy intermediate-scale quantum (NISQ) devices~\cite{Preskill2018NISQ}.
In the case of 1D quantum systems, we believe that NISQ technologies will require significant improvements before they can match the simulation capabilities of classical tensor network methods (see Section \ref{Section5}).
For this reason, we focus on OTOCs describing 1D systems and explore applications of classical machine learning (ML) methods for accurately and efficiently determining such OTOCs (and related quantities), in order to support continued research in this direction.

To this end, we investigate whether classical kernel methods (KMs)~\cite{Steinwart2008Support,Scholkopf2001Kernels,Mohri2018Foundations}, powerful ML algorithms used for data analysis tasks, may help 
to reduce the cost of extensively evaluating OTOCs describing 1D systems. 
Specifically, we explore whether KMs can accurately approximate such OTOCs from a small amount of training data.
While neural networks are typically more efficient to train with larger datasets, here we seek a method for making accurate inferences from a small amount of data, since producing the data can be expensive. 
The small size of the datasets then makes the use of KMs over neural networks favourable since, for a given choice of hyperparameters, KMs offer a deterministic training procedure, allowing us to avoid potential convergence issues that can arise with neural networks.
Additionally, in contrast to the growing computational cost of tensor network methods with system size (see Section \ref{Generating the data}), once we obtain a trained model using KMs, we can use it to predict new OTOC values in a time that scales linearly in just the number of training datapoints (i.e., independent of the system size).

In this work, we consider four parameterised sets of Hamiltonians, each containing a collection of 1D quantum systems which are studied in the condensed matter physics literature.
For each set, we investigate whether classical KMs can be used to learn a specific OTOC, namely the XZ-OTOC, and a useful sum of OTOCs (defined in Section~\ref{Information Spreading and Out-of-Time-Ordered Correlators}) as functions of the Hamiltonian parameters.
We formulate the problem as a regression task, generating labelled data with classical tensor network methods based on matrix product operators (MPOs)~\cite{Xu2019Accessing}, for quantum many-body system sizes up to 40 qubits. 
The input data is given by uniform random samples of the Hamiltonian parameters drawn from a subset of the parameter space (see Section \ref{The parameterised sets of Hamiltonians}).
And the labels are given by either the value of the associated XZ-OTOC, or the sum of OTOCs. 
We split the generated data into training and testing sets and apply regularised empirical risk minimisation (RERM) \cite{Scholkopf2001Kernels, Mohri2018Foundations} with KMs. 
Using a variety of standard kernels, we perform a 10-fold cross-validation (see Section 4.5 of~\cite{Mohri2018Foundations}) on the training data to find good hyperparameter values. We then use these hyperparameter values to train models on all of the training data, and apply the models to the testing data to assess their suitability for the problem.

An application of this method could be in situations where one is interested in understanding the behaviour of the OTOC over large regions of a Hamiltonian parameter space.
For example, this could be used to design a quantum system with a particular information scrambling behaviour.
The idea is that instead of extensively evaluating the OTOC (using, e.g., MPO methods) for many different Hamiltonian parameters, one can perform a small number of evaluations of the OTOC and then use KMs to approximate the OTOC for the remaining parameter values of interest.
The benefit of this approach is that after the models are trained, further evaluations of the model can be done more efficiently than continued evaluations using MPO methods.
In Section~\ref{Section5}, we elaborate upon how this provides an overall reduction in computational costs.

A previous work~\cite{Wu2020Artificial} applied restricted Boltzmann machines to learn early-time OTOCs, demonstrating their approach with an example involving the 2D transverse-field Ising model.
Given that our work focuses on 1D quantum systems, a direct comparison with~\cite{Wu2020Artificial} is challenging. 
Nonetheless we include this reference as it is the only closely related work of which we are aware.

The paper proceeds as follows. 
In Section \ref{Section2}, we explain the necessary background material for discussing the ML tasks in detail, followed by a formal introduction to OTOCs and the associated sum of OTOCs on which we focus. 
In Section \ref{Section3}, we describe the parameterised sets of Hamiltonians considered in this work and the associated ML tasks, together with the classical kernels which are applied to the ML problem. 
We then finish this section by discussing the methods used to generate the datasets.
In Section \ref{Section4}, we report the numerical results, which include learning performance metrics for the trained models making predictions on the training and testing sets for every problem instance. 
Finally, in Section \ref{Section5}, we discuss the results, describe how an overall reduction in runtime can be achieved, and conclude in Section~\ref{Conclusion}, providing suggestions for improving and extending our work in future research.

\section{Background}
\label{Section2}

\subsection{Regularised empirical risk minimisation}

Supervised ML algorithms aim to find a function, called a \emph{model}, that both accurately fits a training dataset of input-output pairs and makes accurate predictions for new data. 
However, if a model fits the training data too closely then it will often generalise poorly to unseen data, referred to as \emph{overfitting}. 
This poses a challenge for supervised ML algorithms. 
The RERM method \cite{Scholkopf2001Kernels,Mohri2018Foundations} addresses this challenge by minimising a combination of the empirical risk, which measures performance on training data using a loss function, and a regularisation term. 
By including the regularisation term, RERM directly discourages overly intricate models, such as those that may arise from fitting noisy data exactly. This can help in finding a more robust model which is less prone to overfitting.

Let \(\mathcal{D}=\{(\mathbf{x}_i,y_i)\}_{i=1}^{M}\subseteq\mathcal{X}\times\R\) be a \emph{training dataset}, where \(\mathcal{X}\equiv\R^d\) is the input data domain of dimension \(d\in\mathbb{N}\), \(\mathbf{x}_i\in\mathcal{X}\) is the \(i^{\textrm{th}}\) input training data sample, \(y_i\in\mathbb{R}\) is the label for the \(i^{\textrm{th}}\) training data sample, and \(M\in\N\) is the total number of training data samples. 
We denote the set of candidate models mapping \(\mathcal{X}\to\R\), called the \emph{hypothesis class}, by \(\textsc{Hyp}(\mathcal{X},\R)\). 
The \emph{regularised empirical risk functional}, denoted \(\mathcal{L}_{\mathcal{D}}:\textsc{Hyp}(\mathcal{X},\R)\to\R\), for \(\mathcal{D}\) is then defined by
\begin{align}
\label{RegularisedRisk}
\mathcal{L}_{\mathcal{D}}(f)=L_{\mathcal{D}}(f)+\lambda\Omega(f),
\end{align}
where \(L_{\mathcal{D}}:\textsc{Hyp}(\mathcal{X},\R)\to\R\) is a \emph{loss function} for \(\mathcal{D}\), \(\Omega:\textsc{Hyp}(\mathcal{X},\R)\to\R\) is the \emph{regularisation term}, and \(\lambda>0\) is the \emph{regularisation strength}. 
A common choice of the loss function for regression is given by the \emph{mean squared error} (MSE), \(L_{\mathcal{D}}(f)=\frac{1}{M}\sum_{i=1}^{M}\left(f(\mathbf{x}_i)-y_i\right)^2\). 
Similarly, common choices of the regularisation term include the \(l_1\)-norm and the squared \(l_2\)-norm of the model parameters, called Lasso and ridge regression (see Section 11.3 of~\cite{Mohri2018Foundations}), respectively.

RERM for \(\mathcal{D}\) over \(\textsc{Hyp}(\mathcal{X},\R)\) is then the procedure of finding a model \(f_{\textrm{opt}}\in\textsc{Hyp}(\mathcal{X},\R)\) which minimises the regularised empirical risk functional for \(\mathcal{D}\),
\begin{equation}
\label{OptimalModel}
f_{\textrm{opt}}=\underset{f\in\textsc{Hyp}(\mathcal{X},\R)}{\arg\min}\mathcal{L}_{\mathcal{D}}(f).
\end{equation}
By choosing \(\lambda\) in \eqref{RegularisedRisk} to be strictly positive, called \emph{regularisation}, we can sometimes help to bound the difference between the loss function for \(f_{\textrm{opt}}\) evaluated on the training dataset, and the loss function for \(f_{\textrm{opt}}\) evaluated on unseen data~\cite{Vapnik1998Statistical}. 
This means that if RERM returns a model which accurately fits the training dataset, then the model will likely also perform well on unseen data, hence improving generalisation.

\subsection{Kernel methods}
\label{Kernel methods}

\begin{figure*}[ht]
\centering
\begin{tikzpicture}
\draw (0.5-0.75,0.5-0.125+0.25) -- ++(2,0) -- ++(1,1) -- ++(-2,0) -- ++(-1,-1) ;
\node[circle, fill=yellow!100!blue, inner sep = 0.35mm] at (1.75,1.475) {};
\node[circle, fill=yellow!95!blue, inner sep = 0.35mm] at (0.85,1.3) {};
\node[circle, fill=yellow!88.4!blue, inner sep = 0.35mm] at (1.75+0.5-0.75,1.225) {};
\node[circle, fill=yellow!73.9!blue, inner sep = 0.35mm] at (0.55+0.5-0.75,0.915) {};
\node[circle, fill=yellow!65!blue, inner sep = 0.35mm] at (2.35+0.5-0.75,1.475) {};
\node[circle, fill=yellow!51.2!blue, inner sep = 0.35mm] at (0.75+0.65-0.75,1.175) {};
\node[circle, fill=yellow!41.9!blue, inner sep = 0.35mm] at (1.5+0.5-0.75,0.975) {};
\node[circle, fill=yellow!32.6!blue, inner sep = 0.35mm] at (1.25+0.5-0.75,1.525) {};
\node[circle, fill=yellow!23.3!blue, inner sep = 0.35mm] at (1.5+0.5-0.75,1.325) {};
\node[circle, fill=yellow!16.3!blue, inner sep = 0.35mm] at (2.2+0.5-0.75,1.175) {};
\node[circle, fill=yellow!11.2!blue, inner sep = 0.35mm] at (1.2+0.5-0.75,1.065) {};
\node[circle, fill=yellow!6.5!blue, inner sep = 0.35mm] at (0.9+0.5-0.75,0.965) {};
\node[circle, fill=yellow!0!blue, inner sep = 0.35mm] at (1.85+0.5-0.75,0.825) {};
\node at (1.75+0.5-0.75-0.75,2-0.125-1.75+0.25) {\(\mathcal{X}\equiv\R^d\)};
\draw[thick,opacity=0.35,fill=yellow!0!blue] (5+0.125,-0.25) -- ++(2,-0.25) -- ++(1,1.25) -- ++(-2,0.25) -- ++(-1,-1.25) ;
\draw[thick,opacity=0.35,fill=yellow!33!blue] (5+0.125,0.25) -- ++(2,-0.25) -- ++(1,1.25) -- ++(-2,0.25) -- ++(-1,-1.25) ;
\draw[thick,opacity=0.35,fill=yellow!66!blue] (5+0.125,0.75) -- ++(2,-0.25) -- ++(1,1.25) -- ++(-2,0.25) -- ++(-1,-1.25) ;
\draw[thick,opacity=0.35,fill=yellow!100!blue] (5+0.125,1.25) -- ++(2,-0.25) -- ++(1,1.25) -- ++(-2,0.25) -- ++(-1,-1.25) ;
\draw[dotted] (5+0.125,-0.5) -- ++(1,1) -- ++(0,2) -- ++(2,0) -- ++(0,-2) -- ++(-2,0) ;
\draw (5+0.125,-0.5) -- ++ (2,0) -- ++(1,1) -- ++(0,2) -- ++(-2,0) -- ++(-1,-1) -- ++ (2,0) -- ++(1,1) -- ++(-1,-1) -- ++(0,-2) -- ++(-2,0) -- ++(0,2) ;
\node[circle, fill=yellow!100!blue, inner sep = 0.35mm] at (2+5+0.125,2.1) {};
\node[circle, fill=yellow!95!blue, inner sep = 0.35mm] at (2+5+0.125-0.9,2.1-0.1075) {};
\node[circle, fill=yellow!88.4!blue, inner sep = 0.35mm] at (1.75+5+0.125,1.85) {};
\node[circle, fill=yellow!73.9!blue, inner sep = 0.35mm] at (0.55+5+0.125,1.54) {};
\node[circle, fill=yellow!65!blue, inner sep = 0.35mm] at (2.35+5+0.125,1.35) {};
\node[circle, fill=yellow!51.2!blue, inner sep = 0.35mm] at (0.75+5.15+0.125,1.05) {};
\node[circle, fill=yellow!41.9!blue, inner sep = 0.35mm] at (1.5+5+0.125,0.85) {};
\node[circle, fill=yellow!32.6!blue, inner sep = 0.35mm] at (1.25+5+0.125,0.65) {};
\node[circle, fill=yellow!23.3!blue, inner sep = 0.35mm] at (1.5+5+0.125,0.45) {};
\node[circle, fill=yellow!16.3!blue, inner sep = 0.35mm] at (2.2+5+0.125,0.3) {};
\node[circle, fill=yellow!11.2!blue, inner sep = 0.35mm] at (1.2+5+0.125,0.19) {};
\node[circle, fill=yellow!6.5!blue, inner sep = 0.35mm] at (0.9+5+0.125,0.09) {};
\node[circle, fill=yellow!0!blue, inner sep = 0.35mm] at (1.85+5+0.125,-0.05) {};
\node at (6.75-0.75+0.125,2.875-3.75) {\(\mathcal{F}\)} ;
\draw[thick] (11-0.25+0.5+0.25,0) -- ++(0,2.5) ;
\node[circle, fill=yellow!100!blue, inner sep = 0.35mm] at (11-0.25+0.5+0.25,2.35) {};
\node[circle, fill=yellow!95!blue, inner sep = 0.35mm] at (11-0.25+0.5+0.25,2.35-0.1075) {};
\node[circle, fill=yellow!88.4!blue, inner sep = 0.35mm] at (11-0.25+0.5+0.25,2.1) {};
\node[circle, fill=yellow!73.9!blue, inner sep = 0.35mm] at (11-0.25+0.5+0.25,1.79) {};
\node[circle, fill=yellow!65!blue, inner sep = 0.35mm] at (11-0.25+0.5+0.25,1.6) {};
\node[circle, fill=yellow!51.2!blue, inner sep = 0.35mm] at (11-0.25+0.5+0.25,1.3) {};
\node[circle, fill=yellow!41.9!blue, inner sep = 0.35mm] at (11-0.25+0.5+0.25,1.1) {};
\node[circle, fill=yellow!32.6!blue, inner sep = 0.35mm] at (11-0.25+0.5+0.25,0.9) {};
\node[circle, fill=yellow!23.3!blue, inner sep = 0.35mm] at (11-0.25+0.5+0.25,0.7) {};
\node[circle, fill=yellow!16.3!blue, inner sep = 0.35mm] at (11-0.25+0.5+0.25,0.55) {};
\node[circle, fill=yellow!11.2!blue, inner sep = 0.35mm] at (11-0.25+0.5+0.25,0.44) {};
\node[circle, fill=yellow!6.5!blue, inner sep = 0.35mm] at (11-0.25+0.5+0.25,0.34) {};
\node[circle, fill=yellow!0!blue, inner sep = 0.35mm] at (11-0.25+0.5+0.25,0.2) {};
\node at (11-0.25+0.5+0.25,2.75-3.25+0.25) {\(\R\)};
\draw[thick,-stealth ] (3.25-0.3,1.75-0.5-0.25) to[out=-30,in=-150] ++(1.75,0);
\node at (4.0725-0.3,1.5-0.5-0.5-0.25) {\(\phi(\cdot)\)};
\draw[thick,-stealth] (8.5+0.25,1.75-0.5-0.25+0.1) to[out=-30,in=-150] ++(1.75,0);
\node at (9.5+0.25,1.5-0.5-0.5-0.25+0.1) {\(\langle\phi(\cdot),\phi(x^\prime)\rangle_{\mathcal{F}}\)};
\draw[thick,-stealth ] (3-1,2.125-0.25) to[out=30,in=145] ++(9,0);
\node at (6.5,4.25-0.4-0.25) {\(\mathcal{K}(\cdot,x^\prime)\)};
\end{tikzpicture}
\caption{A kernel function \(\mathcal{K}\), which implicitly computes an inner-product in a high-dimensional feature space \(\mathcal{F}\), can be used to simplify a regression problem if the associated feature map \(\phi\) arranges the inputs in \(\mathcal{F}\) in a desirable way. For example, above we see points associated with different continuous labels (indicated by the varying colours) being arranged into different parallel hyperplanes. This allows the continuous value associated with the points to be extracted via a simple projection along some axis in \(\mathcal{F}\).}
\label{Fig1}
\end{figure*}
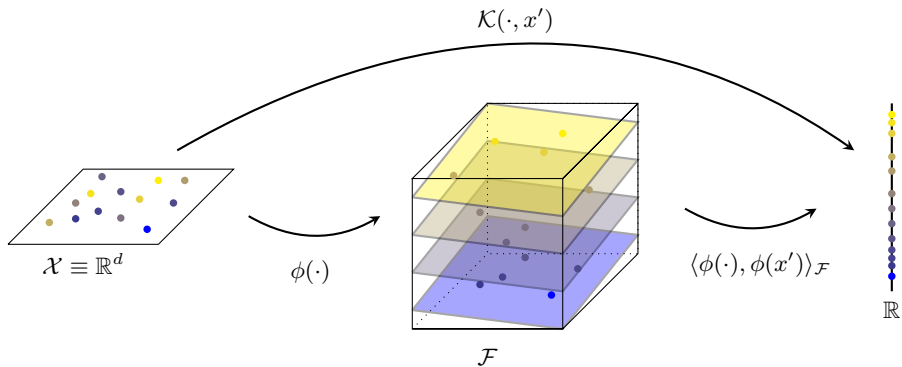

When selecting an algorithm for a supervised ML task, several factors need to be considered.
These include the dataset size and complexity, the type of problem (i.e., regression, classification, etc.), the efficiency of the training algorithm, and the availability of tools for implementation. 
KMs~\cite{Scholkopf2001Kernels,Steinwart2008Support,Mohri2018Foundations} are a collection of  ML algorithms used to capture complex patterns in small to moderate sized datasets.
The success of KMs stems largely from the use of kernel functions, which implicitly compute inner products between embeddings of input data in high-dimensional feature spaces. 
By utilising kernel functions, intricate non-linear structures in the original data can sometimes translate into standard linear functions in the feature space (see Figure \ref{Fig1}). 

Moreover, in many cases the training landscape for KMs is convex which, for a given choice of hyperparameters, enables the discovery of optimal model parameters via a deterministic procedure. 
However, the training procedure has a cubic runtime scaling in the number of training datapoints.
In contrast, back-propagation in conventional neural networks has a runtime scaling linearly in the number of training datapoints. 
So for large datasets, neural networks may be preferable. 
However if one uses a small enough training dataset, such as those used in this work, then kernels offer benefits including a deterministic training procedure, which motivates our use of them here.

Formally, a \emph{kernel} is a symmetric function \(\mathcal{K}:\mathcal{X}\times\mathcal{X}\to\R\) such that the \emph{Gram matrix} \(K_{ij}\equiv\mathcal{K}(x_i,x_j)\) of \(\mathcal{K}\) is positive semi-definite for all choices of the set \(\{x_1,\ldots,x_m\}\subseteq\mathcal{X}\) and all \(m\in\N\). It can be shown (see Chapter 2.2 of \cite{Scholkopf2001Kernels}) that any kernel \(\mathcal{K}\) can be expressed in the form
\begin{equation}
\label{FeatureMapKernel}
\mathcal{K}(x,x^\prime)=\left\langle\phi(x),\phi(x^\prime)\right\rangle_{\mathcal{F}},
\end{equation}
for some function \(\phi:\mathcal{X}\to\mathcal{F}\) and for all \(x,x^\prime\in\mathcal{X}\) (see Figure \ref{Fig1}). The function \(\phi:\mathcal{X}\to\mathcal{F}\) is called a \emph{feature map} and its codomain \(\mathcal{F}\) is a Hilbert space over \(\R\) called a \emph{feature space}.
The kernel $\mathcal{K}$ is then a measure of how similar two inputs are after being mapped into the feature space by \(\phi\).

Associated with each kernel \(\mathcal{K}\) is a Hilbert space over \(\R\) called the \emph{reproducing kernel Hilbert space} (RKHS) of \(\mathcal{K}\) (Definition 2.9 in~\cite{Scholkopf2001Kernels}), which we denote by \(\mathcal{R}_{\mathcal{K}}\). 
The RKHS \(\mathcal{R}_{\mathcal{K}}\) is a space of functions mapping \(\mathcal{X}\to\R\) defined such that
\begin{equation}
\label{RKHS}
\mathcal{R}_{\mathcal{K}}\equiv\overline{\textrm{span}}_{\R}\big\{\mathcal{K}(\cdot,x)|x\in\mathcal{X}\big\}.
\end{equation}
That is, each $x \in \mathcal{X}$ is an index for a function $\mathcal{K}(\cdot,x)$ mapping $\mathcal{X} \to \R$, defined such that $x' \mapsto \mathcal{K}(x',x)$ for all $x'\in\mathcal{X}$.
The span over \(\R\) of this collection of functions is itself a function space containing all real linear combinations of finitely many functions $\{ \mathcal{K}(\cdot,x_i) \}_{i=1}^m$ indexed by some subset $\{ x_i \}_{i=1}^m \subset \mathcal{X}$, where $m \in \N$.
One can define an inner product for such linear combinations $f = \sum_{i=1}^m \alpha_i \mathcal{K}(\cdot,x_i)$ and $g = \sum_{j=1}^{m^\prime} \beta_j \mathcal{K}(\cdot,x'_j)$ by
\begin{equation}
\label{RKHSInnerProduct}
 \langle f, g \rangle_{\mathcal{R}_{\mathcal{K}}} \equiv \sum_{i=1}^m\sum_{j=1}^{m^\prime} \alpha_i \beta_j \mathcal{K}(x_i,x'_j).
\end{equation}
The RKHS $\mathcal{R}_{\mathcal{K}}$ is then the completion of this function space with respect to the inner product in \eqref{RKHSInnerProduct}. 
Generally, an arbitrary element of the completed space $\mathcal{R}_{\mathcal{K}}$ cannot be written as a finite linear combination of kernel functions. 
However, under some weak condtions on \(\mathcal{X}\) and \(\mathcal{K}\) (see Lemma 4.33 in~\cite{Steinwart2008Support}), we can write an arbitrary model \(f\in\mathcal{R}_{\mathcal{K}}\) as a countable sum,
\begin{equation}
\label{RKHSmodel}
f(\cdot)=\sum_{i\in\N}\alpha_i\mathcal{K}(\cdot,x_i),
\end{equation}
for some \(\alpha_i\in\R\) and \(x_i\in\mathcal{X}\).

Performing RERM with KMs amounts to choosing an appropriate kernel $\mathcal{K}$ and using the corresponding RKHS $\mathcal{R}_{\mathcal{K}}$ as the hypothesis class, whose elements can be written as in \eqref{RKHSmodel}.
Optimisation over this space would appear to require searching for the best choice of infinitely many real coefficients $\{ \alpha_i \}_{i \in \N}$. 
However, the \emph{representer theorem} (Theorem 4.2 in~\cite{Scholkopf2001Kernels} and Theorem 6.11 in~\cite{Mohri2018Foundations}) reduces the problem to a finite-dimensional one. 
Suppose that we have a finite training dataset \(\mathcal{D}=\{(\mathbf{x}_i,y_i)\}_{i=1}^{M}\subseteq\mathcal{X}\times\R\) and a regularisation term \(\Omega(f)=g\left(\|f\|_{\mathcal{R}_{\mathcal{K}}}\right)\), where \(g:[0,\infty)\to\R\) is some strictly increasing function, and \(\|\cdot\|_{\mathcal{R}_{\mathcal{K}}}\) is the norm induced by the inner product on \(\mathcal{R}_{\mathcal{K}}\). 
Under these conditions, the representer theorem states that the functions in \(\mathcal{R}_{\mathcal{K}}\) which minimise the regularised empirical risk functional can be expressed in the form
\begin{equation}
\label{RTmodel}
f(\cdot)=\sum_{i=1}^{M}\alpha_i\mathcal{K}(\cdot,\mathbf{x}_i).
\end{equation}
Notice that \eqref{RTmodel} involves only finitely many real coefficients \(\{\alpha_i\}_{i=1}^{M}\) which define a linear combination of finitely many functions \(\{\mathcal{K}(\cdot,\mathbf{x}_i)\}_{i=1}^{M}\) associated with the input training data \(\{\mathbf{x}_i\}_{i=1}^{M}\). 

Therefore, we can find the minimisers of the regularised empirical risk functional by searching over the finite-dimensional space of coefficients $\vec{\alpha} = (\alpha_i)_{i=1}^M \in \R^M$. 
Further, if the regularised empirical risk functional is convex, then the optimal \(\vec{\alpha}\) is unique and we can compute it deterministically. 
In this work, we use the MSE loss function and the regularisation term defined by \(\Omega(f)=\frac{1}{M}\|f\|_{\mathcal{R}_{\mathcal{K}}}^2\) for all \(f\in\mathcal{R}_{\mathcal{K}}\). 
In this case the regularised empirical risk functional is convex (see Appendix \ref{A11}). 
From this, it can be shown (see Appendix \ref{A12}) that the optimal \(\vec{\alpha}\in\R^M\) is given by
\begin{equation}
\label{OptimalAlpha}
\vec{\alpha}=\left(K^2+\lambda K\right)^+K\vec{y},
\end{equation}
where \(\vec{y}=(y_i)_{i=1}^{M}\in\R^M\) is the vector of training data labels, \(K_{ij}=\mathcal{K}(\mathbf{x}_i,\mathbf{x}_j)\) is the  \emph{kernel matrix} for the input training data \(\{\mathbf{x}_i\}_{i=1}^{M}\), and \((\cdot)^+\) is the Moore-Penrose pseudoinverse~\cite{Penrose1955PseudoInverse}.

From \eqref{OptimalAlpha}, we see that determining the optimal \(\vec{\alpha}\in\R^M\) requires the Moore-Penrose pseudoinverse of the \(M\times M\) matrix \(K^2+\lambda K\) to be known. 
The time complexity of calculating the pseudoinverse of a square matrix, which is the most computationally expensive part of this algorithm, is cubic in the number of rows or columns.
This implies that the time complexity of training a model with this algorithm scales as \(\mathcal{O}(M^3)\) with the size of the training data set \(M\in\N\). 
Once the optimal \(\vec{\alpha}\) has been determined, \eqref{RTmodel} shows that the time-complexity of predicting the label for a new input scales as \(\mathcal{O}(M)\). 
Accordingly, it is important for the size of the training dataset to be reasonably small, otherwise the time required for training (and making predictions) can become prohibitively large.
Similarly, we need to be able to compute the kernel \(\mathcal{K}\) efficiently, otherwise both calculating the kernel matrix entries and making predictions could be computationally expensive.

\subsection{Information Spreading and Out-of-Time-Ordered Correlators}
\label{Information Spreading and Out-of-Time-Ordered Correlators}

We now consider the question of how one may quantify the spread of information in quantum many-body systems.
One way to make this concrete is to consider how one could use the natural dynamics of a quantum system to transmit information \cite{Xu2023Tutorial}.
Consider an $n$-qubit quantum system evolving under a time-independent local Hamiltonian $H$.
Let $V_j$, $W_k$ denote 1-qubit Pauli operators $V, W \in \{ X, Y, Z \}$ which act only on qubits $j, k \in \{ 1, \dots, n \}$, respectively.
Suppose Alice has access to qubit $j$ and Bob has access to qubit $k$.
Alice wants to send a classical bit $a \in \{ 0, 1 \}$ to Bob, and does so by either applying $V_j$ at $t=0$ if $a = 1$, or does nothing if $a = 0$.
The system then evolves for some time $t$, after which Bob measures $W_k$ and attempts to determine whether $V_j$ was applied (see Figure~\ref{Fig2}).

If the system begins in the state described by the density operator $\rho$, then we can use the Cauchy-Schwarz inequality to bound the difference between the expectation values that Bob measures in each case \cite{Xu2023Tutorial}. In particular,
\begin{align}
\nonumber
&| \langle V_j W_k(t) V_j \rangle_\rho - \langle W_k(t) \rangle_\rho |^2 \\
\label{eq:cl_comm}
&\qquad\qquad\leq\langle [ V_j, W_k(t) ]^\dagger [ V_j, W_k(t) ] \rangle_\rho,
\end{align}
where $W_k(t) \equiv e^{iHt} W_k e^{-iHt}$ is the operator $W_k$ evolved in the Heisenberg picture under $H$ for $t$ units of time, and $\langle \cdot \rangle_\rho$ denotes the expectation value measured from the state $\rho$.
The right-hand side of \eqref{eq:cl_comm} is a measure of the size of the commutator $[ V_j, W_k(t) ]$ in the state $\rho$.
If $\rho$ is the maximally mixed state \(\mathbb{I}^{\otimes n}/2^n\), where \(\mathbb{I}\) is the \(2\times2\) identity matrix, then the right-hand side is proportional to the square of the Frobenius norm of the operator $[ V_j, W_k(t) ]$.
Thus, we see that the size of this commutator bounds the sensitivity of the expectation value of $W_k(t)$ to the initial perturbation $V_j$.
If the commutator is small, then Alice and Bob would need to repeat the procedure many times in order for Bob to reliably distinguish between the two cases, and hence determine which bit was sent.

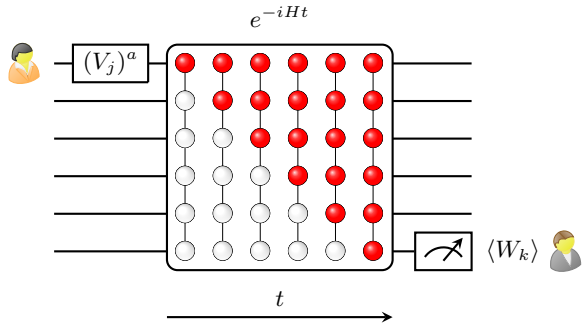
\begin{figure}
\centering
\begin{tikzpicture}[font=\small,overlay]

\node[minimum size=0.4cm,alice,female,skin=yellow!90!black] (A) at (-1.9,2.625) {};

\node[minimum size=0.4cm,bob,skin=yellow!90!black] (B) at (5.3,0.125) {};

\draw[thick, rounded corners] (0,-0.125) -- ++(3,0) -- ++(0,3) -- ++(-3,0) -- cycle ;
\node (Unitary-label) at (1.5,3.25) {\(e^{-i H t}\)};

\draw[thick] (-0.25,5*0.5-0.125) -- ++(0,0.5) -- ++(-1.0,0) -- ++(0,-0.5) -- ++(1.0,0) ;
\node (Unitary-label) at (-0.75,5*0.5+0.125) {$(V_j)^a$};

\draw[thick] (3.3,-0.125) -- ++(0.75,0) -- ++(0,0.5) -- ++(-0.75,0) -- cycle;
\draw[thick] (3.62925-0.2,0.05) arc (145:35:0.3) ;
\draw[thick,-stealth] (3.5+0.75*0.5-0.2,0) -- ++({0.4*cos(50)},{0.4*sin(50)});
\node[] (A) at (4.6,0.125) {$\langle W_k \rangle$};

\draw[thick] (0,0*0.5+0.125) -- ++(-1.5,0) ;
\draw[thick] (0,1*0.5+0.125) -- ++(-1.5,0) ;
\draw[thick] (0,2*0.5+0.125) -- ++(-1.5,0) ;
\draw[thick] (0,3*0.5+0.125) -- ++(-1.5,0) ;
\draw[thick] (0,4*0.5+0.125) -- ++(-1.5,0) ;
\draw[thick] (0,5*0.5+0.125) -- ++(-0.25,0) ;
\draw[thick] (-1.25,5*0.5+0.125) -- ++(-0.25,0) ;

\draw[thick] (3,0*0.5+0.125) -- ++(0.29,0) ;
\draw[thick] (3,1*0.5+0.125) -- ++(1.05,0) ;
\draw[thick] (3,2*0.5+0.125) -- ++(1.05,0) ;
\draw[thick] (3,3*0.5+0.125) -- ++(1.05,0) ;
\draw[thick] (3,4*0.5+0.125) -- ++(1.05,0) ;
\draw[thick] (3,5*0.5+0.125) -- ++(1.05,0) ;

\draw[thick,-stealth] (0,-0.75) -- ++(3,0) ;
\node at (1.5,-0.5) {\(t\)};

\end{tikzpicture}
\begin{tikzpicture}

\draw[line width=0.05pt] (-2,0) -- ++(0,2.5) ;
\draw[line width=0.05pt] (-1.5,0) -- ++(0,2.5) ;
\draw[line width=0.05pt] (-1,0) -- ++(0,2.5) ;
\draw[line width=0.05pt] (-0.5,0) -- ++(0,2.5) ;
\draw[line width=0.05pt] (0,0) -- ++(0,2.5) ;
\draw[line width=0.05pt] (0.5,0) -- ++(0,2.5) ;

\begin{scope}[blend group=screen]
\draw[top color=gray,bottom color=white] (-2,0) circle (0.125cm) ;
\draw[shading=ball,ball color=white]      (-2,0) circle (0.125cm) ;
\end{scope}
\begin{scope}[blend group=screen]
\draw[top color=gray,bottom color=white] (-1.5,0) circle (0.125cm) ;
\draw[shading=ball,ball color=white]      (-1.5,0) circle (0.125cm) ;
\end{scope}
\begin{scope}[blend group=screen]
\draw[top color=gray,bottom color=white] (-1,0) circle (0.125cm) ;
\draw[shading=ball,ball color=white]      (-1,0) circle (0.125cm) ;
\end{scope}
\begin{scope}[blend group=screen]
\draw[top color=gray,bottom color=white] (-0.5,0) circle (0.125cm) ;
\draw[shading=ball,ball color=white]      (-0.5,0) circle (0.125cm) ;
\end{scope}
\begin{scope}[blend group=screen]
\draw[top color=gray,bottom color=white] (0,0) circle (0.125cm) ;
\draw[shading=ball,ball color=white]      (0,0) circle (0.125cm) ;
\end{scope}
\begin{scope}[blend group=screen]
\draw[top color=red,bottom color=red] (0.5,0) circle (0.125cm) ;
\draw[shading=ball,ball color=red]      (0.5,0) circle (0.125cm) ;
\end{scope}

\begin{scope}[blend group=screen]
\draw[top color=gray,bottom color=white] (-2,0.5) circle (0.125cm) ;
\draw[shading=ball,ball color=white]      (-2,0.5) circle (0.125cm) ;
\end{scope}
\begin{scope}[blend group=screen]
\draw[top color=gray,bottom color=white] (-1.5,0.5) circle (0.125cm) ;
\draw[shading=ball,ball color=white]      (-1.5,0.5) circle (0.125cm) ;
\end{scope}
\begin{scope}[blend group=screen]
\draw[top color=gray,bottom color=white] (-1,0.5) circle (0.125cm) ;
\draw[shading=ball,ball color=white]      (-1,0.5) circle (0.125cm) ;
\end{scope}
\begin{scope}[blend group=screen]
\draw[top color=gray,bottom color=white] (-0.5,0.5) circle (0.125cm) ;
\draw[shading=ball,ball color=white]      (-0.5,0.5) circle (0.125cm) ;
\end{scope}
\begin{scope}[blend group=screen]
\draw[top color=red,bottom color=red] (0,0.5) circle (0.125cm) ;
\draw[shading=ball,ball color=red]      (0,0.5) circle (0.125cm) ;
\end{scope}
\begin{scope}[blend group=screen]
\draw[top color=red,bottom color=red] (0.5,0.5) circle (0.125cm) ;
\draw[shading=ball,ball color=red]      (0.5,0.5) circle (0.125cm) ;
\end{scope}

\begin{scope}[blend group=screen]
\draw[top color=gray,bottom color=white] (-2,1) circle (0.125cm) ;
\draw[shading=ball,ball color=white]      (-2,1) circle (0.125cm) ;
\end{scope}
\begin{scope}[blend group=screen]
\draw[top color=gray,bottom color=white] (-1.5,1) circle (0.125cm) ;
\draw[shading=ball,ball color=white]      (-1.5,1) circle (0.125cm) ;
\end{scope}
\begin{scope}[blend group=screen]
\draw[top color=gray,bottom color=white] (-1,1) circle (0.125cm) ;
\draw[shading=ball,ball color=white]      (-1,1) circle (0.125cm) ;
\end{scope}
\begin{scope}[blend group=screen]
\draw[top color=red,bottom color=red] (-0.5,1) circle (0.125cm) ;
\draw[shading=ball,ball color=red]      (-0.5,1) circle (0.125cm) ;
\end{scope}
\begin{scope}[blend group=screen]
\draw[top color=red,bottom color=red] (0,1) circle (0.125cm) ;
\draw[shading=ball,ball color=red]      (0,1) circle (0.125cm) ;
\end{scope}
\begin{scope}[blend group=screen]
\draw[top color=red,bottom color=red] (0.5,1) circle (0.125cm) ;
\draw[shading=ball,ball color=red]      (0.5,1) circle (0.125cm) ;
\end{scope}

\begin{scope}[blend group=screen]
\draw[top color=gray,bottom color=white] (-2,1.5) circle (0.125cm) ;
\draw[shading=ball,ball color=white]      (-2,1.5) circle (0.125cm) ;
\end{scope}
\begin{scope}[blend group=screen]
\draw[top color=gray,bottom color=white] (-1.5,1.5) circle (0.125cm) ;
\draw[shading=ball,ball color=white]      (-1.5,1.5) circle (0.125cm) ;
\end{scope}
\begin{scope}[blend group=screen]
\draw[top color=red,bottom color=red] (-1,1.5) circle (0.125cm) ;
\draw[shading=ball,ball color=red]      (-1,1.5) circle (0.125cm) ;
\end{scope}
\begin{scope}[blend group=screen]
\draw[top color=red,bottom color=red] (-0.5,1.5) circle (0.125cm) ;
\draw[shading=ball,ball color=red]      (-0.5,1.5) circle (0.125cm) ;
\end{scope}
\begin{scope}[blend group=screen]
\draw[top color=red,bottom color=red] (0,1.5) circle (0.125cm) ;
\draw[shading=ball,ball color=red]      (0,1.5) circle (0.125cm) ;
\end{scope}
\begin{scope}[blend group=screen]
\draw[top color=red,bottom color=red] (0.5,1.5) circle (0.125cm) ;
\draw[shading=ball,ball color=red]      (0.5,1.5) circle (0.125cm) ;
\end{scope}

\begin{scope}[blend group=screen]
\draw[top color=gray,bottom color=white] (-2,2) circle (0.125cm) ;
\draw[shading=ball,ball color=white]      (-2,2) circle (0.125cm) ;
\end{scope}
\begin{scope}[blend group=screen]
\draw[top color=red,bottom color=red] (-1.5,2) circle (0.125cm) ;
\draw[shading=ball,ball color=red]      (-1.5,2) circle (0.125cm) ;
\end{scope}
\begin{scope}[blend group=screen]
\draw[top color=red,bottom color=red] (-1,2) circle (0.125cm) ;
\draw[shading=ball,ball color=red]      (-1,2) circle (0.125cm) ;
\end{scope}
\begin{scope}[blend group=screen]
\draw[top color=red,bottom color=red] (-0.5,2) circle (0.125cm) ;
\draw[shading=ball,ball color=red]      (-0.5,2) circle (0.125cm) ;
\end{scope}
\begin{scope}[blend group=screen]
\draw[top color=red,bottom color=red] (0,2) circle (0.125cm) ;
\draw[shading=ball,ball color=red]      (0,2) circle (0.125cm) ;
\end{scope}
\begin{scope}[blend group=screen]
\draw[top color=red,bottom color=red] (0.5,2) circle (0.125cm) ;
\draw[shading=ball,ball color=red]      (0.5,2) circle (0.125cm) ;
\end{scope}

\begin{scope}[blend group=screen]
\draw[top color=red,bottom color=red] (-2,2.5) circle (0.125cm) ;
\draw[shading=ball,ball color=red]      (-2,2.5) circle (0.125cm) ;
\end{scope}
\begin{scope}[blend group=screen]
\draw[top color=red,bottom color=red] (-1.5,2.5) circle (0.125cm) ;
\draw[shading=ball,ball color=red]      (-1.5,2.5) circle (0.125cm) ;
\end{scope}
\begin{scope}[blend group=screen]
\draw[top color=red,bottom color=red] (-1,2.5) circle (0.125cm) ;
\draw[shading=ball,ball color=red]      (-1,2.5) circle (0.125cm) ;
\end{scope}
\begin{scope}[blend group=screen]
\draw[top color=red,bottom color=red] (-0.5,2.5) circle (0.125cm) ;
\draw[shading=ball,ball color=red]      (-0.5,2.5) circle (0.125cm) ;
\end{scope}
\begin{scope}[blend group=screen]
\draw[top color=red,bottom color=red] (0,2.5) circle (0.125cm) ;
\draw[shading=ball,ball color=red]      (0,2.5) circle (0.125cm) ;
\end{scope}
\begin{scope}[blend group=screen]
\draw[top color=red,bottom color=red] (0.5,2.5) circle (0.125cm) ;
\draw[shading=ball,ball color=red]      (0.5,2.5) circle (0.125cm) ;
\end{scope}

\end{tikzpicture}
\vskip1cm
\caption{The Alice-Bob classical communication protocol in the case where the system is a 1D spin chain with open boundary conditions, evolving under a time-independent local Hamiltonian \(H\).
Alice and Bob have access to the individual qubits at opposite ends of the chain. 
Alice wants to send a bit $a \in \{0,1\}$ to Bob by applying $V_j$ if $a=1$, or doing nothing if $a=0$.
The system then evolves for a time \(t\), during which the influence of Alice's operation propagates through the system.
Bob then performs a measurement of \(W_k\) to try and determine the value of the bit \(a\).}
\label{Fig2}
\end{figure}

One can gain further insight into how information spreads in the system by examining the Heisenberg evolution of local observables.
Specifically, the Heisenberg time evolution of $W_k$ can be expanded as
\begin{align}
\label{BCHHeisenberg}
W_k(t)
=& W_k+\sum_{l=1}^{\infty}\frac{(it)^l}{l!}\underbrace{[H,\ldots,[H,[H}_{l\textrm{ times}},W_k]]\ldots].
\end{align}
If the Hamiltonian \(H\) only contains local (e.g., nearest-neighbour) interactions, then the number of sites upon which the nested commutators act non-trivially generally increases with the sum index $l$ in \eqref{BCHHeisenberg}.
The operator $W_k(t)$ is thus highly non-local in general.
However, for small $|t|$, the terms in the sum with large support (i.e., large \(l\)) are insignificant and only become significant as $|t|$ increases.
This gradual increase in the magnitude of the non-local terms describes how $W_k(t)$ spreads to become increasingly non-local over time.
The extent to which $W_k(t)$ fails to commute with an operator $V_j$ then depends on the magnitude of the terms in \eqref{BCHHeisenberg} which do not simply contain the identity operator on qubit $j$ (i.e., those with support on qubit $j$).

These considerations suggest quantifying the spread of information in terms of the size of commutators between local operators in the Heisenberg picture.
To this end, we define the \emph{OTOC} as follows.
Firstly, the \emph{squared-commutator OTOC} (SC-OTOC)~\cite{Swingle2018Unscrambling,Harrow2021Separation}, denoted \(C_{jk}(t)\), is the non-negative real number defined such that
\begin{equation}
\label{SqComOTOC}
C_{jk}(t) \equiv \frac{1}{2}\big\langle[V_j,W_k(t)]^\dagger[V_j,W_k(t)]\big\rangle_{\mathbb{I}^{\otimes n}\slash2^n}.
\end{equation}
Since $V_j$ and $W_k$ are both Hermitian and unitary, one can expand the commutators in \eqref{SqComOTOC} to write \(C_{jk}(t)=1-F_{jk}(t)\), where \(F_{jk}(t)\) is the real number given by
\begin{equation}
\label{OTOC}
F_{jk}(t) \equiv \big\langle V_j W_k(t) V_j W_k(t)\big\rangle_{\mathbb{I}^{\otimes n}\slash2^n}.
\end{equation}
The latter quantity, $F_{jk}(t)$, will be referred to as simply the \emph{OTOC}~\cite{Maldacena2016Bound,Hashimoto2017OTOC,Xu2023Tutorial}.

The OTOC and SC-OTOC are sometimes defined via an expectation value measured from an arbitrary state.
However, throughout this paper we will focus on the case of the maximally mixed state, which ensures that $F_{jk}(t)$ is real.
This choice is well-motivated in the literature on scrambling, since it provides a simple setting where the dynamics of information spreading can be studied in isolation from any pre-existing correlations in the state \cite{Xu2023Tutorial,Hayden2007Black}.
The absence of correlations in $\rho = \mathbb{I}^{\otimes n}\slash2^n$ also greatly simplifies the numerical simulations described in Section~\ref{Generating the data}.
Further, since the case of the maximally mixed state is related to the Frobenius norm of the commutator, it can be used to derive a bound which holds for any state $\rho$. 
Specifically, using H\"older's inequality, we can show that
\begin{equation}
\label{BoundForArbitraryState}
  \frac{1}{2}\big\langle[V_j,W_k(t)]^\dagger[V_j,W_k(t)]\big\rangle_{\rho} \leq 2^n \| \rho \|_{\infty}\, C_{jk}(t),
\end{equation}
where $C_{jk}(t)$ is the quantity from \eqref{SqComOTOC} defined with respect to the maximally mixed state, and $\| \rho \|_{\infty}$ is the operator norm of $\rho$ (i.e., the largest eigenvalue of $\rho$).
Note that if $\rho$ is pure, then $\| \rho \|_{\infty}= 1$.
In general, it can also be bounded by the purity $\| \rho \|_{\infty} \leq \sqrt{\text{tr} (\rho^2)}$.
Thus we see, from \eqref{BoundForArbitraryState}, that the choice of the state $\mathbb{I}^{\otimes n}\slash2^n$ in \eqref{SqComOTOC} captures general properties of the information spreading, since then $C_{jk}(t)$ can be used to infer bounds for any other state.

Associated with the OTOC and the SC-OTOC are three distinct time regimes---short, intermediate, and long---each characterised by different time-dependent behaviours. 
In the short-time regime, the SC-OTOC grows, with the growth rate depending on the system's characteristics. 
For chaotic systems, it is expected to grow exponentially. 
The intermediate-time regime marks a transition from growth to decay, occurring around the so-called \emph{scrambling time} of the system. 
Finally, in the long-time regime, the SC-OTOC decays to a constant value, possibly with some superposed oscillations. 
All three time regimes have been studied for various reasons in the literature (see Section 2 of~\cite{GarciaMata2023Chaos} for an overview and references). 
Accordingly, methods for calculating OTOCs in any of the three regimes may be of practical use, but we first need to know where the boundaries of each regime lies. 
For qubit models, such as those considered in this work, the scrambling time scales at most linearly in the size of the system, depending on the locality of the Hamiltonian (see the fourth paragraph of the Introduction in~\cite{Maldacena2016Bound}).
As a result, we know that the short-time regime and part of the intermediate-time regime (which contains the scrambling time) must occur before a time which is linear in the system size.
We will make use of this in Section~\ref{The parameterised sets of Hamiltonians}.

Equation \eqref{eq:cl_comm} illustrates how the OTOC bounds the transmission of classical information in a quantum system, but the OTOC can also be used to study the transmission of quantum information.
This is often investigated in the setup of the Hayden-Preskill protocol, which is a toy model for studying the recovery of information from black holes~\cite{Hayden2007Black}.
Similar to above, we consider an isolated $n$-qubit system, which we denote by $\mathcal{S}$.
Alice has access to a qubit $q_A$ in \(\mathcal{S}\), and Bob has access to a qubit $q_B$, also in \(\mathcal{S}\).
Instead of applying an operation to $q_A$ as before, Alice will prepare $q_A$ in a state which is maximally entangled with another reference system $R_1$,
\begin{equation}
  \ket{\psi}_{q_A \cup R_1} = \tfrac{1}{\sqrt{2}}( \ket{00} + \ket{11}).
\end{equation}
The remainder of the system, \(\mathcal{S}\setminus q_A\), is prepared in the maximally mixed state, $\rho_{\mathcal{S} \setminus q_A} = \mathbb{I}^{\otimes n-1} / 2^{n-1}$.
Since $\mathcal{S}$ is a closed system, the reference system $R_1$ remains maximally entangled with $\mathcal{S}$ as the system evolves in time.
Initially, $R_1$ is only entangled with the subsystem $q_A$, but as $q_A$ interacts with the rest of the system, the entanglement with $R_1$ generally spreads into the non-local degrees of freedom of $\mathcal{S}$.
One can track the spreading of entanglement by examining to what extent Bob is able to recover the entanglement with $R_1$ from another qubit, such as $q_B$, in \(\mathcal{S}\).

In the Hayden-Preskill protocol, Bob is also given access to a reference system introduced to purify the subsystem $\mathcal{S} \setminus q_A$~\cite{Hayden2007Black}.
Since $\mathcal{S} \setminus q_A$ begins in a maximally mixed state, it can be purified with another $(n-1)$-qubit reference system $R_2$, with which $\mathcal{S} \setminus q_A$ forms $n-1$ Bell pairs,
\begin{equation}
  \ket{\psi}_{\left(\mathcal{S}\setminus q_A\right)\cup R_2 } = \tfrac{1}{2^{(n-1)/2}} ( \ket{00} + \ket{11} )^{\otimes (n-1)}.
\end{equation}
Figure~\ref{Fig3} shows a tensor network representation of the various systems involved, where \(\scalebox{1.4}{\(\circ\)}\) indicates that the subsystems form a collection of Bell pairs.
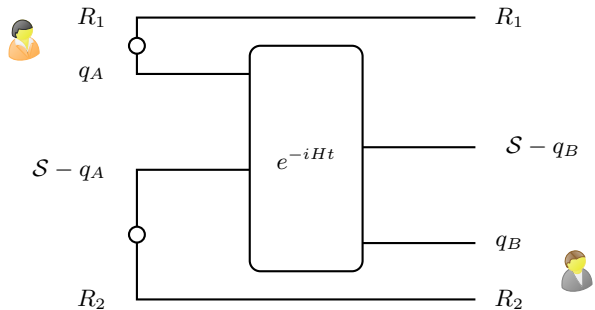
\begin{figure}
\centering
\begin{tikzpicture}[scale=1.5]
\draw[thick, rounded corners] (0,0) -- ++(1,0) -- ++(0,2) -- ++(-1,0) -- cycle;
\node at (0.5,1) {\(e^{-i H t}\)};
\draw[thick] (0,1.75) -- ++(-1,0) -- ++(0,0.5) -- ++(3,0);
\draw[color=black,thick, fill=white] (-1,2) circle (2pt);
\node at (-1.4,2.25) {\(R_1\)};
\node at (2.3,2.25) {\(R_1\)};
\node at (-1.4,1.75) {\(q_A\)};
\draw[thick] (0,0.9) -- ++(-1,0) -- ++(0,-0.9-0.25) -- ++(3,0);
\draw[color=black,thick, fill=white] (-1,0.325) circle (2pt);
\node at (-1.6,0.9) {\(\mathcal{S}\setminus q_A\)};
\node at (-1.4,-0.25) {\(R_2\)};
\node at (2.3,-0.25) {\(R_2\)};
\draw[thick] (1,0.25) -- ++(1,0);
\node at (2.3,0.25) {\(q_B\)};
\draw[thick] (1,1.1) -- ++(1,0);
\node at (2.6,1.1) {\(\mathcal{S}\setminus q_B\)};
\node[minimum size=0.4cm,alice,female,skin=yellow!90!black] (A) at (-2.0,2.05) {};
\node[minimum size=0.4cm,bob,skin=yellow!90!black] (B) at (2.9,0) {};
\end{tikzpicture}
\caption{Tensor network representation of the systems involved in the Alice-Bob quantum communication protocol. Initially, subsystems \(R_1\) and \(q_A\) form a Bell state, while subsystems \(\mathcal{S}\setminus q_A\) and \(R_2\) form \(n-1\) Bell states. The unitary operator \(e^{-i H t}\) is then applied to the system \(\mathcal{S}\). Bob's aim is to detect and quantify the entanglement with $R_1$ that has spread from $q_A$ to $q_B$.}
\label{Fig3}
\end{figure}
The amount of information that Bob has about $R_1$ can then be quantified by the mutual information,
\begin{equation}
\label{MutualInformation}
I(R_1:q_B \cup R_2) \equiv S(R_1)+S(q_B \cup R_2)-S(R_1 \cup q_B \cup R_2),
\end{equation}
where $S(A)\equiv -\Tr ( \rho_A \log_2 \rho_A )$ is the von Neumann entropy of the reduced state $\rho_A$ of a general subsystem $A$.
The mutual information $I(R_1:q_B \cup R_2)$ is importantly also related to the coherent information $I( R_1 \rangle q_B \cup R_2)$ of the induced channel from $q_A$ to $q_B \cup R_2$, which in turn is related to the quantum capacity of the channel~\cite{Wilde2013Quantum,Nielsen2000Quantum}.
Indeed, from a simple application of the definition of $I( R_1 \rangle q_B \cup R_2)$ it is easy to verify that
\begin{align}
I( R_1 \rangle q_B \cup R_2) 
&= I(R_1:q_B \cup R_2) - 1,
\end{align}
where we have used $S(R_1) = 1$. 

How can the OTOC be used to understand the spreading of entanglement?
To this end, it can be shown (see Section IV.A in~\cite{Xu2023Tutorial}) that the OTOC provides a lower bound on the mutual information,
\begin{align}
\nonumber
&I(R_1:q_B\cup R_2) \\
\label{MutualInformationAndOTOC}
&\geq 4 - \log_2 \left( 7 + \sum_{\substack{V,W\\\in\{X,Y,Z\}}} \frac{1}{2^n} \Tr\left(\left(V_{q_B} W_{q_A}(-t)\right)^2\right)\right).
\end{align}
The sum appearing in the right-hand side of \eqref{MutualInformationAndOTOC} is a sum of OTOCs for different choices of $V,W\in\{X,Y,Z\}$ acting on qubits $q_B$ and $q_A$, respectively.
From this inequality, we can clearly see the role that the OTOC plays in detecting how much information about $R_1$ has propagated from $q_A$ to $q_B$.
For instance, at $t=0$ the operators \(W_{q_A}(0)\) and \(V_{q_B}\) commute, so each term in the sum above is equal to $1$, yielding a trivial lower bound of $I(R_1:q_B\cup R_2) \geq 0$.
However, as time progresses, the support of $W_{q_A}(-t)$ will generally grow to include qubit $q_B$ so that \(W_{q_A}(-t)\) and \(V_{q_B}\) no longer commute.
When this happens, the OTOC values begin to decay from 1, causing the right-hand side of \eqref{MutualInformationAndOTOC} to be strictly greater than $0$.
This allows Bob to infer that at least some information about $R_1$ has reached $q_B$, and gives him a lower bound on how much information he possesses.

Beyond simply determining how much information has reached $q_B$, in~\cite{Yoshida2017Efficient} an efficient decoding procedure was proposed for recovering the entanglement with $R_1$, by applying a unitary operation on $q_B \cup R_2$ to reconstruct a Bell state on $R_1 \cup q_B$.
The fidelity of the decoding protocol is closely related to the sum of OTOCs appearing in \eqref{MutualInformationAndOTOC}.
Because of the importance of this sum in determining the spreading of quantum information through a quantum many-body system, in the following sections we will investigate how to approximate its value in different systems using kernel machines.

\section{Methods}
\label{Section3}

The manner in which information propagates through the system \(\mathcal{S}\) depends on the details of its dynamics, which is specified by the time-independent Hamiltonian \(H\). 
This means that in order to compute OTOCs and related quantities, such as the sum of OTOCs in \eqref{MutualInformationAndOTOC}, generally we need to simulate the evolution of the system \(\mathcal{S}\).
This amounts to calculating \(e^{-i H t}\). 
Such a calculation becomes computationally intensive for a generic \(H\) as the size of the system increases, and needs to be performed for every \(H\) and value of time \(t\) that one wishes to investigate. 
In spite of this, is there some property of many-body Hamiltonians that can be used to determine OTOCs without the need to simulate the full dynamics? 
Here we investigate whether such properties may be captured by a ML model which makes use of KMs.
In this section we describe the methods used to conduct the investigation. 
This includes a description of the parameterised sets of Hamiltonians we consider, how the task is framed as a learning problem, the specific kernels which are trained and applied to the data, and the techniques used to generate the data.

\subsection{The parameterised sets of Hamiltonians}
\label{The parameterised sets of Hamiltonians}

The learning problem addressed in this work requires us to choose a parameterised set of Hamiltonians \(\{H(x):x\in\R^d\}\) to investigate. 
This determines both the set of possible dynamics for $\mathcal{S}$, and the explicit form of the OTOCs considered in the learning problem. 
With this in mind, we consider four distinct choices of the parameterised set of \(n\)-qubit Hamiltonians. Each of the sets forms a 3-dimensional (i.e., \(d=3\)) subspace of Hermitian operators describing a collection of 1D quantum systems. 
The parameterised sets are chosen to contain systems which are widely considered in condensed matter physics.

The first set of Hamiltonians, denoted \(\mathscr{H}_1=\{H_1(x):x\in\R^3\}\), has elements \(H_1(x)\) defined by
\begin{align}
\nonumber
H_{1}(x) &= x_1\left(\sum_{i=1}^{n}X_i\right)+x_2\left(\sum_{j=1}^{n-1}X_jX_{j+1}\right)\\
\label{HaldaneHamiltonians}
&\qquad+x_3\left(\sum_{k=1}^{n-2}Z_kX_{k+1}Z_{k+2}\right)
\end{align}
for all \(x = (x_1,x_2,x_3) \in\R^3\). 
The ground states of the Hamiltonians in \(\mathscr{H}_1\) exhibit a \(\mathbb{Z}_2\times\mathbb{Z}_2\) symmetry-protected topological phase that is considered in the machine learning articles~\cite{Cong2019Convolutional,Wu2023Phase}, and motivates our consideration of the set here.

The second set, denoted \(\mathscr{H}_2=\{H_2(x):x\in\R^3\}\), is often called the XYZ Heisenberg model \cite{Takhtadzhan1979Heisenberg,Pinheiro2013Heisenberg,Rota2018Heisenberg,Mohamed2021Heisenberg} and has elements \(H_2(x)\) defined by
\begin{align}
\nonumber
H_2(x) &= x_1\left(\sum_{i=1}^{n-1}X_iX_{i+1}\right)+x_2\left(\sum_{j=1}^{n-1}Y_jY_{j+1}\right)\\
\label{HeisenbergHamiltonians}
&\qquad\qquad+x_3\left(\sum_{k=1}^{n-1}Z_kZ_{k+1}\right)
\end{align}
for all \(x = (x_1,x_2,x_3) \in\R^3\). The Hamiltonians in \(\mathscr{H}_2\) serve as quantum models of ferromagnetic and anti-ferromagnetic materials. 
Specifically, they describe 1D spin-chains with nearest-neighbour spin couplings.
In the case $x_1, x_2, x_3 > 0$, this model captures the tendancy of neighbouring atoms in anti-ferromagnetic materials to align their atomic magnetic moments in an anti-parallel fashion when occupying low-energy states.
For $x_1, x_2, x_3 < 0$, the spins tend to align, which models ferromagnetic materials.

The third set, denoted \(\mathscr{H}_3=\{H_3(x):x\in\R^3\}\), is a generalised version of the Majumdar-Ghosh model \cite{Majumdar1969On1,Majumdar1969On2,Caspers1984Majumdar,Chhajlany2007Majumdar} and has elements \(H_3(x)\) defined by
\begin{align}
\nonumber
H_3(x)=&x_1\left(\sum_{i=1}^{n-1}X_iX_{i+1}+0.5\sum_{j=1}^{n-1}X_jX_{j+2}\right)\\
\nonumber
+&x_2\left(\sum_{i=1}^{n-1}Y_iY_{i+1}+0.5\sum_{j=1}^{n-1}Y_jY_{j+2}\right)\\
\label{MajumdarGhoshHamiltonians}
+&x_3\left(\sum_{i=1}^{n-1}Z_iZ_{i+1}+0.5\sum_{j=1}^{n-1}Z_jZ_{j+2}\right)
\end{align}
for all \(x = (x_1,x_2,x_3) \in\R^3\). Similar to \(\mathscr{H}_2\), the Hamiltonians in \(\mathscr{H}_3\) provide quantum models of ferromagnetic and anti-ferromagnetic materials. 
They also describe 1D spin-chains with nearest-neighbour spin couplings, in addition to next-nearest neighbour spin-couplings with half the coupling strength. 
This is a slightly more realistic model of anti-ferromagnetic and ferromagnetic materials which does not assume that atoms only interact with their nearest-neighbours.

The fourth and final set, denoted \(\mathscr{H}_4=\{H_4(x):x\in\R^3\}\), is often called the mixed-field Ising chain \cite{Weng1996MixedField,Wurtz2020MixedField,Chiba2024MixedField} and has elements defined by
\begin{align}
\nonumber
H_4(x)&=x_1\left(\sum_{i=1}^{n}X_i\right)+x_2\left(\sum_{i=1}^{n}Z_i\right)\\
\label{MixedFieldIsingHamiltonians}
&\qquad+x_3\left(\sum_{k=1}^{n-1}Z_kZ_{k+1}\right)
\end{align}
for all \(x = (x_1,x_2,x_3) \in\R^3\). 
The Hamiltonians in this set describe 1D spin-chains with nearest-neighbour \(Z\)-component spin couplings under the influence of an external homogenous magnetic field with non-zero components in the physical \(x\) and \(z\) directions.

The parameterised Hamiltonians we consider all satisfy the property $H(tx) = t H(x)$ for any $t \in \mathbb{R}$, where \(H\) denotes one of the parameterised Hamiltonians defined in equations \eqref{HaldaneHamiltonians}--\eqref{MixedFieldIsingHamiltonians}. 
Using this property, but replacing \(t\mapsto\|x\|\) and \(x\mapsto\hat{x}\) (where \(\hat{x}=x\slash\|x\|\) is the unit vector parallel to \(x\)), we see that \(H(x)=\|x\|H(\hat{x})\), which further implies that 
\begin{align}
\label{NormOfParameterVectorIsTime}
e^{-iH(x)}=e^{-i\|x\|H(\hat{x})}.
\end{align}  
Equation \eqref{NormOfParameterVectorIsTime} shows how the operator \(e^{-iH(x)}\) is equivalent to the time-evolution operator describing evolution under the Hamiltonian \(H(\hat{x})\) for a time \(t=\|x\|\).
This holds for all parameterised Hamiltonians defined in \eqref{HaldaneHamiltonians}--\eqref{MixedFieldIsingHamiltonians}, and means that we can capture changes in time by scaling the parameter vector $x$, which changes $\|x\|$ but leaves $H\left(\hat{x}\right)$ fixed (up to a sign if we scale by a negative constant). 
Hence, choosing a range of evolution times for the associated quantum systems corresponds to choosing the parameter vectors with norms lying in the same range.

Accordingly, it is most natural to sample the input data (i.e., the parameter vectors \(x\)) from a ball in \(\R^3\) centered at the origin, whose radius determines the maximum evolution time.
The goal is then to choose an appropriate radius for each of the balls which is both feasible to simulate and captures an interesting range of evolution times. 
Specifically, we chose the radius of the ball from which the input data is sampled to scale linearly with the size of the underlying quantum system. 
As discussed in Section \ref{Information Spreading and Out-of-Time-Ordered Correlators}, the scrambling time for qubit models scales linearly in the size of the system, so we expect that a linear scaling (in the system size) of the maximum evolution time will be sufficient to capture both the short-time regime and possibly part of the intermediate-time regime.

For \(\mathscr{H}_1\), \(\mathscr{H}_2\) and \(\mathscr{H}_3\), we chose to uniformly sample input data from the ball of radius \(n\) centered at the origin. 
This means that the effective evolution time \(t\) of the corresponding \(n\)-qubit quantum systems is bounded above by the system size \(t\leq t_{max}=n\). 
For \(\mathscr{H}_4\), we chose to uniformly sample input data from the ball of radius \(2n\) centered at the origin, meaning that \(t\leq t_{max}=2n\).
Increasing the radius of the ball for \(\mathscr{H}_4\) was motivated primarily by the observation that sampling from the ball of radius \(n\) did not produce much variation in the value of the OTOCs. 
Informally, one possible explanation for this is that the only term in \eqref{MixedFieldIsingHamiltonians} which contributes to the spread of information is the term multiplied by \(x_3\), and this term only transfers information about \(Z\) to neighbouring qubits.
In contrast, the terms in the Hamiltonians in \(\mathscr{H}_1\), \(\mathscr{H}_2\) and \(\mathscr{H}_3\) that contribute to the spread of information collectively transfer at least \(Z\) and \(X\) information. 
Accordingly, it seems reasonable that the systems in \(\mathscr{H}_4\) require a greater maximum evolution time to observe changes in the correlations captured by the OTOC.

\subsection{The learning problem}
\label{The learning problem}

In order to formulate the problem as a learning task, we consider one of the parameterised sets of Hamiltonians \(\mathscr{H}=\{H(x):x\in\R^d\}\), and set up ML models to learn two distinct functions of the Hamiltonian parameters $x \in \mathbb{R}^d$.
The first function, denoted \(\mathscr{O}_{XZ}:\R^d\to\R\), captures the XZ-OTOC when \(j=n\) and \(k=1\) (i.e., $W=Z$ on the first qubit and $V=X$ on the last qubit in the 1D spin-chain). Specifically, \(\mathscr{O}_{XZ}\) is defined such that
\begin{equation}
\label{XZOTOC}
\mathscr{O}_{XZ}(x)=\frac{1}{2^n}\Tr\left(X_n\mathcal{Z}_1(x)X_n\mathcal{Z}_1(x)\right)
\end{equation}
for all \(x\in\R^d\), where we denote by \(\mathcal{Z}_1(x)\equiv e^{iH(x)}Z_1e^{-iH(x)}\) the operator \(Z_1\) evolved in the Heisenberg picture under \(H(\hat{x})\) for \(\|x\|\) units of time (see Section \ref{The parameterised sets of Hamiltonians}). The second function, denoted \(\mathscr{O}_{Sum}:\R^d\to\R\), captures the sum of OTOCs in \eqref{MutualInformationAndOTOC} when \(q_A\) and \(q_B\) lie at opposite ends of the 1D spin chain, and is given by
\begin{equation}
\label{SumOfOTOCs}
\mathscr{O}_{Sum}(x)=\sum_{V,W\in\{X,Y,Z\}}\frac{1}{2^n}\Tr\left(V_n\mathcal{W}_1(x)V_n\mathcal{W}_1(x)\right),
\end{equation}
where we denote \(\mathcal{W}_1(x)\equiv e^{-iH(x)}W_1e^{iH(x)}\). 
Note that the signs of the exponents in the definition of \(\mathcal{W}_1(x)\) have been reversed to account for the minus sign in the arguments of the Heisenberg operators in \eqref{MutualInformationAndOTOC}.

The learning task is then to find a model which accurately learns \(\mathscr{O}_{XZ}\) or \(\mathscr{O}_{Sum}\) from sampled data (we describe our procedure for generating data in Section~\ref{Generating the data}).
In total, we consider 64 instances of the learning task: one for each choice of \(\mathscr{O}_{XZ}\) or \(\mathscr{O}_{Sum}\), with one of the four parameterised sets of Hamiltonians \(\mathscr{H}\) from Section~\ref{The parameterised sets of Hamiltonians}, and a system size of $n\in\{5, 10, 15, 20, 25, 30, 35, 40\}$ qubits.
For each of the tasks, we trial six different kernel functions (see Section~\ref{The kernels}) and compare their performance by determining how well the models perform on unseen test data.
Thus, the generated data is split into training and testing data, 80\% for training and 20\% for testing. 
Next we perform a 10-fold cross-validation (see Section 4.5 of~\cite{Mohri2018Foundations}) on the training data to find suitable hyperparameter values for the various kernels using a grid search.
The hyperparameter values which we trial are listed in Table \ref{HyperparameterValues} of Appendix \ref{Tables of hyperparameter values}. 
Of the trialled values, the best value of each hyperparameter is selected based on the average coefficient of determination (\(R^2\)) over all ten folds of the training data. These hyperparameter values are then used to train ML models on the complete training dataset, which are subsequently applied to the testing data. The best trialled hyperparameter values for every instance of the learning task can be found in Tables \ref{BestHyperparametersH1}, \ref{BestHyperparametersH2}, \ref{BestHyperparametersH3}, and \ref{BestHyperparametersH4} of Appendix \ref{Tables of hyperparameter values}.

\subsection{The kernels}
\label{The kernels}

We now provide definitions of the six kernels which are used to train ML models and make predictions on the testing sets for each problem instance. 
Some of the kernels depend on hyperparameters \(\gamma\in\R\), \(c_0\in\R\) and \(d\in\N\) which are tuned via a 10-fold cross-validation on the training data (see Section \ref{The learning problem}).

The first kernel is the \emph{linear kernel}, denoted \(\mathcal{K}_{\textrm{lin}}:\mathcal{X}\times\mathcal{X}\to\R\), which is defined by the standard Euclidean inner product of its arguments,
\begin{equation}
\label{LinearKernel}
\mathcal{K}_{\textrm{lin}}(x,x^\prime) = \langle x,x^\prime\rangle.
\end{equation}
The second is the \emph{polynomial kernel}, denoted \(\mathcal{K}_{\textrm{poly}}:\mathcal{X}\times\mathcal{X}\to\R\), defined such that
\begin{equation}
\label{PolynomialKernel}
\mathcal{K}_{\textrm{poly}}(x,x^\prime) = \big(\gamma\langle x,x^\prime\rangle+c_0\big)^d.
\end{equation}
The third is the \emph{radial basis function} (RBF) \emph{kernel}, denoted \(\mathcal{K}_{\textrm{RBF}}:\mathcal{X}\times\mathcal{X}\to\R\), defined such that
\begin{equation}
\label{RBFKernel}
\mathcal{K}_{\textrm{RBF}}(x,x^\prime) = e^{-\gamma\|x-x^\prime\|^2},
\end{equation}
where \(\|\cdot\|\) is the standard Euclidean norm. 
The fourth is the \emph{Laplacian kernel}, denoted \(\mathcal{K}_{\textrm{Lap}}:\mathcal{X}\times\mathcal{X}\to\R\), defined such that
\begin{equation}
\label{LaplacianKernel}
\mathcal{K}_{\textrm{Lap}}(x,x^\prime) = e^{-\gamma\|x-x^\prime\|_1},
\end{equation}
where \(\|\cdot\|_1\) is the \(l_1\) norm. 
The fifth is the \emph{sigmoid kernel}, denoted \(\mathcal{K}_{\textrm{sig}}:\mathcal{X}\times\mathcal{X}\to\R\), defined such that
\begin{equation}
\label{SigmoidKernel}
\mathcal{K}_{\textrm{sig}}(x,x^\prime) = \tanh\big(\gamma\langle x,x^\prime\rangle+c_0\big).
\end{equation}
And the sixth kernel is the \emph{cosine kernel}, denoted \(\mathcal{K}_{\textrm{cos}}:\mathcal{X}\times\mathcal{X}\to\R\), defined such that
\begin{equation}
\label{CosineKernel}
\mathcal{K}_{\textrm{cos}}(x,x^\prime) = \frac{\langle x,x^\prime\rangle}{\|x\|\,\|x^\prime\|}.
\end{equation}

Each of these kernels are commonly used in ML algorithms that apply KMs to regression problems, which lead to our use of them here. Additionally, the Scikit Learn Python package~\cite{Pedregosa2011Scikit} has a function \texttt{pairwise\_kernels} for conveniently calculating the associated kernel matrices. 
Using this function, the 10-fold cross-validation for the regularisation strength \(\lambda\) and other hyperparameters \(\gamma\), \(c_0\) and \(d\) is performed.

\subsection{Generating the data}
\label{Generating the data}

To generate the data for this work (see Data Availability section to access the data), we use an efficient numerical algorithm based on matrix product operators (MPOs)~\cite{Xu2019Accessing}, which are classical tensor networks well suited to simulating local 1D quantum systems.
Specifically, we sample inputs \(x\) from subsets of \(\R^3\) (see Section \ref{The parameterised sets of Hamiltonians}), which effectively selects a Hamiltonian $H(x)$ from one of the sets $\mathscr{H}$.
We then directly calculate the associated labels determined by \(\mathscr{O}_{XZ}\) and \(\mathscr{O}_{Sum}\) with the MPO-based algorithm.
This involves representing each of the operators under the trace in \eqref{XZOTOC} and \eqref{SumOfOTOCs} as an MPO, evolving them in the Heisenberg picture using time-evolving block decimation (TEBD) \cite{Vidal2003Efficient,Vidal2004Efficient}, then contracting the MPOs together to calculate the OTOC.

Applying the algorithm requires us to choose an appropriate maximum bond dimension \(\chi\) for the MPOs.
A larger \(\chi\) makes the algorithm more computationally expensive to implement, but also means that the computed numerical values are generally more accurate.
Accordingly, we need to pick a bond dimension that both fits within the computational budget, and leads to reasonably accurate labels for the data points. 
In Figure \ref{ConvergenceInChiFigure} of Appendix \ref{Convergence in the bond dimension}, we provide plots of some of the labels for the 40-qubit systems against the value of \(\chi\) used to calculate them. 
This shows how well the labels have converged with respect to the bond dimensions selected for producing the datasets.
In an attempt to balance the trade-off between compute costs and convergence, we chose to use a maximum bond dimension \(\chi\) of 110, 140, 110 and 150 to calculate the values of the function \(\mathscr{O}_{XZ}\) with \(\mathscr{H}_1\), \(\mathscr{H}_2\), \(\mathscr{H}_3\) and \(\mathscr{H}_4\), respectively.
Similarly, we chose to use a maximum bond dimension \(\chi\) of 70, 90, 70 and 100 to calculate the values of the function \(\mathscr{O}_{Sum}\) with \(\mathscr{H}_1\), \(\mathscr{H}_2\), \(\mathscr{H}_3\) and \(\mathscr{H}_4\), respectively.

In addition to this selection of bond dimensions, we use a Trotter step size of \(\Delta t=0.05\) to perform the Heisenberg picture time evolution of the MPOs using TEBD. 
This seems reasonable when one considers that the magnitude of the error terms involved in the algorithm scale as \(\mathcal{O}\left((\Delta t)^3\right)\). 
We also make use of a time-splitting method (see Section VI.B.2 in~\cite{Xu2023Tutorial}) to improve the numerical accuracy of the algorithm by reducing the amount of time that the MPOs need to be evolved for. 

Overall, the runtime for calculating the OTOC with this method scales as
\begin{align}
\mathcal{O}\left(n\chi^3t\slash\Delta t+n\chi^4\right),
\label{MPOruntime}
\end{align} 
where $n$ is the size of the quantum system and $t$ is the evolution time. The first term in \eqref{MPOruntime} comes from performing TEBD to evolve the MPOs (see Lemma 2 of~\cite{Vidal2003Efficient}), and the second term comes from contracting the MPOs together to calculate the OTOC (see Section III.A.3 in~\cite{Stoudenmire2010Minimally}).
Compared with the exponential runtime scaling for exact diagonalisation, the runtime for the MPO-based method appears to scale just linearly with the system size $n$.
However as the system size increases, we increase the maximum evolution time $t_{max}$ linearly (see Section \ref{The parameterised sets of Hamiltonians}), and typically a larger $\chi$ is necessary to maintain accuracy, which effectively results in a high degree polynomial runtime scaling in $n$.

Using this method, for each qubit number considered ($n\in\{5,10,15,20, 25,30,35,40\}$), each parameterised set of Hamiltonians, and both functions \(\mathscr{O}_{XZ}\) and \(\mathscr{O}_{Sum}\), we generate 1250 labelled data samples, 1000 for training and 250 for testing.
The random samples of the inputs are drawn using the same random seed. Thus, for each qubit number, all datasets contain the same input parameters. The values of the functions are then computed using the ITensor Julia library~\cite{Fishman2022ITensor} to implement the MPO-based time-splitting algorithm. 
It should be noted though, that even with the MPO-based algorithm, producing the data in this work required $\sim 10^6$ CPU hours, and thus required access to a supercomputer.

\section{Results}
\label{Section4}

\begin{figure*}
\centering
\includegraphics[width=\textwidth]{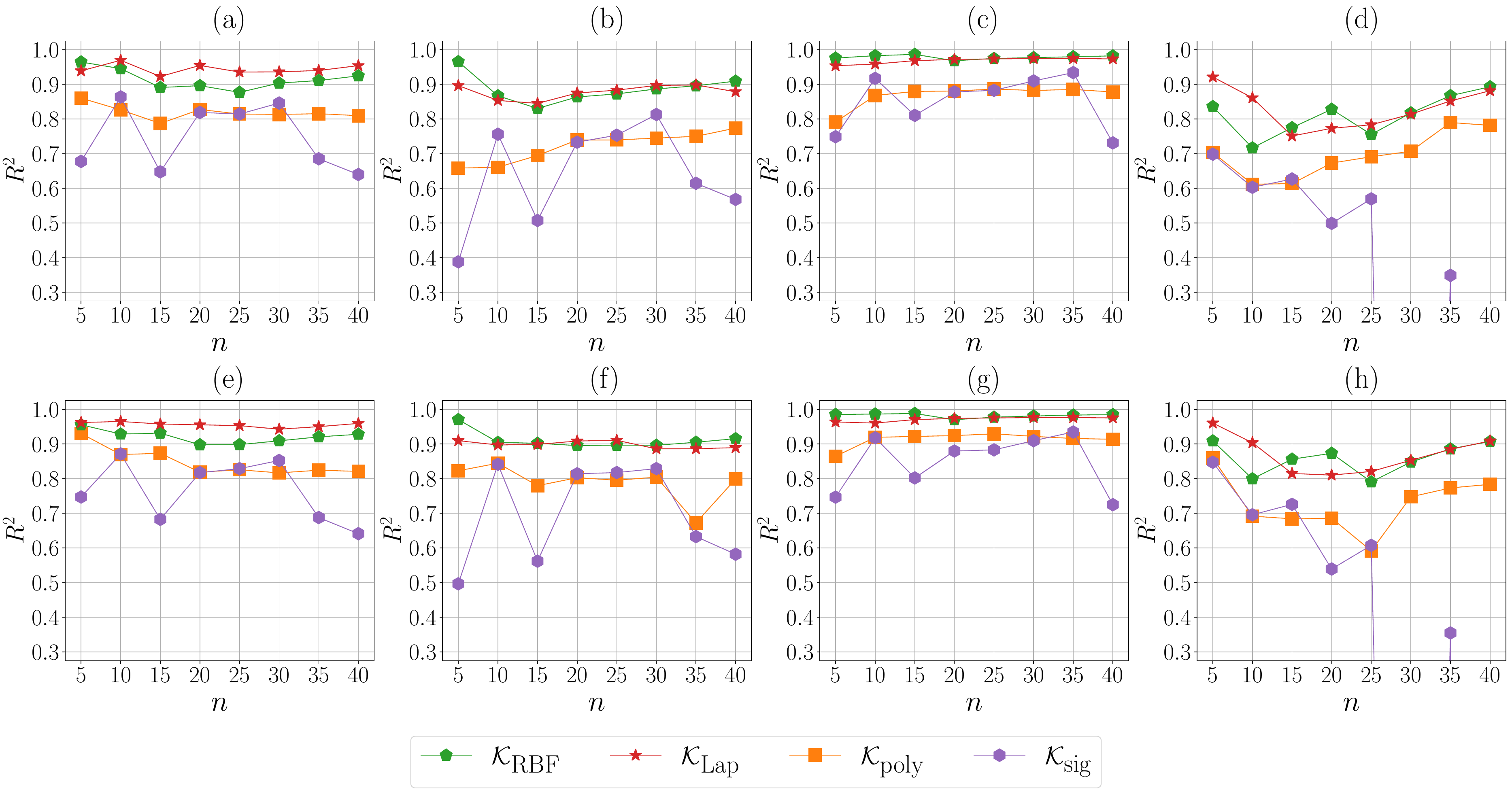}
\caption{Coefficient of determination ($R^2$) for the models trained on all 1000 training datapoints with the best hyperparameter values trialled in the cross-validation, making predictions on the testing sets. 
The top row of plots shows the results for the datasets with labels determined by \(\mathscr{O}_{XZ}\) with (a) \(\mathscr{H}_1\), (b) \(\mathscr{H}_2\), (c) \(\mathscr{H}_3\) and (d) \(\mathscr{H}_4\).
The bottom row of plots shows the results for the datasets with labels determined by \(\mathscr{O}_{Sum}\) with (e) \(\mathscr{H}_1\), (f) \(\mathscr{H}_2\), (g) \(\mathscr{H}_3\) and (h) \(\mathscr{H}_4\).}
\label{Results_R2}
\end{figure*}

We now report the main numerical results of this study. 
In Figure \ref{Results_R2} we plot the coefficient of determination (\(R^2\)) performance metric for the RBF, Laplacian, polynomial, and sigmoid kernels based on predictions made on the testing sets for all problem instances. 
The linear and cosine kernels are not included in the figure because they perform so poorly that adjusting the plot ranges to include them would obscure details in the important \(R^2\in[0.8,1]\) range.
Table \ref{Results_Summary} presents a summary of the \(R^2\), root-mean-squared-error (RMSE), and mean-absolute-error (MAE) performance metrics obtained with the Laplacian and RBF kernels, which consistently achieved the best performance metrics of all the kernels we consider.
Additional numerical results, including the \(R^2\), RMSE, and MAE values for all six kernels, making predictions on both the training and testing sets, are provided in Tables \ref{ResultsOXZH1}, \ref{ResultsOXZH2}, \ref{ResultsOXZH3}, \ref{ResultsOXZH4}, \ref{ResultsOSumH1}, \ref{ResultsOSumH2}, \ref{ResultsOSumH3}, and \ref{ResultsOSumH4} of Appendix \ref{Tables of numerical results}.
All results in Appendix \ref{Tables of numerical results} (which includes the \(R^2\) values in Figure \ref{Results_R2}, and the results which are summarised in Table \ref{Results_Summary}) are obtained using models trained on all 1000 training datapoints with the best hyperparameter values trialled in the cross-validation (see Section \ref{The learning problem}). 
In Figure \ref{ConvergenceInNumTrainData}, we also provide results showing how the kernels perform when supplied less than 1000 training datapoints.

Since the linear and cosine kernels perform poorly, we will not discuss their performance in detail, but instead provide reasons explaining their poor performance.
Specifically, we can combine  \eqref{RTmodel} and \eqref{LinearKernel} to show that the linear kernel can only be used to learn models of the form,
\begin{align}
f(x)=\left\langle x,\sum_{i=1}^{M}\alpha_i\mathbf{x}_i\right\rangle.
\label{LinearModel}
\end{align}
Similarly, combining \eqref{RTmodel} and \eqref{CosineKernel} shows that the cosine model can only be used to learn models of the form,
\begin{align}
f(x)=\left\langle\frac{x}{\|x\|},\sum_{i=1}^{M}\alpha_i\frac{\mathbf{x}_i}{\|\mathbf{x}_i\|}\right\rangle.
\label{CosineModel}
\end{align}
From \eqref{LinearModel} is is clear that the linear kernel can only be used to express linear functions, while \eqref{CosineModel} shows that the cosine kernel can only be used to express functions which are independent of the Euclidean norm of their argument. 
However, from \eqref{XZOTOC} and \eqref{SumOfOTOCs}, one can see that $\mathscr{O}_{XZ}$ and $\mathscr{O}_{Sum}$ are highly non-linear and depend substantively on the Euclidean norm of their argument (see Section \ref{The parameterised sets of Hamiltonians}). 
For this reason, we should not expect the linear and cosine kernels to perform well on the learning problems considered in this work.

On the other hand, the polynomial and sigmoid kernels perform reasonably well. 
As can be seen in Figure \ref{Results_R2}, at least one of these kernels achieves an \(R^2\) score on the testing sets exceeding 0.8 in 37 of the 64 problem instances, with similar RMSE and MAE values to the Laplacian and RBF kernels (see Appendix \ref{Tables of numerical results}). 
Note though, that the sigmoid kernel performs poorly at predicting both \(\mathscr{O}_{XZ}\) and \(\mathscr{O}_{Sum}\) with \(\mathscr{H}_4\) for system sizes of \(n\in\{30, 40\}\) (see Figure \ref{Results_R2}(d) and \ref{Results_R2}(h)), achieving large negative \(R^2\) scores which fall outside the plot ranges.
Despite the overall good performance of the polynomial and sigmoid kernels, it is clear from Figure \ref{Results_R2} that the Laplacian and RBF kernels consistently perform the best.
For this reason we focus the rest of this section primarily on reporting the performance of the Laplacian and RBF kernels.

\begin{table*}
\begin{adjustbox}{center,max width=\textwidth}
\begin{tabular}{| c | c | c | c  c  c | c  c  c | c  c  c |}

\cline{4-12}

\multicolumn{3}{c}{}&\multicolumn{3}{|c|}{\(\boldsymbol{R^2}\)}&\multicolumn{3}{c|}{\textbf{RMSE}}&\multicolumn{3}{c|}{\textbf{MAE}}\\

\cline{1-12}
\multicolumn{1}{|c|}{\textbf{Function}}&\multicolumn{1}{c|}{\textbf{Kernel}}&\multicolumn{1}{c|}{\textbf{Hamiltonians}}&\multicolumn{1}{c|}{Min}&\multicolumn{1}{c|}{Mean}&\multicolumn{1}{c|}{SD}&\multicolumn{1}{c|}{Max}&\multicolumn{1}{c|}{Mean}&\multicolumn{1}{c|}{SD}&\multicolumn{1}{c|}{Max}&\multicolumn{1}{c|}{Mean}&\multicolumn{1}{c|}{SD}\\

\cline{1-12}


\multirow{8}{*}{\(\mathscr{O}_{XZ}\)}&\multirow{4}{*}{$\mathcal{K}_{\textrm{Lap}}$}&\(\mathscr{H}_1\)& 0.9228 & 0.9440 & 0.0136 & 0.1137 & 0.0930 & 0.0150 & 0.0763 & 0.0562 & 0.0112 \\

&&\(\mathscr{H}_2\)& 0.8452 & 0.8785 & 0.0188 & 0.1372 & 0.1229 & 0.0105 & 0.0920 & 0.0805 & 0.0075 \\

&&\(\mathscr{H}_3\)& 0.9537 & 0.9686 & 0.0075 & 0.0843 & 0.0714 & \textbf{0.0060}& 0.0573 & 0.0483 & \textbf{0.0055} \\

&&\(\mathscr{H}_4\)& 0.7511 & 0.8298 & 0.0552 & 0.1282 & 0.1085 & 0.0110 & 0.0715 & 0.0603 & 0.0073 \\

\cline{2-12} 


&\multirow{4}{*}{$\mathcal{K}_{\textrm{RBF}}$}&\(\mathscr{H}_1\)& 0.8771 & 0.9142 & 0.0273 & 0.1359 & 0.1132 & 0.0174 & 0.0790 & 0.0634 & 0.0123 \\

&&\(\mathscr{H}_2\)& 0.8306 & 0.8865 & 0.0372 & 0.1399 & 0.1158 & 0.0173 & 0.0894 & 0.0766 & 0.0132 \\

&&\(\mathscr{H}_3\)& \textbf{0.9682} & \textbf{0.9785} & \textbf{0.0054} & \textbf{0.0724} & \textbf{0.0593} & 0.0079 & \textbf{0.0511} & \textbf{0.0393} & 0.0064 \\

&&\(\mathscr{H}_4\)& 0.7167 & 0.8112 & 0.0550 & 0.1499 & 0.1164 & 0.0201 & 0.0731 & 0.0585 & 0.0097 \\

\cline{1-12}


\multirow{8}{*}{\(\mathscr{O}_{Sum}\)}&\multirow{4}{*}{$\mathcal{K}_{\textrm{Lap}}$}&\(\mathscr{H}_1\)& 0.9429 & 0.9558 & 0.0066 & 0.8534 & 0.7363 & 0.0687 & 0.5653 & 0.4475 & 0.0682 \\

&&\(\mathscr{H}_2\)& 0.8862 & 0.8985 & 0.0097 & 1.0664 & 0.9749 & 0.0713 & 0.7712 & 0.6841 & 0.0661 \\

&&\(\mathscr{H}_3\)& 0.9606 & 0.9716 & 0.0058 & 0.6719 & 0.5973 & \textbf{0.0302} & 0.4879 & 0.3979 & \textbf{0.0402} \\

&&\(\mathscr{H}_4\)& 0.8103 & 0.8697 & 0.0505 & 0.9806 & 0.8370 & 0.1279 & 0.5750 & 0.4868 & 0.0760 \\

\cline{2-12} 


&\multirow{4}{*}{$\mathcal{K}_{\textrm{RBF}}$}&\(\mathscr{H}_1\)& 0.8979 & 0.9214 & 0.0180 & 1.1179 & 0.9772 & 0.1227 & 0.6916 & 0.5689 & 0.0821 \\

&&\(\mathscr{H}_2\)& 0.8957 & 0.9109 & 0.0234 & 0.9708 & 0.8939 & 0.1096 & 0.6700 & 0.5925 & 0.1095 \\

&&\(\mathscr{H}_3\)& \textbf{0.9703} & \textbf{0.9822} & \textbf{0.0055} & \textbf{0.6243} & \textbf{0.4737} & 0.0820 & \textbf{0.4376} & \textbf{0.3113} & 0.0693 \\

&&\(\mathscr{H}_4\)& 0.7908 & 0.8589 & 0.0420 & 1.0930 & 0.8920 & 0.1084 & 0.5934 & 0.4860 & 0.0740 \\

\cline{1-12}

\end{tabular}
\end{adjustbox}
\caption{Summary of the \(R^2\), RMSE, and MAE values for the Laplacian and RBF kernels on the testing sets with labels determined by \(\mathscr{O}_{XZ}\) or \(\mathscr{O}_{Sum}\), and \(\mathscr{H}_1\), \(\mathscr{H}_2\), \(\mathscr{H}_3\), or \(\mathscr{H}_4\). 
The table includes the smallest \(R^2\), highest RMSE, and highest MAE, together with the mean and standard deviation (SD) of the eight values of each metric obtained across eight system sizes (\(n\in\{5,10,15,\ldots,40\}\)). 
For each function, the best value of each quantity reported in the table is written in boldface, with a smaller SD considered better so that the reported means are more statistically representative of expected performance.
Note that the models used to obtain the results summarised here are trained on all 1000 training datapoints with the best hyperparameter values trialled in the cross-validation.}
\label{Results_Summary}
\end{table*}

In Table \ref{Results_Summary}, we provide a summary of the \(R^2\), RMSE, and MAE values obtained with the Laplacian and RBF kernels on the testing sets with labels determined by \(\mathscr{O}_{XZ}\) or \(\mathscr{O}_{Sum}\), and one of the parameterised sets of Hamiltonians \(\mathscr{H}_1\), \(\mathscr{H}_2\), \(\mathscr{H}_3\), or \(\mathscr{H}_4\). 
Specifically, the table contains the smallest \(R^2\), highest RMSE, and highest MAE, together with the mean and standard deviation of the eight values of each metric obtained across the eight system sizes \(n\in\{5,10,15,\ldots,40\}\). 
These results offer an overview of how well the Laplacian and RBF kernels perform in learning the two OTOC functions, \(\mathscr{O}_{XZ}\) and \(\mathscr{O}_{Sum}\), across a range of quantum systems with various system sizes up to 40 qubits. 
This provides insight into the kernels' suitability for the learning problems under study.

From Table \ref{Results_Summary}, we see that in the problem instances associated with \(\mathscr{O}_{XZ}\), the \(R^2\) scores achieved by the Laplacian and RBF kernels on the testing sets are all at least 0.7167, with averages over the eight system sizes lying between 0.8112 and 0.9785 for the different sets of Hamiltonians.  Similarly, for the problem instances associated with \(\mathscr{O}_{Sum}\), the \(R^2\) achieved by the kernels are at least 0.7908, with averages lying between 0.8598 and 0.9822 for the various Hamiltonians. 
We also see that the collections of \(R^2\) scores are distributed with small standard deviations, the greatest of which is 0.0552. 
This indicates that the \(R^2\) scores are distributed closely around their associated means, which is consistent with Figure \ref{Results_R2}.

Table \ref{Results_Summary} also shows that, for the problem instances associated with \(\mathscr{O}_{XZ}\), the RMSE and MAE values all lie below 0.1499 and 0.0920 respectively. While those for the problem instances associated with \(\mathscr{O}_{Sum}\) all lie below 1.1179 and 0.7712 respectively. 
This indicates that the numerical error in the labels predicted by either one of the kernels for \(\mathscr{O}_{XZ}\), or \(\mathscr{O}_{Sum}\), will on average be much smaller than \(7.5\%\),  or \(9.3\%\), of the range of the functions which return values in \([-1,1]\), or \([-3,9]\), respectively.
Note though, that the smallest numerical value of any label determined by \(\mathscr{O}_{XZ}\), and \(\mathscr{O}_{Sum}\), present in the datasets was -0.2297, and -0.1086, respectively.

\begin{figure*}
\centering
\begin{tikzpicture}
\node at (0,0) {\includegraphics[width=0.95\textwidth]{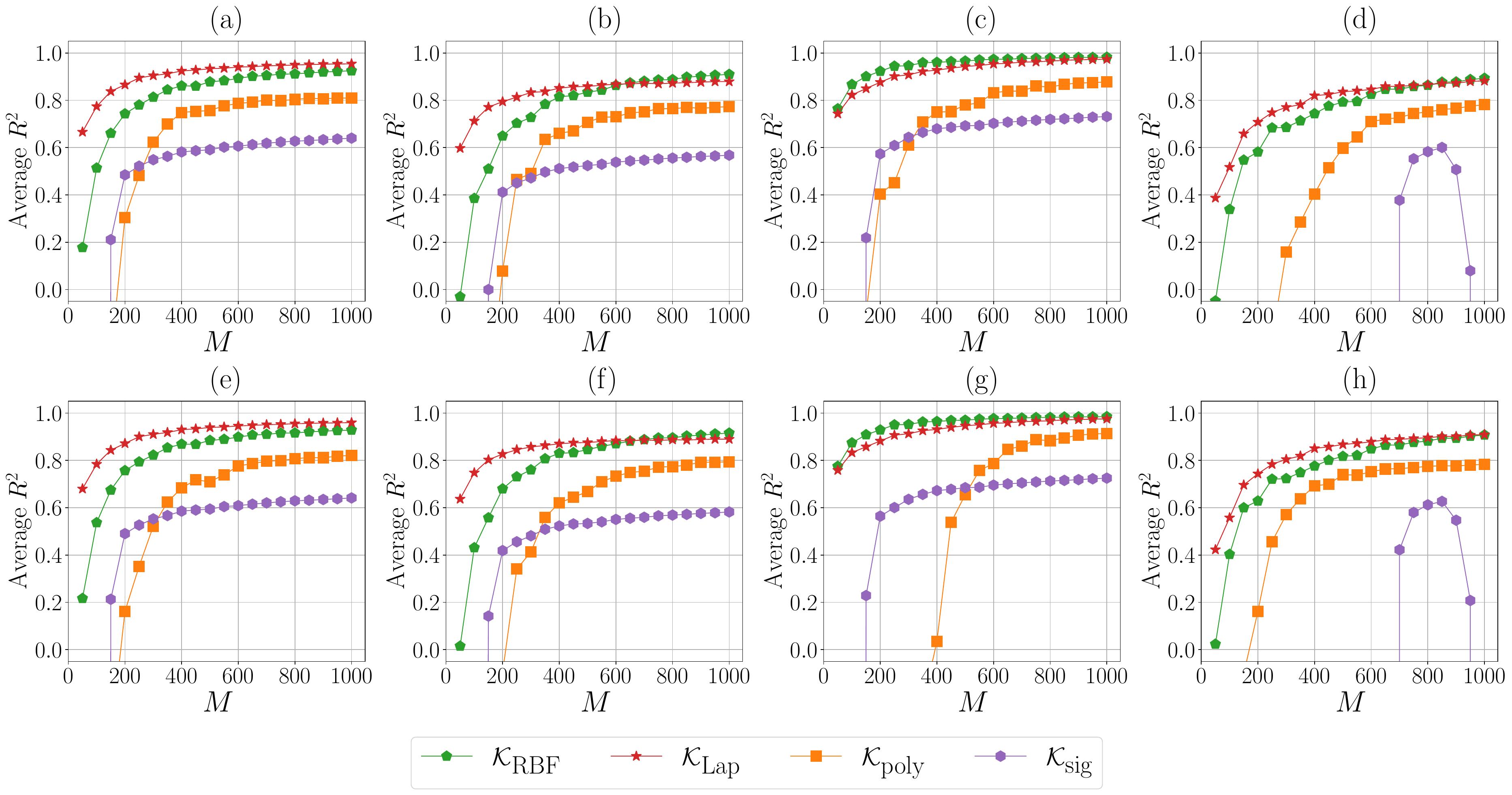}};
\end{tikzpicture}
\caption{Average coefficient of determination (Average \(R^2\)) for each kernel making predictions on the 40-qubit testing sets, plotted against the number of training data samples \(M\) used to train the associated ML models. 
Each point in the plots corresponds to the numerical value of \(R^2\) averaged over 20 different models using the same kernel. 
Each of the models is the result of 
using a random subset of the training data containing \(M\) of the total 1000 training data samples. 
The hyperparameters are fixed to be the best hyperparameters found during the 10-fold cross-validation performed on all 1000 training data samples. 
Plots (a), (b), (c) and (d) show the average \(R^2\) score on the testing datasets with labels determined by \(\mathscr{O}_{XZ}\) with \(\mathscr{H}_1\), \(\mathscr{H}_2\), \(\mathscr{H}_3\) and \(\mathscr{H}_4\), respectively. Plots (e), (f), (g) and (h) show the average \(R^2\) score on the testing datasets with labels determined by \(\mathscr{O}_{Sum}\) with \(\mathscr{H}_1\), \(\mathscr{H}_2\), \(\mathscr{H}_3\) and \(\mathscr{H}_4\), respectively.}
\label{ConvergenceInNumTrainData}
\end{figure*}

Another aspect of the results which warrants mention are the variations in the best trialled values of each hyperparameter (see Tables \ref{BestHyperparametersH1}, \ref{BestHyperparametersH2}, \ref{BestHyperparametersH3}, and \ref{BestHyperparametersH4} of Appendix \ref{Tables of hyperparameter values}), and the performance of each kernel, as the system size \(n\) changes.
For the kernels which did not perform well (e.g., the linear and cosine kernels), one would not expect the hyperparameter values or performance metrics to change in a predictable way. 
However for the Laplacian and RBF kernels, which both depend on two hyperparameters \(\lambda\) and \(\gamma\), one might expect the optimal choice of these hyperparameters and the associated performance metrics to change in a predictable manner. 
Unfortunately, this does not appear to be the case, which raises the question of the underlying causes of these variations.

One factor of our method which likely contributes to these variations is the decreasing density of input training data samples with increasing system size \(n\).
This happens because we sample the input training data from a sphere of radius proportional to \(n\), while the total number of training data samples remains fixed at 1000. 
The decrease in density likely impacts the value of \(\gamma\), since \(\gamma\) controls the rate of exponential decay of the terms in \eqref{RTmodel}, which in turn affects how the model interpolates between the input training data samples.
In order to establish a clearer trend for the hyperparameter values, and hence the associated performance metrics, we expect that maintaining the density of training data samples would assist. 
However, this would require a number of training data samples scaling cubically with the system size (as does the volume of the sphere from which the data is sampled), which was not feasible with our computational budget.

Similarly, while the Laplacian kernel may appear to perform better than the RBF kernel in many instances, it also exhibits a larger difference between training and testing performance (see Appendix \ref{Tables of numerical results}). 
This could be interpreted as partial overfitting, 
which can usually be remedied by using a small positive value of $\lambda$.
However the best regularisation strength among those trialled for the Laplacian kernel was frequently \(\lambda=0\) (see Appendix \ref{Tables of hyperparameter values}). 
This, and the absence of a clear trend in the best trialled hyperparameter values, suggests that a more finely grained grid search may improve the results.
This would allow hyperparameter values which are closer to optimal to be found, which may help in uncovering the underlying trend that one might expect to observe in the optimal values.
However we did not perform this here in order to keep the costs of training manageable (see~\eqref{NsamplesKMruntime}).

Given the cost of producing training data, it is natural to ask whether one can obtain similar results with fewer data samples.
To investigate this, we performed a collection of learning experiments which involve varying the amount of training data supplied to the models.
First, we set the hyperparameter values for each kernel in each problem instance equal to the best values trialled in the cross-validation performed on all of the training data (see Appendix \ref{Tables of hyperparameter values}).
Next, we choose different training dataset sizes \(M\in\N\) to supply to the models. 
Specifically, for each \(M\in\{50m:m=1,2,\ldots,20\}\), we randomly sample \(M\) of the 1000 training data points, and repeat this 20 times to obtain 20 (possibly intersecting) subsets of size \(M\).
We then train a model on each of these 20 subsets of the training data. 
Finally, we calculate the average \(R^2\) score achieved by the 20 models making predictions on the entire testing set, as an indicator of how the models perform with only \(M\) training data samples. 
The results of this analysis are shown in Figure \ref{ConvergenceInNumTrainData} for the 40-qubit problem instances. 
The plots show that, even with only a few hundred training data points, the average \(R^2\) scores for the Laplacian and RBF kernels often remain above \(0.8\), and occasionally reach above \(0.9\). 
This shows that the kernels are capable of obtaining good learning performance metrics when supplied with fewer training data samples than were used in the earlier analysis.

Overall, the results show that the learning performance of the RBF and Laplacian kernels are similar and consistently the best (see Figure \ref{Results_R2}). 
On all problem instances, 
they obtain good values for all three performance metrics considered (see Table \ref{Results_Summary}) after tuning their hyperparameters with a standard and inexpensive cross-validation. 
Even when supplied with just a few hundred training data samples, the RBF and Laplacian kernels manage to achieve $R^2$ scores frequently in excess of 0.8 (see Figure \ref{ConvergenceInNumTrainData}).
We also note that the results indicate that the performance of the Laplacian and RBF kernels depends more on the choice of Hamiltonian than the OTOC function we try to learn. 
For example, it is clear from Table \ref{Results_Summary} that the kernels perform better at learning OTOCs describing \(\mathscr{H}_3\) than OTOCs describing the other sets.

\section{Discussion}
\label{Section5}

The results in the previous section demonstrate that the Laplacian and RBF kernels can be used to obtain accurate estimates of OTOCs and related quantities, such as the lower bound in \eqref{MutualInformationAndOTOC}, from a relatively small amount of training data. 
The suitability of these kernels for the learning problems described in Section \ref{The learning problem} suggests that the OTOC functions $\mathscr{O}_{XZ}$ and $\mathscr{O}_{Sum}$ change smoothly in such a way that the value of either function, for a given input, can be closely inferred from the function values at nearby inputs. 
This behaviour is captured well by the Laplacian and RBF kernels since the value of these kernels decay exponentially with the distance between their arguments (as measured by the \(l_1\) and squared \(l_2\) norms respectively). 
The polynomial and sigmoid kernels also performed reasonably well, but the results show that the Laplacian and RBF kernels consistently performed best of all six kernels that we trialled in this work.

Given that the Laplacian and RBF kernels are capable of accurately predicting OTOCs from a small amount of training data, this means that we can use KMs together with the MPO-based algorithm~\cite{Xu2019Accessing} to obtain a reduction in overall runtime when one is interested in evaluating the OTOC for many different systems and values of time. 
Specifically, once we have a trained model, it can be used to substitute the MPO-based algorithm by making accurate predictions of OTOC values in a time scaling linearly in just the size of the training dataset (i.e., independent of the system size).

To see this, suppose that we are interested in evaluating \(\mathscr{O}_{XZ}\) or \(\mathscr{O}_{Sum}\) for a large number \(N\in\N\) of different inputs drawn from the spherical domains described in Section \ref{The parameterised sets of Hamiltonians}. 
From \eqref{MPOruntime}, the runtime scaling of using the MPO-based algorithm to calculate these values is
\begin{align}
\label{NsamplesMPOruntime}
\mathscr{O}\left(Nn\chi^3 t\slash\Delta t+Nn\chi^4\right).
\end{align}
If, however, we decide to instead produce $M$ training data samples using the MPO-based algorithm, where \(M<<N\), and then predict the OTOC values using the Laplacian or RBF kernels, then the overall runtime scales as
\begin{align}
\label{NsamplesKMruntime}
\mathcal{O}\left(Mn\chi^3t\slash\Delta t+Mn\chi^4+N_{\lambda}N_{\gamma}M^3+NM\right),
\end{align}
where $N_{\lambda}$ and $N_{\gamma}$ are the number of different values of $\lambda$ and $\gamma$, respectively, used in the grid search. 
The first two terms in \eqref{NsamplesKMruntime} come from producing the \(M\) training datapoints, the third term from training the model using KMs with a cross-validation by grid search, and the last term from predicting the OTOC values with the trained model (see the end of Section \ref{Kernel methods}). 
Comparing \eqref{NsamplesMPOruntime} with \eqref{NsamplesKMruntime} shows that, after producing the training dataset and training the model, the cost of making further predictions is reduced significantly from \(\mathcal{O}\left(n\chi^3t\slash\Delta t+n\chi^4\right)\) to \(\mathcal{O}(M)\) per prediction. 
So in cases where the number of required training data samples is much smaller than the overall number of predictions we would like to make (i.e., \(M<<N\)), applying KMs may significantly reduce the overall computational cost of obtaining the OTOC values.

This may be useful, for example, to a user who is interested in finding a quantum system which exhibits some interesting property that can be inferred from OTOCs. 
For instance, perhaps one is interested in seeking a highly chaotic system, so that local information quickly spreads across the system, which may be convenient in situations such as the classical communication protocol depicted in Figure \ref{Fig2}. 
One possible way of finding such a system would be to search for a choice of Hamiltonian parameters which results in the greatest quantum Lyapunov exponent. 
This entails calculating the short-time regime behaviour of the OTOC for many different systems and values of time associated with various choices of the Hamiltonian parameters. 
In such cases, where \(N>>M\), the proposed runtime reduction could be practically useful.

The main cost of applying KMs to the learning problems we consider is that of producing the training datasets.
And in order to produce the data in this work, as discussed in Section \ref{Generating the data}, we make use of the MPO-based algorithm, since this algorithm lends itself well to simulating local 1D quantum systems.
However, producing the data, even with this algorithm, can become expensive as the size of the system becomes large. 
This is because the maximum evolution time that we simulate, and often the necessary maximum bond dimension, increases with system size, leading to a high degree polynomial runtime scaling in the system size.

One could also consider generating data, or even directly calculating the OTOC values, on a QC instead.
However, we expect that currently available NISQ devices will require substantial improvements before it will be possible to match the capabilities of tensor network methods for simulating 1D quantum systems.
For example, neutral atom processors provide large qubit numbers, but exhibit two qubit gate error rates of roughly 0.5\%~\cite{Evered2023Neutral}, which severely limits circuit depths. 
In contrast, trapped ion devices provide small two qubit gate errors around 0.1\%, but only have on the order of 50 qubits~\cite{DeCross2024QuantinuumH2}.
In fact, of all currently available devices, only IBM's superconducting processors~\cite{Kim2023Evidence} have enough qubits with sufficiently small error rates to challenge classical methods for 1D systems. 
And it appears that even these devices would not be capable of generating the datasets in this work with good accuracy.

To see this, in Section I.B.1 of~\cite{Mohseni2024How}, the authors state that performing quantum simulations on NISQ hardware typically requires circuit depths greater than what can be performed with current devices, even those achieving a 0.1\% average two qubit gate error rate.
We believe the quantum simulations required to calculate the OTOCs considered in our work are no exception.
For example, Trotterising the Hamiltonians discussed in Section \ref{The parameterised sets of Hamiltonians} using a Trotter step size of \(\Delta t=0.05\), up to a time equal to the system size \(n\), requires at least \(20n\) layers of two qubit gates.
The number of layers which can feasibly be implemented depends on many factors. 
But as a baseline, in 2023 the authors of~\cite{Kim2023Evidence} reported an implementation of 60 layers of two qubit CNOT gates with high fidelity on IBM's 127-qubit superconducting processor \texttt{ibm\_kyiv}. 
Since the publication of~\cite{Kim2023Evidence}, the baseline error levels have marginally improved, so it is likely that more than 60 layers of CNOT gates would now be possible. 
However, the current error rates still exceed those which are expected to be necessary for many useful near-term applications (see Table I of~\cite{Mohseni2024How}).
And when we compare the baseline of 60 layers with $20n$, it becomes apparent that the number of layers required to use \(\Delta t=0.05\) will quickly exceed that which is feasible. 

One could argue that we may simply increase the Trotter step size to reduce the number of layers, while still simulating up to a time $t_{max}=n$. 
However, this will result in the Trotter error \(\mathcal{O}\left((\Delta t)^3\right)\) scaling cubically with the system size. 
For example, with 40 qubits and 100 layers, we would have \(\Delta t=0.4\), which increases the magnitude of the Trotter error by a factor of \(\sim 500\) when compared with \(\Delta t=0.05\).
This leads us to believe that current NISQ devices are not yet capable of generating the data used in this work with the same accuracy as the MPO-based algorithm, especially in cases with large $n$.
This is not in contradiction with the claims of \cite{Kim2023Evidence,DeCross2024QuantinuumH2}, since there the authors simulate quantum systems in higher dimensions, while here we consider only 1D quantum systems, to which tensor network methods are well suited. 
Instead, we simply argue that classical tensor networks are preferable over methods utilising NISQ devices for generating the datasets used in this work.

Despite this, one of the limitations of the approach taken here is that the OTOC labels, which were generated using the MPO-based algorithm, are only approximate due to the use of a finite bond dimension.
The analysis presented in Figure~\ref{ConvergenceInChiFigure} suggests that the chosen maximum bond dimension \(\chi\) was sufficiently large for most problem instances, since the OTOC values appear to have converged.
However, calculations for some of the Hamiltonians from the set \(\mathscr{H}_4\) (see Figures \ref{ConvergenceInChiFigure}(d) and \ref{ConvergenceInChiFigure}(h)) could have benefitted from a larger \(\chi\), though this would have been beyond our computational resources.
We also note that the plots in Figure~\ref{ConvergenceInChiFigure} correspond to the 40-qubit systems, which is the largest of the systems used here.
For smaller system sizes, the convergence with respect to \(\chi\) will occur more quickly.

Another limitation of our ML approach that warrants discussion is that one should not expect the kernels to be able to extrapolate beyond the ball in $\R^3$ from which the input training data is sampled.
Recall, the testing data, on which we evaluated the performance of the models, was drawn from the same distribution as the training data.
Thus, what we have demonstrated is that the kernels are able to make accurate predictions for other Hamiltonians within this region in parameter space.
This may still be useful however, since the region of parameter space on which the models perform well captures evolution times which scale linearly in the system size. 
As discussed in Section \ref{Information Spreading and Out-of-Time-Ordered Correlators}, we expect that this captures the short-time regime, and possibly part of the intermediate-time regime. 
Accordingly, we expect that the models may be used to make predictions for sufficiently large values of times to be of theoretical interest. 
In contrast, the only related work of which we are aware \cite{Wu2020Artificial} uses ML techniques to learn OTOCs in just a small part of the short-time regime (when Re$[F_{jk}](t)<0.85$). 
However, a direct comparison of our work with~\cite{Wu2020Artificial} is not straightforward since they demonstrated their method using a 2D quantum system, while this work focuses on 1D systems.

\section{Conclusion}
\label{Conclusion}

In this paper, we explored the application of classical KMs to learn and predict the OTOC, an important quantity used to investigate quantum information scrambling.
By viewing the XZ-OTOC, and the sum of OTOCs in \eqref{MutualInformationAndOTOC}, as functions of parameters associated with a parameterised set of Hamiltonians, we frame the learning task as a regression problem and employ six standard kernel functions. 
For all four of the parameterised sets of Hamiltonians considered, each describing local 1D quantum systems, we provide a significant body of numerical results (see Appendix \ref{Tables of numerical results}) which shows that the Laplacian and RBF kernels perform best. 
Indeed, both kernels closely approximate the XZ-OTOC and the sum of OTOCs, achieving $R^2$ scores on the testing data of at least 0.7167 and typically around 0.9, with small RMSE and MAE (see Table \ref{Results_Summary}), from a small number of training datapoints and a standard cross-validation.
In fact, Figure \ref{ConvergenceInNumTrainData} shows that the kernels can achieve $R^2$ scores in excess of 0.8 even with just a few hundred datapoints.

Since the kernels perform well with a small amount of training data, we can apply KMs to reduce the overall computational cost of extensively evaluating OTOCs.
Specifically, by using classical tensor network methods to generate a small amount of training data, training a model using the RBF or Laplacian kernel, then applying the model to make predictions, we can quickly generate accurate approximations of OTOCs. 
One must, of course, incur the cost of producing the training data, which can be expensive. 
However, in cases where one needs to evaluate OTOCs for many different systems and values of time belonging to the short- and possibly intermediate-time regimes, we can use kernel methods to reduce the cost of many of the predictions to a time scaling linearly in the number of training datapoints, as demonstrated by \eqref{NsamplesKMruntime}.

The paper naturally lends itself to a handful of possible extensions, which we suggest as future research directions. 
For example, one may like to investigate whether the KM approach works with other OTOCs, other linear combinations of OTOCs, or other types of correlation functions such as the Loschmidt echo \cite{Wisniacki2012Loschmidt} or entanglement entropy~\cite{Nielsen2000Quantum}. 
We could also examine how the approach works for other parameterised sets of Hamiltonians. 
For instance, we could investigate sets with \(d>3\) (i.e., parameterised by four or more parameters), sets containing Hamiltonians describing quantum systems in higher dimensions, or sets containing highly non-local Pauli strings in one or more terms.
Though producing data for the latter two cases would likely be challenging. 
It might also be interesting to apply quantum kernel methods \cite{Schuld2021Kernels,Liu2021Rigorous,Havlicek2019Supervised} to investigate whether a problem-specific quantum kernel may offer significant increases in learning performance. 
Similarly, it would be useful to investigate how the results change in response to increasing the size of the parameter region from which we sample the input data. 
This may shed light on whether the kernels can be used to accurately predict long-time OTOCs, though producing the data in this case would likely be expensive. 
Finally, as discussed earlier, it would be useful to explore whether a predictable trend in the hyperparameter values may be uncovered by using a more finely grained grid search.

\section*{Data availability}
\label{Data availability}

The datasets produced in this study are available in the GitHub repository at \href{https://github.com/John-J-Tanner/ITensor-OTOC-Data}{https://github.com/John-J-Tanner/ITensor-OTOC-Data}. 
The repository includes all relevant datasets necessary to replicate the study's results, and example Python codes that can be used to visualise the data.
For further information or inquiries, please contact the corresponding author, John Tanner, at \emph{john.tanner@uwa.edu.au}.

\section*{Acknowledgments}

The authors thank Jason Twamley for insightful conversations on the out-of-time ordered correlator, which inspired the initial concept of this study. 
The authors would also like to thank Lyle Noakes and Yusen Wu for invaluable discussions.
JT is supported by an Australian Government Research Training Program (RTP) Scholarship. 
JP was supported by the Brian Dunlop Physics Fellowship at The University of Western Australia and is supported at Nordita by the Wenner-Gren Foundations and, in part, by the Wallenberg Initiative on Networks and Quantum Information (WINQ).
Nordita is supported in part by NordForsk.
This work was supported by resources provided by the Pawsey Supercomputing Research Centre with funding from the Australian Government and the Government of Western Australia.

\clearpage

\bibliographystyle{unsrt}
\typeout{}
\bibliography{Bibliography_V2}


\clearpage



\appendix

\renewcommand{\thetheorem}{A.\arabic{theorem}} 
\setcounter{theorem}{0} 
\counterwithin{figure}{section}
\counterwithin{table}{section}


\section{Supplementary results}

\renewcommand\thefigure{A.\arabic{figure}} 
\renewcommand\thetable{A.\arabic{table}} 

\begin{theorem}
\label{A11} 
Let \(\mathcal{X}\equiv\R^d\) be the input data domain and \(\mathcal{D}=\{(\mathbf{x}_i,y_i)\}_{i=1}^{M}\subseteq\mathcal{X}\times\R\) be a training dataset. Let \(\mathcal{K}:\mathcal{X}\times\mathcal{X}\to\R\) be a kernel and define \(\widetilde{\mathcal{L}}_{\mathcal{D}}:\R^M\to\R\) s.t.
\begin{align}
\label{ARKHSRegularisedLTilde}
\widetilde{\mathcal{L}}_{\mathcal{D}}(\vec{\alpha})=\frac{1}{M}\big\|K\vec{\alpha}-\vec{y}\big\|^2+\frac{\lambda}{M}\big\langle\vec{\alpha},K\vec{\alpha}\rangle
\end{align}
where \(\|\cdot\|\) is the Euclidean norm, \(\langle\cdot,\cdot\rangle\) is the Euclidean inner product, \(\vec{y}=(y_i)_{i=1}^{M}\in\R^M\) is the vector of training data labels, \(K_{ij}=\mathcal{K}(\mathbf{x}_i,\mathbf{x}_j)\) is the  \emph{kernel matrix} for the input training data \(\{\mathbf{x}_i\}_{i=1}^{M}\).
Then \eqref{ARKHSRegularisedLTilde} defines a smooth convex map.
\end{theorem}

\begin{proof}
\renewcommand{\qed}{\hfill\(\blacksquare\)}
The map \(\widetilde{\mathcal{L}}_{\mathcal{D}}\) is clearly smooth since it is just a polynomial in the components of \(\vec{\alpha}\).
In order to show that \(\widetilde{\mathcal{L}}_{\mathcal{D}}\) is convex we need to show that 
\begin{equation}
\label{ConvexEqn}
\widetilde{\mathcal{L}}_{\mathcal{D}}\big(\mu\vec{\alpha}_1+(1-\mu)\vec{\alpha}_2\big)\leq\mu\widetilde{\mathcal{L}}_{\mathcal{D}}(\vec{\alpha_1})+(1-\mu)\widetilde{\mathcal{L}}_{\mathcal{D}}(\vec{\alpha}_2)
\end{equation}
for all \(\vec{\alpha}_1,\vec{\alpha}_2\in\R^M\) and \(\mu\in[0,1]\).
So let \(\vec{\alpha}_1,\vec{\alpha}_2\in\R^M\) and \(\mu\in[0,1]\) then
\begin{widetext}
\begin{align}
\nonumber
&\widetilde{\mathcal{L}}_{\mathcal{D}}\big(\mu\vec{\alpha}_1+(1-\mu)\vec{\alpha}_2\big)\\
\nonumber
=&\frac{1}{M}\big\|K\big(\mu \vec{\alpha}_1+(1-\mu)\vec{\alpha}_2\big)-\vec{y}\big\|^2
+\frac{\lambda}{M}\big\langle \mu\vec{\alpha}_1+(1-\mu)\vec{\alpha}_2,K\big(\mu\vec{\alpha}_1+(1-\mu)\vec{\alpha}_2\big)\big\rangle\\
\nonumber
=&\frac{1}{M}\big\|\mu(K\vec{\alpha}_1-\vec{y})+(1-\mu)(K\vec{\alpha}_2-\vec{y})\big\|^2
+\frac{\lambda}{M}\big\langle \mu\vec{\alpha}_1+(1-\mu)\vec{\alpha}_2,K\big(\mu\vec{\alpha}_1+(1-\mu)\vec{\alpha}_2\big)\big\rangle.\\
\nonumber
=&\frac{1}{M}\Big(\mu^2\big\|K\vec{\alpha}_1-\vec{y}\big\|^2+2\mu(1-\mu)\big\langle K\vec{\alpha}_1-\vec{y},K\vec{\alpha}_2-\vec{y}\big\rangle+(1-\mu)^2\big\|K\vec{\alpha}_2-\vec{y}\big\|^2\Big)\\
\nonumber
&\qquad+\frac{\lambda}{M}\Big(\mu^2\big\langle\vec{\alpha}_1,K\vec{\alpha}_1\big\rangle+2\mu(1-\mu)\big\langle\vec{\alpha}_1,K\vec{\alpha}_2\big\rangle+(1-\mu)^2\big\langle\vec{\alpha}_2,K\vec{\alpha}_2\big\rangle\Big).\\
\nonumber
=&\frac{1}{M}\Big(\big(-\mu(1-\mu)+\mu\big)\big\|K\vec{\alpha}_1-\vec{y}\big\|^2+2\mu(1-\mu)\big\langle K\vec{\alpha}_1-\vec{y},K\vec{\alpha}_2-\vec{y}\big\rangle+\big(-\mu(1-\mu)+(1-\mu)\big)\big\|K\vec{\alpha}_2-\vec{y}\big\|^2\Big)\\
\nonumber
&\qquad+\frac{\lambda}{M}\Big(\big(-\mu(1-\mu)+\mu\big)\big\langle\vec{\alpha}_1,K\vec{\alpha}_1\big\rangle+2\mu(1-\mu)\big\langle\vec{\alpha}_1,K\vec{\alpha}_2\big\rangle+\big(-\mu(1-\mu)+(1-\mu)\big)\big\langle\vec{\alpha}_2,K\vec{\alpha}_2\big\rangle\Big).\\
\nonumber
=&\frac{1}{M}\Big(\mu\big\|K\vec{\alpha}_1-\vec{y}\big\|^2+(1-\mu)\big\|K\vec{\alpha}_2-\vec{y}\big\|^2-\mu(1-\mu)\big\|\big(K\vec{\alpha}_1-\vec{y}\big)-\big(K\vec{\alpha}_2-\vec{y}\big)\big\|^2\Big)\\
\nonumber
&\qquad+\frac{\lambda}{M}\Big(\mu\big\langle\vec{\alpha}_1,K\vec{\alpha}_1\big\rangle+(1-\mu)\big\langle\vec{\alpha}_2,K\vec{\alpha}_2\big\rangle-\mu(1-\mu)\big\langle\vec{\alpha}_1-\vec{\alpha}_2,K(\vec{\alpha}_1-\vec{\alpha}_2)\big\rangle\Big).\\
\nonumber
\leq&\frac{1}{M}\Big(\mu\big\|K\vec{\alpha}_1-\vec{y}\big\|^2+(1-\mu)\big\|K\vec{\alpha}_2-\vec{y}\big\|^2\Big)+\frac{\lambda}{M}\Big(\mu\big\langle\vec{\alpha}_1,K\vec{\alpha}_1\big\rangle+(1-\mu)\big\langle\vec{\alpha}_2,K\vec{\alpha}_2\big\rangle\Big).\\
=&\mu\widetilde{\mathcal{L}}_{\mathcal{D}}(\vec{\alpha}_1) + (1-\mu) \widetilde{\mathcal{L}}_{\mathcal{D}}(\vec{\alpha}_2).
\end{align}
\end{widetext}
Hence \(\widetilde{\mathcal{L}}_{\mathcal{D}}\big(\mu\vec{\alpha}_1+(1-\mu)\vec{\alpha}_2\big)\leq\mu\widetilde{\mathcal{L}}_{\mathcal{D}}(\vec{\alpha}_1) + (1-\mu) \widetilde{\mathcal{L}}_{\mathcal{D}}(\vec{\alpha}_2)\) for all \(\vec{\alpha}_1,\vec{\alpha}_2\in\R^M\) and \(\mu\in[0,1]\), which is \eqref{ConvexEqn}. This concludes the proof.
\end{proof}

\begin{theorem}
\label{A12}
Let \(\mathcal{X}\equiv\R^d\) and \(\mathcal{D}=\{(\mathbf{x}_i,y_i)\}_{i=1}^{M}\subseteq\mathcal{X}\times\R\). Given a kernel \(\mathcal{K}:\mathcal{X}\times\mathcal{X}\to\R\), define the mean squared error \(L_{\mathcal{D}}:\mathcal{R}_{\mathcal{K}}\to\R\) such that \(L_{\mathcal{D}}(f)=\frac{1}{M}\sum_{i=1}^{M}(f(\mathbf{x}_i)-y_i)^2\), the regularisation term \(\Omega:\mathcal{R}_{\mathcal{K}}\to\R\) such that \(\Omega(f)=\frac{1}{M}||f||_{\mathcal{R}_{\mathcal{K}}}^2\) and the regularised empirical risk functional \(\mathcal{L}_{\mathcal{D}}:\mathcal{R}_{\mathcal{K}}\to\R\) such that \(\mathcal{L}_{\mathcal{D}}(f)=L_{\mathcal{D}}(f)+\lambda\Omega(f)\) for some \(\lambda>0\). Under these conditions, the optimal model \(f_{\textrm{opt}}\) for \(\mathcal{D}\) in \(\mathcal{R}_{\mathcal{K}}\) with respect to \(\mathcal{L}_{\mathcal{D}}\) is given by
\begin{equation}
\label{RKHSOptimalProof}
f_{\textrm{opt}}(x)=\sum_{i=1}^{M}\alpha_i\mathcal{K}(x,\mathbf{x}_i)
\end{equation}
for all \(x\in\mathcal{X}\) with
\begin{equation}
\label{RKHSOptimalAlphaProof}
\vec{\alpha}=(K^2+\lambda K)^+K\vec{y},
\end{equation}
where \(\vec{y}=(y_i)_{i=1}^{M}\in\R^M\) is the vector of training data labels, \(K_{ij}=\mathcal{K}(\mathbf{x}_i,\mathbf{x}_j)\) is the  \emph{kernel matrix} for the input training data \(\{\mathbf{x}_i\}_{i=1}^{M}\), and \((\cdot)^+\) is the Moore-Penrose pseudoinverse~\cite{Penrose1955PseudoInverse}.
\end{theorem}

\begin{proof}
\renewcommand{\qed}{\hfill\(\blacksquare\)}
In this case the representer theorem (Theorem 4.2 in~\cite{Scholkopf2001Kernels} and Theorem 6.11 in~\cite{Mohri2018Foundations}) states that the optimal model \(f_{\textrm{opt}}\) for \(\mathcal{D}\) in \(\mathcal{R}_{\mathcal{K}}\) with respect to \(\mathcal{L}_{\mathcal{D}}\) admits a representation of the form
\begin{align}
\label{foptProof}
f_{\textrm{opt}}(\cdot)=\sum_{i=1}^{M}\alpha_i\mathcal{K}(\cdot, \mathbf{x}_i)
\end{align}
for some \(\alpha_i\in\R\). Substituting \eqref{foptProof} into the regularised empirical risk functional we have that
\begin{widetext}
\begin{align}
\nonumber
\mathcal{L}_{\mathcal{D}}(f)
&=\frac{1}{M}\sum_{i=1}^{M}\left(\sum_{j=1}^{M}\alpha_j\mathcal{K}(\mathbf{x}_i,\mathbf{x}_j)-y_i\right)^2+\lambda\left\langle\sum_{k=1}^{M}\alpha_k\mathcal{K}(\cdot,x_k),\sum_{l=1}^{M}\alpha_l\mathcal{K}(\cdot,x_l)\right\rangle_{\mathcal{R}_{\mathcal{K}}}\\
\nonumber
&=\frac{1}{M}\sum_{i=1}^{M}\left(\sum_{j=1}^{M}\alpha_j\mathcal{K}(\mathbf{x}_i,x_j)-y_i\right)^2+\lambda\sum_{k,l=1}^{M}\alpha_k\mathcal{K}(x_k,x_l)\alpha_l
\end{align}
\end{widetext}
At this point we can view \(\alpha_i\) for \(i\in\{1,\ldots,M\}\) as variational parameters which parameterise the space of possible optimal models. This allows us to express the value of the regularised risk on this space of possible minimisers in terms of \(\vec{\alpha}\in\R^M\) by defining \(\widetilde{\mathcal{L}}_{\mathcal{D}}:\R^M\to\R\) such that
\begin{align}
\label{RKHSKernelMatrixRegularisedRisk}
\widetilde{\mathcal{L}}_{\mathcal{D}}(\vec{\alpha})&=\frac{1}{M}\big|\hskip-0.05cm\big|K\vec{\alpha}-\vec{y}\big|\hskip-0.05cm\big|^2+\lambda\left\langle\vec{\alpha},K\vec{\alpha}\right\rangle,
\end{align}
where \(\|\cdot\|\) is the Euclidean norm, \(\langle\cdot,\cdot\rangle\) is the Euclidean inner product, \(\vec{y}=(y_i)_{i=1}^{M}\in\R^M\) is the vector of training data labels, \(K_{ij}=\mathcal{K}(\mathbf{x}_i,\mathbf{x}_j)\) is the  \emph{kernel matrix} for the input training data \(\{\mathbf{x}_i\}_{i=1}^{M}\). By Theorem \ref{A11}, \eqref{RKHSKernelMatrixRegularisedRisk} defines a smooth convex function \(\R^M\to\R\) so we have that \(\widetilde{\mathcal{L}}_{\mathcal{D}}(\alpha^*)\) is the global minimum of \(\widetilde{\mathcal{L}}_{\mathcal{D}}\) if and only if \(\nabla_{\vec{\alpha}}\widetilde{\mathcal{L}}_{\mathcal{D}}(\vec{\alpha})|_{\vec{\alpha}=\alpha^*}=\vec{0}\). This implies that \(\alpha^*\in\R^M\) satisfies
\begin{align}
\nonumber
&\frac{2}{M}K(K\alpha^*-\vec{y})+2\frac{\lambda}{M} K\alpha^*=\vec{0}\\
\label{Eq28}
\implies&(K^2+\lambda K)\alpha^*=K\vec{y}.
\end{align}
Equation \eqref{Eq28} may be underdetermined and not permit a unique solution for \(\alpha^*\), or it may be overdetermined and not permit a solution at all. In either case, the best choice of \(\alpha^*\) is the minimum-norm solution to
\[\underset{\alpha^*\in\R^M}{\arg\min}\big\|(K^2+\lambda K)\alpha^*-K\vec{y}\big\|^2\]
which is given by
\begin{align}
\nonumber
\alpha^*=(K^2+\lambda K)^+K\vec{y}
\end{align}
where \((\cdot)^+\) is the Moore-Penrose pseudo-inverse (see Chapter 5.5 of~\cite{Golub2013Matrix}). Hence the optimal model is defined by
\begin{equation}
\nonumber
f_{\textrm{opt}}(x)=\sum_{i=1}^{M}\alpha_i\mathcal{K}(x,x_i)
\end{equation}
for all \(x\in\mathcal{X}\) with
\begin{equation}
\nonumber
\vec{\alpha}=(K^2+\lambda K)^+K\vec{y}
\end{equation}
which is \eqref{RKHSOptimalAlphaProof}. This concludes the proof.
\end{proof}

\onecolumngrid
\newgeometry{bottom=1.75cm,left=2cm,right=2cm}
\section{Convergence in the bond dimension}
\label{Convergence in the bond dimension}
\renewcommand\thefigure{B.\arabic{figure}} 
\renewcommand\thetable{B.\arabic{table}}
In this appendix, Figure \ref{ConvergenceInChiFigure} provides insight into how well the labels in the datasets produced in this work have converged with respect to the bond dimensions used to calculate them. 
These results are discussed in more detail in the second last paragraph of Sec. \ref{Section5}.
\vskip-0.25cm
\begin{figure}[H]
\centering
\begin{tikzpicture}[scale=0.9]
\node at (0,0) {\includegraphics[width=0.765\textwidth]{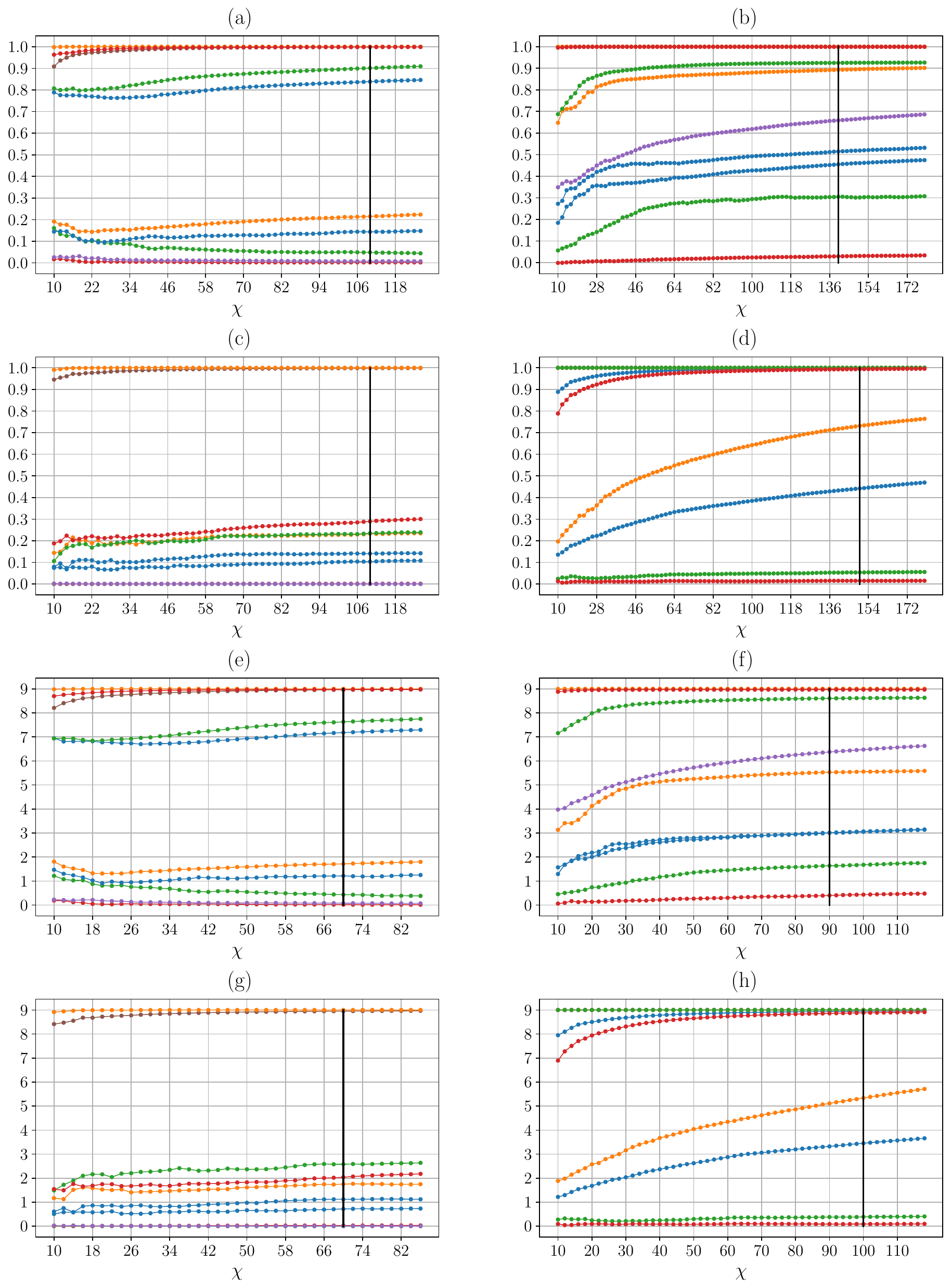}};
\node at (-7.825,7.7) {\scriptsize\(\mathscr{O}_{XZ}\)};
\node at (0.15,7.7) {\scriptsize\(\mathscr{O}_{XZ}\)};
\node at (-7.825,2.6375) {\scriptsize\(\mathscr{O}_{XZ}\)};
\node at (0.15,2.6375) {\scriptsize\(\mathscr{O}_{XZ}\)};
\node at (-7.75,-2.375) {\scriptsize\(\mathscr{O}_{Sum}\)};
\node at (0.2,-2.375) {\scriptsize\(\mathscr{O}_{Sum}\)};
\node at (-7.75,-7.45) {\scriptsize\(\mathscr{O}_{Sum}\)};
\node at (0.2,-7.45) {\scriptsize\(\mathscr{O}_{Sum}\)};
\end{tikzpicture}
\caption{The numerical values of \(\mathscr{O}_{XZ}\) and \(\mathscr{O}_{Sum}\) for ten randomly sampled inputs plotted against the maximum bond dimension \(\chi\) used to calculate them for 40-qubit quantum systems. 
Plots (a), (b), (c) and (d) show \(\mathscr{O}_{XZ}\) with \(\mathscr{H}_1\), \(\mathscr{H}_2\), \(\mathscr{H}_3\) and \(\mathscr{H}_4\), respectively. 
Plots (e), (f), (g) and (h) show \(\mathscr{O}_{Sum}\) with \(\mathscr{H}_1\), \(\mathscr{H}_2\), \(\mathscr{H}_3\) and \(\mathscr{H}_4\), respectively. 
The vertical black lines indicate the value of \(\chi\) used to calculate the rest of the data for the corresponding function and parameterised set of Hamiltonians. 
In the same order as the plots, the chosen values of \(\chi\) are 110, 140, 110, 150, 70, 90, 70 and 100.}
\label{ConvergenceInChiFigure}
\end{figure}

\section{Hyperparameter values}
\label{Tables of hyperparameter values}
\renewcommand\thefigure{C.\arabic{figure}} 
\renewcommand\thetable{C.\arabic{table}}
In this section of the appendices, Table \ref{HyperparameterValues} provides a list of the hyperparameter values which were trialled for each kernel during the 10-fold cross-validation. 
Of all the values which were trialled, Tables \ref{BestHyperparametersH1}, \ref{BestHyperparametersH2}, \ref{BestHyperparametersH3}, and \ref{BestHyperparametersH1} then provide the best hyperparameter values which achieved the highest average coefficient of determination over all ten folds of the training dataset. 
These hyperparameter values were then used to train the kernels on the full training datasets and subsequently apply the trained models to the testing data.
\begin{table}[H]
\begin{adjustbox}{center,max width=\textwidth}
\begin{tabular}{| c c c |}
\cline{1-3}
\multicolumn{1}{|c|}{\textbf{Kernel}}& \multicolumn{1}{c|}{\textbf{Hyperparameter}}
& \multicolumn{1}{c|}{\textbf{Cross-validated hyperparameter values}}\\
\cline{1-3}


\multirow{1}{*}{Linear}& \multicolumn{1}{|c|}{\(\lambda\)}&\footnotesize\(\{0,\,10^{-8},\,10^{-7},\,10^{-6},\,10^{-5},\,10^{-4},\,10^{-3},\,10^{-2},\,10^{-1},\,1,\,10,\,10^{2},\,10^{3},\,10^{4},\,10^{5}\}\)\\
\cline{1-3}


\multirow{4}{*}{Polynomial}
& \multicolumn{1}{|c|}{\(\lambda\)}&\footnotesize\(\{0,\,10^{-8},\,10^{-7},\,10^{-6},\,10^{-5},\,10^{-4},\,10^{-3},\,10^{-2},\,10^{-1},\,1,\,10,\,10^{2},\,10^{3},\,10^{4},\,10^{5}\}\)\\
& \multicolumn{1}{|c|}{\(\gamma\)}&\footnotesize\(\{10^{-3},\,10^{-2},\,10^{-1},\,1,\,10,\,10^{2},\,10^{3}\}\)\\
& \multicolumn{1}{|c|}{\(c_0\)}&\footnotesize\(\{10^{-3},\,10^{-2},\,10^{-1},\,1,\,10,\,10^{2},\,10^{3}\}\)\\
& \multicolumn{1}{|c|}{\(d\)}&\footnotesize\(\{1,\,2,\,3,\,4,\,5,\,6,\,7,\,8,\,9,\,10\}\)\\
\cline{1-3}


\multirow{2}{*}{RBF}
& \multicolumn{1}{|c|}{\(\lambda\)}&\footnotesize\(\{0,\,10^{-8},\,10^{-7},\,10^{-6},\,10^{-5},\,10^{-4},\,10^{-3},\,10^{-2},\,10^{-1},\,1,\,10,\,10^{2},\,10^{3},\,10^{4},\,10^{5}\}\)\\
& \multicolumn{1}{|c|}{\(\gamma\)}&\footnotesize\(\{10^{-3},\,10^{-2},\,10^{-1},\,1,\,10,\,10^{2},\,10^{3}\}\)\\
\cline{1-3}


\multirow{2}{*}{Laplacian}
& \multicolumn{1}{|c|}{\(\lambda\)}&\footnotesize\(\{0,\,10^{-8},\,10^{-7},\,10^{-6},\,10^{-5},\,10^{-4},\,10^{-3},\,10^{-2},\,10^{-1},\,1,\,10,\,10^{2},\,10^{3},\,10^{4},\,10^{5}\}\)\\
& \multicolumn{1}{|c|}{\(\gamma\)}&\footnotesize\(\{10^{-3},\,10^{-2},\,10^{-1},\,1,\,10,\,10^{2},\,10^{3}\}\)\\
\cline{1-3}


\multirow{3}{*}{Sigmoid}
& \multicolumn{1}{|c|}{\(\lambda\)}&\footnotesize\(\{0,\,10^{-8},\,10^{-7},\,10^{-6},\,10^{-5},\,10^{-4},\,10^{-3},\,10^{-2},\,10^{-1},\,1,\,10,\,10^{2},\,10^{3},\,10^{4},\,10^{5}\}\)\\
& \multicolumn{1}{|c|}{\(\gamma\)}&\footnotesize\(\{10^{-3},\,10^{-2},\,10^{-1},\,1,\,10,\,10^{2},\,10^{3}\}\)\\
& \multicolumn{1}{|c|}{\(c_0\)}&\footnotesize\(\{10^{-3},\,10^{-2},\,10^{-1},\,1,\,10,\,10^{2},\,10^{3}\}\)\\
\cline{1-3}


\multirow{1}{*}{Cosine}& \multicolumn{1}{|c|}{\(\lambda\)}&\footnotesize\(\{0,\,10^{-8},\,10^{-7},\,10^{-6},\,10^{-5},\,10^{-4},\,10^{-3},\,10^{-2},\,10^{-1},\,1,\,10,\,10^{2},\,10^{3},\,10^{4},\,10^{5}\}\)\\
\cline{1-3}

\end{tabular}
\end{adjustbox}
\caption{Lists of hyperparameters values which were trialled during the 10-fold cross-validation performed on the training datasets. The hyperparameter \(\lambda\) refers to the regularisation strength which appears in \eqref{RegularisedRisk}, while the hyperparameters \(\gamma\), \(c_0\) and \(d\) correspond to those which appear in definitions of the kernels in \eqref{LinearKernel}, \eqref{PolynomialKernel}, \eqref{RBFKernel}, \eqref{LaplacianKernel}, \eqref{SigmoidKernel} and \eqref{CosineKernel}.}
\label{HyperparameterValues}
\end{table}

\vskip-0.5cm

\begin{table}[H]
\centering
\begin{adjustbox}{center,max width=\textwidth}
\begin{tabular}{| c | c c c c c c c c c c |}

\cline{4-11}

\multicolumn{3}{c|}{}&\multicolumn{8}{c|}{\textbf{Number of qubits} \((n)\)}\\

\cline{1-11}

\multicolumn{1}{|c|}{\textbf{Function}}&\multicolumn{1}{c|}{\textbf{Kernel}}& \multicolumn{1}{c|}{\textbf{Hyperparameter}}
& \multicolumn{1}{c|}{5}& \multicolumn{1}{c|}{10}
& \multicolumn{1}{c|}{15}& \multicolumn{1}{c|}{20}
& \multicolumn{1}{c|}{25}& \multicolumn{1}{c|}{30}
& \multicolumn{1}{c|}{35}& \multicolumn{1}{c|}{40}\\

\cline{1-11}


\multirow{13}{*}{\(\mathscr{O}_{XZ}\)}&\multirow{1}{*}{Linear}& \multicolumn{1}{|c|}{\(\lambda\)}&\footnotesize\(10^3\)&\footnotesize\(10^4\)&\footnotesize\(10^4\)&\footnotesize\(10^5\)&\footnotesize\(10^5\)&\footnotesize\(10^5\)&\footnotesize\(10^5\)&\footnotesize\(10^5\)\\

\cline{2-11}


&\multirow{4}{*}{Polynomial}&\multicolumn{1}{|c|}{\(\lambda\)}&\footnotesize\(10^5\)&\footnotesize\(1\)&\footnotesize\(1\)&\footnotesize\(1\)&\footnotesize\(10^{-1}\)&\footnotesize\(1\)&\footnotesize\(10^4\)&\footnotesize\(10^5\)\\

&&\multicolumn{1}{|c|}{\(\gamma\)}&\footnotesize\(1\)&\footnotesize\(10^{-2}\)&\footnotesize\(10^{-2}\)&\footnotesize\(10^{-2}\)&\footnotesize\(10^{-2}\)&\footnotesize\(10^{-3}\)&\footnotesize\(10^{-2}\)&\footnotesize\(10^{-2}\)\\

&&\multicolumn{1}{|c|}{\(c_0\)}&\footnotesize\(10\)&\footnotesize\(1\)&\footnotesize\(1\)&\footnotesize\(1\)&\footnotesize\(1\)&\footnotesize\(1\)&\footnotesize\(10\)&\footnotesize\(10\)\\

&&\multicolumn{1}{|c|}{\(d\)}&\footnotesize\(6\)&\footnotesize\(9\)&\footnotesize\(6\)&\footnotesize\(6\)&\footnotesize\(6\)&\footnotesize\(10\)&\footnotesize\(6\)&\footnotesize\(6\)\\

\cline{2-11}


&\multirow{2}{*}{RBF}&\multicolumn{1}{|c|}{\footnotesize \(\lambda\)}&\footnotesize\(10^{-2}\)&\footnotesize\(10^{-2}\)&\footnotesize\(10^{-2}\)&\footnotesize\(10^{-2}\)&\footnotesize\(10^{-2}\)&\footnotesize\(10^{-1}\)&\footnotesize\(10^{-1}\)&\footnotesize\(10^{-1}\)\\

&&\multicolumn{1}{|c|}{\(\gamma\)}&\footnotesize\(1\)&\footnotesize\(10^{-1}\)&\footnotesize\(10^{-1}\)&\footnotesize\(10^{-2}\)&\footnotesize\(10^{-2}\)&\footnotesize\(10^{-2}\)&\footnotesize\(10^{-2}\)&\footnotesize\(10^{-2}\)\\

\cline{2-11}


&\multirow{2}{*}{Laplacian}& \multicolumn{1}{|c|}{\(\lambda\)}&\footnotesize\(0\)&\footnotesize\(0\)&\footnotesize\(0\)&\footnotesize\(0\)&\footnotesize\(0\)&\footnotesize\(10^{-3}\)&\footnotesize\(10^{-3}\)&\footnotesize\(10^{-3}\)\\

&&\multicolumn{1}{|c|}{\(\gamma\)}&\footnotesize\(10^{-1}\)&\footnotesize\(10^{-1}\)&\footnotesize\(10^{-1}\)&\footnotesize\(10^{-1}\)&\footnotesize\(10^{-1}\)&\footnotesize\(10^{-2}\)&\footnotesize\(10^{-2}\)&\footnotesize\(10^{-2}\)\\

\cline{2-11}


&\multirow{3}{*}{Sigmoid}& \multicolumn{1}{|c|}{\(\lambda\)}&\footnotesize\(10^{-4}\)&\footnotesize\(10^{-6}\)&\footnotesize\(10^{-5}\)&\footnotesize\(0\)&\footnotesize\(10^{-5}\)&\footnotesize\(10^{-8}\)&\footnotesize\(10^{-1}\)&\footnotesize\(1\)\\

&&\multicolumn{1}{|c|}{\(\gamma\)}&\footnotesize\(10^{-2}\)&\footnotesize\(10^{-2}\)&\footnotesize\(10^{-3}\)&\footnotesize\(10^{-3}\)&\footnotesize\(10^{-3}\)&\footnotesize\(10^{-3}\)&\footnotesize\(10^{-3}\)&\footnotesize\(10^{-3}\)\\

&&\multicolumn{1}{|c|}{\(c_0\)}&\footnotesize\(10^{-1}\)&\footnotesize\(1\)&\footnotesize\(1\)&\footnotesize\(10^{-1}\)&\footnotesize\(1\)&\footnotesize\(1\)&\footnotesize\(10^{-1}\)&\footnotesize\(1\)\\

\cline{2-11}


&\multirow{1}{*}{Cosine}& \multicolumn{1}{|c|}{\(\lambda\)}&\footnotesize\(10^{2}\)&\footnotesize\(10^{3}\)&\footnotesize\(10^{3}\)&\footnotesize\(10^{3}\)&\footnotesize\(10^{3}\)&\footnotesize\(10^{3}\)&\footnotesize\(10^{2}\)&\footnotesize\(10^{2}\)\\

\cline{1-11}


\multirow{13}{*}{\(\mathscr{O}_{Sum}\)}&\multirow{1}{*}{Linear}& \multicolumn{1}{|c|}{\(\lambda\)}&\footnotesize\(10^3\)&\footnotesize\(10^4\)&\footnotesize\(10^4\)&\footnotesize\(10^4\)&\footnotesize\(10^5\)&\footnotesize\(10^5\)&\footnotesize\(10^5\)&\footnotesize\(10^5\)\\

\cline{2-11}


&\multirow{4}{*}{Polynomial}&\multicolumn{1}{|c|}{\(\lambda\)}&\footnotesize\(1\)&\footnotesize\(10^{-3}\)&\footnotesize\(1\)&\footnotesize\(10^{-2}\)&\footnotesize\(10^{-1}\)&\footnotesize\(1\)&\footnotesize\(1\)&\footnotesize\(1\)\\

&&\multicolumn{1}{|c|}{\(\gamma\)}&\footnotesize\(10^{-1}\)&\footnotesize\(10^{-2}\)&\footnotesize\(10^{-2}\)&\footnotesize\(10^{-3}\)&\footnotesize\(10^{-3}\)&\footnotesize\(10^{-3}\)&\footnotesize\(10^{-3}\)&\footnotesize\(10^{-3}\)\\

&&\multicolumn{1}{|c|}{\(c_0\)}&\footnotesize\(1\)&\footnotesize\(1\)&\footnotesize\(1\)&\footnotesize\(1\)&\footnotesize\(1\)&\footnotesize\(1\)&\footnotesize\(1\)&\footnotesize\(1\)\\

&&\multicolumn{1}{|c|}{\(d\)}&\footnotesize\(8\)&\footnotesize\(8\)&\footnotesize\(8\)&\footnotesize\(10\)&\footnotesize\(10\)&\footnotesize\(10\)&\footnotesize\(9\)&\footnotesize\(8\)\\

\cline{2-11}


&\multirow{2}{*}{RBF}&\multicolumn{1}{|c|}{\footnotesize \(\lambda\)}&\footnotesize\(10^{-3}\)&\footnotesize\(10^{-1}\)&\footnotesize\(10^{-2}\)&\footnotesize\(10^{-2}\)&\footnotesize\(10^{-2}\)&\footnotesize\(10^{-1}\)&\footnotesize\(10^{-1}\)&\footnotesize\(10^{-1}\)\\

&&\multicolumn{1}{|c|}{\(\gamma\)}&\footnotesize\(10^{-1}\)&\footnotesize\(10^{-1}\)&\footnotesize\(10^{-1}\)&\footnotesize\(10^{-2}\)&\footnotesize\(10^{-2}\)&\footnotesize\(10^{-2}\)&\footnotesize\(10^{-2}\)&\footnotesize\(10^{-2}\)\\

\cline{2-11}


&\multirow{2}{*}{Laplacian}& \multicolumn{1}{|c|}{\(\lambda\)}&\footnotesize\(10^{-8}\)&\footnotesize\(0\)&\footnotesize\(0\)&\footnotesize\(0\)&\footnotesize\(0\)&\footnotesize\(10^{-3}\)&\footnotesize\(10^{-3}\)&\footnotesize\(10^{-3}\)\\

&&\multicolumn{1}{|c|}{\(\gamma\)}&\footnotesize\(10^{-1}\)&\footnotesize\(10^{-1}\)&\footnotesize\(10^{-1}\)&\footnotesize\(10^{-1}\)&\footnotesize\(10^{-1}\)&\footnotesize\(10^{-2}\)&\footnotesize\(10^{-2}\)&\footnotesize\(10^{-2}\)\\

\cline{2-11}


&\multirow{3}{*}{Sigmoid}& \multicolumn{1}{|c|}{\(\lambda\)}&\footnotesize\(1\)&\footnotesize\(10^{-6}\)&\footnotesize\(10^{-5}\)&\footnotesize\(0\)&\footnotesize\(10^{-5}\)&\footnotesize\(10^{-6}\)&\footnotesize\(10^{-1}\)&\footnotesize\(1\)\\

&&\multicolumn{1}{|c|}{\(\gamma\)}&\footnotesize\(10^{-1}\)&\footnotesize\(10^{-2}\)&\footnotesize\(10^{-3}\)&\footnotesize\(10^{-3}\)&\footnotesize\(10^{-3}\)&\footnotesize\(10^{-3}\)&\footnotesize\(10^{-3}\)&\footnotesize\(10^{-3}\)\\

&&\multicolumn{1}{|c|}{\(c_0\)}&\footnotesize\(10^{-1}\)&\footnotesize\(1\)&\footnotesize\(1\)&\footnotesize\(10^{-1}\)&\footnotesize\(1\)&\footnotesize\(1\)&\footnotesize\(10^{-1}\)&\footnotesize\(1\)\\

\cline{2-11}


&\multirow{1}{*}{Cosine}& \multicolumn{1}{|c|}{\(\lambda\)}&\footnotesize\(10^2\)&\footnotesize\(10^3\)&\footnotesize\(10^3\)&\footnotesize\(10^3\)&\footnotesize\(10^3\)&\footnotesize\(10^2\)&\footnotesize\(10^2\)&\footnotesize\(10^2\)\\

\cline{1-11}

\end{tabular}
\end{adjustbox}
\caption{The best hyperparameters trialled during the 10-fold cross-validation performed on the training datasets for \(\mathscr{H}_1\). Of all the hyperparameters trialled, the ones listed here resulted in models which achieved the highest average coefficient of determination (\(R^2\)) over all ten folds of the training dataset.}
\label{BestHyperparametersH1}
\end{table}


\begin{table}[H]
\begin{adjustbox}{center,max width=\textwidth}
\begin{tabular}{| c | c c c c c c c c c c |}

\cline{4-11}

\multicolumn{3}{c|}{}&\multicolumn{8}{c|}{\textbf{Number of qubits} \((n)\)}\\

\cline{1-11}

\multicolumn{1}{|c|}{\textbf{Function}}&\multicolumn{1}{c|}{\textbf{Kernel}}& \multicolumn{1}{c|}{\textbf{Hyperparameter}}
& \multicolumn{1}{c|}{5}& \multicolumn{1}{c|}{10}
& \multicolumn{1}{c|}{15}& \multicolumn{1}{c|}{20}
& \multicolumn{1}{c|}{25}& \multicolumn{1}{c|}{30}
& \multicolumn{1}{c|}{35}& \multicolumn{1}{c|}{40}\\

\cline{1-11}


\multirow{13}{*}{\(\mathscr{O}_{XZ}\)}&\multirow{1}{*}{Linear}& \multicolumn{1}{|c|}{\(\lambda\)}&\footnotesize\(10^4\)&\footnotesize\(10^4\)&\footnotesize\(10^5\)&\footnotesize\(10^5\)&\footnotesize\(10^5\)&\footnotesize\(10^5\)&\footnotesize\(10^5\)&\footnotesize\(10^5\)\\

\cline{2-11}


&\multirow{4}{*}{Polynomial}&\multicolumn{1}{|c|}{\(\lambda\)}&\footnotesize\(10^{5}\)&\footnotesize\(10^{-1}\)&\footnotesize\(1\)&\footnotesize\(10^{-2}\)&\footnotesize\(10^{-1}\)&\footnotesize\(10^{-1}\)&\footnotesize\(10^{5}\)&\footnotesize\(10^{5}\)\\

&&\multicolumn{1}{|c|}{\(\gamma\)}&\footnotesize\(1\)&\footnotesize\(10^{-2}\)&\footnotesize\(10^{-2}\)&\footnotesize\(10^{-3}\)&\footnotesize\(10^{-3}\)&\footnotesize\(10^{-3}\)&\footnotesize\(10^{-2}\)&\footnotesize\(10^{-2}\)\\

&&\multicolumn{1}{|c|}{\(c_0\)}&\footnotesize\(10\)&\footnotesize\(1\)&\footnotesize\(1\)&\footnotesize\(1\)&\footnotesize\(1\)&\footnotesize\(1\)&\footnotesize\(10\)&\footnotesize\(10\)\\

&&\multicolumn{1}{|c|}{\(d\)}&\footnotesize\(6\)&\footnotesize\(7\)&\footnotesize\(6\)&\footnotesize\(10\)&\footnotesize\(9\)&\footnotesize\(7\)&\footnotesize\(6\)&\footnotesize\(6\)\\

\cline{2-11}


&\multirow{2}{*}{RBF}&\multicolumn{1}{|c|}{\footnotesize \(\lambda\)}&\footnotesize\(10^{-3}\)&\footnotesize\(10^{-1}\)&\footnotesize\(10^{-2}\)&\footnotesize\(10^{-2}\)&\footnotesize\(10^{-1}\)&\footnotesize\(10^{-1}\)&\footnotesize\(10^{-1}\)&\footnotesize\(10^{-1}\)\\

&&\multicolumn{1}{|c|}{\(\gamma\)}&\footnotesize\(1\)&\footnotesize\(10^{-1}\)&\footnotesize\(10^{-1}\)&\footnotesize\(10^{-2}\)&\footnotesize\(10^{-2}\)&\footnotesize\(10^{-2}\)&\footnotesize\(10^{-2}\)&\footnotesize\(10^{-2}\)\\

\cline{2-11}


&\multirow{2}{*}{Laplacian}& \multicolumn{1}{|c|}{\(\lambda\)}&\footnotesize\(0\)&\footnotesize\(10^{-1}\)&\footnotesize\(10^{-2}\)&\footnotesize\(10^{-2}\)&\footnotesize\(0\)&\footnotesize\(0\)&\footnotesize\(0\)&\footnotesize\(10^{-2}\)\\

&&\multicolumn{1}{|c|}{\(\gamma\)}&\footnotesize\(1\)&\footnotesize\(10^{-1}\)&\footnotesize\(10^{-1}\)&\footnotesize\(10^{-1}\)&\footnotesize\(10^{-1}\)&\footnotesize\(10^{-1}\)&\footnotesize\(10^{-1}\)&\footnotesize\(10^{-2}\)\\

\cline{2-11}


&\multirow{3}{*}{Sigmoid}& \multicolumn{1}{|c|}{\(\lambda\)}&\footnotesize\(1\)&\footnotesize\(10^{-8}\)&\footnotesize\(10^{-4}\)&\footnotesize\(10^{-8}\)&\footnotesize\(10^{-5}\)&\footnotesize\(10^{-6}\)&\footnotesize\(10^{-2}\)&\footnotesize\(1\)\\

&&\multicolumn{1}{|c|}{\(\gamma\)}&\footnotesize\(10^{-1}\)&\footnotesize\(10^{-2}\)&\footnotesize\(10^{-3}\)&\footnotesize\(10^{-3}\)&\footnotesize\(10^{-3}\)&\footnotesize\(10^{-3}\)&\footnotesize\(10^{-3}\)&\footnotesize\(10^{-3}\)\\

&&\multicolumn{1}{|c|}{\(c_0\)}&\footnotesize\(10^{-1}\)&\footnotesize\(1\)&\footnotesize\(10^{-1}\)&\footnotesize\(10^{-1}\)&\footnotesize\(1\)&\footnotesize\(1\)&\footnotesize\(1\)&\footnotesize\(1\)\\

\cline{2-11}


&\multirow{1}{*}{Cosine}& \multicolumn{1}{|c|}{\(\lambda\)}&\footnotesize\(10^{3}\)&\footnotesize\(10^{3}\)&\footnotesize\(10^{3}\)&\footnotesize\(10^{4}\)&\footnotesize\(10^{4}\)&\footnotesize\(10^{4}\)&\footnotesize\(10^{4}\)&\footnotesize\(10^{5}\)\\

\cline{1-11}


\multirow{13}{*}{\(\mathscr{O}_{Sum}\)}&\multirow{1}{*}{Linear}& \multicolumn{1}{|c|}{\(\lambda\)}&\footnotesize\(10^4\)&\footnotesize\(10^4\)&\footnotesize\(10^5\)&\footnotesize\(10^5\)&\footnotesize\(10^5\)&\footnotesize\(10^5\)&\footnotesize\(10^5\)&\footnotesize\(10^5\)\\

\cline{2-11}


&\multirow{4}{*}{Polynomial}&\multicolumn{1}{|c|}{\(\lambda\)}&\footnotesize\(1\)&\footnotesize\(10^5\)&\footnotesize\(10^4\)&\footnotesize\(10^{-4}\)&\footnotesize\(10^{-2}\)&\footnotesize\(10^{-2}\)&\footnotesize\(10^{4}\)&\footnotesize\(10^{3}\)\\

&&\multicolumn{1}{|c|}{\(\gamma\)}&\footnotesize\(10^{-1}\)&\footnotesize\(10^{-1}\)&\footnotesize\(10^{-2}\)&\footnotesize\(10^{-3}\)&\footnotesize\(10^{-3}\)&\footnotesize\(10^{-3}\)&\footnotesize\(10^{-3}\)&\footnotesize\(10^{-3}\)\\

&&\multicolumn{1}{|c|}{\(c_0\)}&\footnotesize\(1\)&\footnotesize\(10\)&\footnotesize\(10\)&\footnotesize\(1\)&\footnotesize\(1\)&\footnotesize\(1\)&\footnotesize\(10\)&\footnotesize\(10\)\\

&&\multicolumn{1}{|c|}{\(d\)}&\footnotesize\(8\)&\footnotesize\(8\)&\footnotesize\(8\)&\footnotesize\(6\)&\footnotesize\(7\)&\footnotesize\(6\)&\footnotesize\(9\)&\footnotesize\(8\)\\

\cline{2-11}


&\multirow{2}{*}{RBF}&\multicolumn{1}{|c|}{\footnotesize \(\lambda\)}&\footnotesize\(10^{-3}\)&\footnotesize\(10^{-1}\)&\footnotesize\(10^{-2}\)&\footnotesize\(10^{-2}\)&\footnotesize\(10^{-1}\)&\footnotesize\(10^{-1}\)&\footnotesize\(10^{-1}\)&\footnotesize\(10^{-1}\)\\

&&\multicolumn{1}{|c|}{\(\gamma\)}&\footnotesize\(1\)&\footnotesize\(10^{-1}\)&\footnotesize\(10^{-1}\)&\footnotesize\(10^{-2}\)&\footnotesize\(10^{-2}\)&\footnotesize\(10^{-2}\)&\footnotesize\(10^{-2}\)&\footnotesize\(10^{-2}\)\\

\cline{2-11}


&\multirow{2}{*}{Laplacian}& \multicolumn{1}{|c|}{\(\lambda\)}&\footnotesize\(10^{-1}\)&\footnotesize\(10^{-1}\)&\footnotesize\(10^{-2}\)&\footnotesize\(10^{-2}\)&\footnotesize\(0\)&\footnotesize\(10^{-2}\)&\footnotesize\(10^{-2}\)&\footnotesize\(10^{-2}\)\\

&&\multicolumn{1}{|c|}{\(\gamma\)}&\footnotesize\(10^{-1}\)&\footnotesize\(10^{-1}\)&\footnotesize\(10^{-1}\)&\footnotesize\(10^{-1}\)&\footnotesize\(10^{-1}\)&\footnotesize\(10^{-2}\)&\footnotesize\(10^{-2}\)&\footnotesize\(10^{-2}\)\\

\cline{2-11}


&\multirow{3}{*}{Sigmoid}& \multicolumn{1}{|c|}{\(\lambda\)}&\footnotesize\(0\)&\footnotesize\(10^{-7}\)&\footnotesize\(10^{-5}\)&\footnotesize\(10^{-8}\)&\footnotesize\(10^{-5}\)&\footnotesize\(10^{-6}\)&\footnotesize\(10^{-2}\)&\footnotesize\(1\)\\

&&\multicolumn{1}{|c|}{\(\gamma\)}&\footnotesize\(10^{-2}\)&\footnotesize\(10^{-2}\)&\footnotesize\(10^{-3}\)&\footnotesize\(10^{-3}\)&\footnotesize\(10^{-3}\)&\footnotesize\(10^{-3}\)&\footnotesize\(10^{-3}\)&\footnotesize\(10^{-3}\)\\

&&\multicolumn{1}{|c|}{\(c_0\)}&\footnotesize\(10^{-1}\)&\footnotesize\(1\)&\footnotesize\(1\)&\footnotesize\(10^{-1}\)&\footnotesize\(1\)&\footnotesize\(1\)&\footnotesize\(1\)&\footnotesize\(1\)\\

\cline{2-11}


&\multirow{1}{*}{Cosine}& \multicolumn{1}{|c|}{\(\lambda\)}&\footnotesize\(10^{3}\)&\footnotesize\(10^{3}\)&\footnotesize\(10^{4}\)&\footnotesize\(10^{4}\)&\footnotesize\(10^{4}\)&\footnotesize\(10^{5}\)&\footnotesize\(10^{4}\)&\footnotesize\(10^{5}\)\\

\cline{1-11}

\end{tabular}
\end{adjustbox}
\caption{The best hyperparameters trialled during the 10-fold cross-validation performed on the training datasets for \(\mathscr{H}_2\). Of all the hyperparameters trialled, the ones listed here resulted in models which achieved the highest average coefficient of determination (\(R^2\)) over all ten folds of the training dataset.}
\label{BestHyperparametersH2}
\end{table}


\begin{table}[H]
\begin{adjustbox}{center,max width=\textwidth}
\begin{tabular}{| c | c c c c c c c c c c |}

\cline{4-11}

\multicolumn{3}{c|}{}&\multicolumn{8}{c|}{\textbf{Number of qubits} \((n)\)}\\

\cline{1-11}

\multicolumn{1}{|c|}{\textbf{Function}}&\multicolumn{1}{c|}{\textbf{Kernel}}& \multicolumn{1}{c|}{\textbf{Hyperparameter}}
& \multicolumn{1}{c|}{5}& \multicolumn{1}{c|}{10}
& \multicolumn{1}{c|}{15}& \multicolumn{1}{c|}{20}
& \multicolumn{1}{c|}{25}& \multicolumn{1}{c|}{30}
& \multicolumn{1}{c|}{35}& \multicolumn{1}{c|}{40}\\

\cline{1-11}


\multirow{13}{*}{\(\mathscr{O}_{XZ}\)}&\multirow{1}{*}{Linear}& \multicolumn{1}{|c|}{\(\lambda\)}&\footnotesize\(10^4\)&\footnotesize\(10^5\)&\footnotesize\(10^5\)&\footnotesize\(10^5\)&\footnotesize\(10^5\)&\footnotesize\(10^5\)&\footnotesize\(10^5\)&\footnotesize\(10^5\)\\

\cline{2-11}


&\multirow{4}{*}{Polynomial}&\multicolumn{1}{|c|}{\(\lambda\)}&\footnotesize\(10^5\)&\footnotesize\(10^{-1}\)&\footnotesize\(1\)&\footnotesize\(1\)&\footnotesize\(10^{-1}\)&\footnotesize\(1\)&\footnotesize\(1\)&\footnotesize\(1\)\\

&&\multicolumn{1}{|c|}{\(\gamma\)}&\footnotesize\(1\)&\footnotesize\(10^{-2}\)&\footnotesize\(10^{-2}\)&\footnotesize\(10^{-2}\)&\footnotesize\(10^{-3}\)&\footnotesize\(10^{-3}\)&\footnotesize\(10^{-3}\)&\footnotesize\(10^{-3}\)\\

&&\multicolumn{1}{|c|}{\(c_0\)}&\footnotesize\(10\)&\footnotesize\(1\)&\footnotesize\(1\)&\footnotesize\(1\)&\footnotesize\(1\)&\footnotesize\(1\)&\footnotesize\(1\)&\footnotesize\(1\)\\

&&\multicolumn{1}{|c|}{\(d\)}&\footnotesize\(6\)&\footnotesize\(7\)&\footnotesize\(6\)&\footnotesize\(6\)&\footnotesize\(10\)&\footnotesize\(10\)&\footnotesize\(9\)&\footnotesize\(8\)\\

\cline{2-11}


&\multirow{2}{*}{RBF}&\multicolumn{1}{|c|}{\footnotesize \(\lambda\)}&\footnotesize\(10^{-2}\)&\footnotesize\(10^{-2}\)&\footnotesize\(10^{-2}\)&\footnotesize\(10^{-3}\)&\footnotesize\(10^{-2}\)&\footnotesize\(10^{-2}\)&\footnotesize\(10^{-2}\)&\footnotesize\(10^{-2}\)\\

&&\multicolumn{1}{|c|}{\(\gamma\)}&\footnotesize\(1\)&\footnotesize\(10^{-1}\)&\footnotesize\(10^{-1}\)&\footnotesize\(10^{-2}\)&\footnotesize\(10^{-2}\)&\footnotesize\(10^{-2}\)&\footnotesize\(10^{-2}\)&\footnotesize\(10^{-2}\)\\

\cline{2-11}


&\multirow{2}{*}{Laplacian}& \multicolumn{1}{|c|}{\(\lambda\)}&\footnotesize\(0\)&\footnotesize\(0\)&\footnotesize\(0\)&\footnotesize\(0\)&\footnotesize\(0\)&\footnotesize\(0\)&\footnotesize\(0\)&\footnotesize\(0\)\\

&&\multicolumn{1}{|c|}{\(\gamma\)}&\footnotesize\(1\)&\footnotesize\(10^{-1}\)&\footnotesize\(10^{-1}\)&\footnotesize\(10^{-1}\)&\footnotesize\(10^{-1}\)&\footnotesize\(10^{-1}\)&\footnotesize\(10^{-1}\)&\footnotesize\(10^{-1}\)\\

\cline{2-11}


&\multirow{3}{*}{Sigmoid}& \multicolumn{1}{|c|}{\(\lambda\)}&\footnotesize\(1\)&\footnotesize\(10^{-6}\)&\footnotesize\(10^{-5}\)&\footnotesize\(10^{-8}\)&\footnotesize\(10^{-5}\)&\footnotesize\(0\)&\footnotesize\(10^{-6}\)&\footnotesize\(1\)\\

&&\multicolumn{1}{|c|}{\(\gamma\)}&\footnotesize\(10^{-1}\)&\footnotesize\(10^{-2}\)&\footnotesize\(10^{-3}\)&\footnotesize\(10^{-3}\)&\footnotesize\(10^{-3}\)&\footnotesize\(10^{-3}\)&\footnotesize\(10^{-3}\)&\footnotesize\(10^{-3}\)\\

&&\multicolumn{1}{|c|}{\(c_0\)}&\footnotesize\(10^{-1}\)&\footnotesize\(1\)&\footnotesize\(1\)&\footnotesize\(10^{-1}\)&\footnotesize\(1\)&\footnotesize\(1\)&\footnotesize\(1\)&\footnotesize\(1\)\\

\cline{2-11}


&\multirow{1}{*}{Cosine}& \multicolumn{1}{|c|}{\(\lambda\)}&\footnotesize\(10^{3}\)&\footnotesize\(10^{3}\)&\footnotesize\(10^{4}\)&\footnotesize\(10^{4}\)&\footnotesize\(10^{4}\)&\footnotesize\(10^{4}\)&\footnotesize\(10^{4}\)&\footnotesize\(10^{4}\)\\

\cline{1-11}


\multirow{13}{*}{\(\mathscr{O}_{Sum}\)}&\multirow{1}{*}{Linear}& \multicolumn{1}{|c|}{\(\lambda\)}&\footnotesize\(10^4\)&\footnotesize\(10^5\)&\footnotesize\(10^5\)&\footnotesize\(10^5\)&\footnotesize\(10^5\)&\footnotesize\(10^5\)&\footnotesize\(10^5\)&\footnotesize\(10^5\)\\

\cline{2-11}


&\multirow{4}{*}{Polynomial}&\multicolumn{1}{|c|}{\(\lambda\)}&\footnotesize\(1\)&\footnotesize\(10^{5}\)&\footnotesize\(1\)&\footnotesize\(10\)&\footnotesize\(10^{5}\)&\footnotesize\(10^{5}\)&\footnotesize\(10^{-2}\)&\footnotesize\(10^{-2}\)\\

&&\multicolumn{1}{|c|}{\(\gamma\)}&\footnotesize\(10^{-1}\)&\footnotesize\(10^{-1}\)&\footnotesize\(10^{-2}\)&\footnotesize\(10^{-2}\)&\footnotesize\(10^{-2}\)&\footnotesize\(10^{-2}\)&\footnotesize\(10^{-3}\)&\footnotesize\(10^{-3}\)\\

&&\multicolumn{1}{|c|}{\(c_0\)}&\footnotesize\(1\)&\footnotesize\(10\)&\footnotesize\(1\)&\footnotesize\(1\)&\footnotesize\(10\)&\footnotesize\(10\)&\footnotesize\(1\)&\footnotesize\(1\)\\

&&\multicolumn{1}{|c|}{\(d\)}&\footnotesize\(8\)&\footnotesize\(8\)&\footnotesize\(8\)&\footnotesize\(8\)&\footnotesize\(9\)&\footnotesize\(8\)&\footnotesize\(8\)&\footnotesize\(8\)\\

\cline{2-11}


&\multirow{2}{*}{RBF}&\multicolumn{1}{|c|}{\footnotesize \(\lambda\)}&\footnotesize\(10^{-3}\)&\footnotesize\(10^{-2}\)&\footnotesize\(10^{-2}\)&\footnotesize\(10^{-3}\)&\footnotesize\(10^{-2}\)&\footnotesize\(10^{-2}\)&\footnotesize\(10^{-2}\)&\footnotesize\(10^{-2}\)\\

&&\multicolumn{1}{|c|}{\(\gamma\)}&\footnotesize\(1\)&\footnotesize\(10^{-1}\)&\footnotesize\(10^{-1}\)&\footnotesize\(10^{-2}\)&\footnotesize\(10^{-2}\)&\footnotesize\(10^{-2}\)&\footnotesize\(10^{-2}\)&\footnotesize\(10^{-2}\)\\

\cline{2-11}


&\multirow{2}{*}{Laplacian}& \multicolumn{1}{|c|}{\(\lambda\)}&\footnotesize\(0\)&\footnotesize\(0\)&\footnotesize\(0\)&\footnotesize\(0\)&\footnotesize\(0\)&\footnotesize\(0\)&\footnotesize\(0\)&\footnotesize\(0\)\\

&&\multicolumn{1}{|c|}{\(\gamma\)}&\footnotesize\(1\)&\footnotesize\(10^{-1}\)&\footnotesize\(10^{-1}\)&\footnotesize\(10^{-1}\)&\footnotesize\(10^{-1}\)&\footnotesize\(10^{-1}\)&\footnotesize\(10^{-1}\)&\footnotesize\(10^{-1}\)\\

\cline{2-11}


&\multirow{3}{*}{Sigmoid}& \multicolumn{1}{|c|}{\(\lambda\)}&\footnotesize\(1\)&\footnotesize\(10^{-6}\)&\footnotesize\(10^{-5}\)&\footnotesize\(10^{-8}\)&\footnotesize\(10^{-5}\)&\footnotesize\(0\)&\footnotesize\(10^{-6}\)&\footnotesize\(1\)\\

&&\multicolumn{1}{|c|}{\(\gamma\)}&\footnotesize\(10^{-1}\)&\footnotesize\(10^{-2}\)&\footnotesize\(10^{-3}\)&\footnotesize\(10^{-3}\)&\footnotesize\(10^{-3}\)&\footnotesize\(10^{-3}\)&\footnotesize\(10^{-3}\)&\footnotesize\(10^{-3}\)\\

&&\multicolumn{1}{|c|}{\(c_0\)}&\footnotesize\(10^{-1}\)&\footnotesize\(1\)&\footnotesize\(1\)&\footnotesize\(10^{-1}\)&\footnotesize\(1\)&\footnotesize\(1\)&\footnotesize\(1\)&\footnotesize\(1\)\\

\cline{2-11}


&\multirow{1}{*}{Cosine}& \multicolumn{1}{|c|}{\(\lambda\)}&\footnotesize\(10^{4}\)&\footnotesize\(10^{5}\)&\footnotesize\(10^{5}\)&\footnotesize\(10^{5}\)&\footnotesize\(10^{4}\)&\footnotesize\(10^{4}\)&\footnotesize\(10^{4}\)&\footnotesize\(10^{4}\)\\

\cline{1-11}

\end{tabular}
\end{adjustbox}
\caption{The best hyperparameters trialled during the 10-fold cross-validation performed on the training datasets for \(\mathscr{H}_3\). Of all the hyperparameters trialled, the ones listed here resulted in models which achieved the highest average coefficient of determination (\(R^2\)) over all ten folds of the training dataset.}
\label{BestHyperparametersH3}
\end{table}


\begin{table}[H]
\begin{adjustbox}{center,max width=\textwidth}
\begin{tabular}{| c | c c c c c c c c c c |}

\cline{4-11}

\multicolumn{3}{c|}{}&\multicolumn{8}{c|}{\textbf{Number of qubits} \((n)\)}\\

\cline{1-11}

\multicolumn{1}{|c|}{\textbf{Function}}&\multicolumn{1}{c|}{\textbf{Kernel}}& \multicolumn{1}{c|}{\textbf{Hyperparameter}}
& \multicolumn{1}{c|}{5}& \multicolumn{1}{c|}{10}
& \multicolumn{1}{c|}{15}& \multicolumn{1}{c|}{20}
& \multicolumn{1}{c|}{25}& \multicolumn{1}{c|}{30}
& \multicolumn{1}{c|}{35}& \multicolumn{1}{c|}{40}\\

\cline{1-11}


\multirow{13}{*}{\(\mathscr{O}_{XZ}\)}&\multirow{1}{*}{Linear}& \multicolumn{1}{|c|}{\(\lambda\)}&\footnotesize\(10^4\)&\footnotesize\(10^4\)&\footnotesize\(10^5\)&\footnotesize\(10^5\)&\footnotesize\(10^5\)&\footnotesize\(10^5\)&\footnotesize\(10^5\)&\footnotesize\(10^5\)\\

\cline{2-11}


&\multirow{4}{*}{Polynomial}&\multicolumn{1}{|c|}{\(\lambda\)}&\footnotesize\(10^{-2}\)&\footnotesize\(10^{-2}\)&\footnotesize\(1\)&\footnotesize\(10\)&\footnotesize\(10^{5}\)&\footnotesize\(10\)&\footnotesize\(10^{5}\)&\footnotesize\(10^{5}\)\\

&&\multicolumn{1}{|c|}{\(\gamma\)}&\footnotesize\(10^{-2}\)&\footnotesize\(10^{-3}\)&\footnotesize\(10^{-3}\)&\footnotesize\(10^{-3}\)&\footnotesize\(10^{-2}\)&\footnotesize\(10^{-3}\)&\footnotesize\(10^{-3}\)&\footnotesize\(10^{-3}\)\\

&&\multicolumn{1}{|c|}{\(c_0\)}&\footnotesize\(1\)&\footnotesize\(1\)&\footnotesize\(1\)&\footnotesize\(1\)&\footnotesize\(10\)&\footnotesize\(1\)&\footnotesize\(10\)&\footnotesize\(10\)\\

&&\multicolumn{1}{|c|}{\(d\)}&\footnotesize\(6\)&\footnotesize\(10\)&\footnotesize\(9\)&\footnotesize\(9\)&\footnotesize\(6\)&\footnotesize\(6\)&\footnotesize\(9\)&\footnotesize\(8\)\\

\cline{2-11}


&\multirow{2}{*}{RBF}&\multicolumn{1}{|c|}{\footnotesize \(\lambda\)}&\footnotesize\(10^{-1}\)&\footnotesize\(10^{-2}\)&\footnotesize\(10^{-2}\)&\footnotesize\(10^{-2}\)&\footnotesize\(10^{-2}\)&\footnotesize\(10^{-2}\)&\footnotesize\(10^{-2}\)&\footnotesize\(10^{-2}\)\\

&&\multicolumn{1}{|c|}{\(\gamma\)}&\footnotesize\(10^{-1}\)&\footnotesize\(10^{-2}\)&\footnotesize\(10^{-2}\)&\footnotesize\(10^{-2}\)&\footnotesize\(10^{-3}\)&\footnotesize\(10^{-3}\)&\footnotesize\(10^{-3}\)&\footnotesize\(10^{-3}\)\\

\cline{2-11}


&\multirow{2}{*}{Laplacian}& \multicolumn{1}{|c|}{\(\lambda\)}&\footnotesize\(0\)&\footnotesize\(0\)&\footnotesize\(10^{-3}\)&\footnotesize\(10^{-3}\)&\footnotesize\(10^{-3}\)&\footnotesize\(10^{-3}\)&\footnotesize\(10^{-3}\)&\footnotesize\(10^{-4}\)\\

&&\multicolumn{1}{|c|}{\(\gamma\)}&\footnotesize\(10^{-1}\)&\footnotesize\(10^{-1}\)&\footnotesize\(10^{-2}\)&\footnotesize\(10^{-2}\)&\footnotesize\(10^{-2}\)&\footnotesize\(10^{-2}\)&\footnotesize\(10^{-2}\)&\footnotesize\(10^{-2}\)\\

\cline{2-11}


&\multirow{3}{*}{Sigmoid}& \multicolumn{1}{|c|}{\(\lambda\)}&\footnotesize\(10^{-6}\)&\footnotesize\(10^{-8}\)&\footnotesize\(10^{-5}\)&\footnotesize\(1\)&\footnotesize\(1\)&\footnotesize\(10\)&\footnotesize\(10^{-7}\)&\footnotesize\(10\)\\

&&\multicolumn{1}{|c|}{\(\gamma\)}&\footnotesize\(10^{-2}\)&\footnotesize\(10^{-3}\)&\footnotesize\(10^{-3}\)&\footnotesize\(10^{-3}\)&\footnotesize\(10^{-3}\)&\footnotesize\(10^{-3}\)&\footnotesize\(10^{-3}\)&\footnotesize\(10^{-3}\)\\

&&\multicolumn{1}{|c|}{\(c_0\)}&\footnotesize\(1\)&\footnotesize\(10^{-1}\)&\footnotesize\(1\)&\footnotesize\(1\)&\footnotesize\(10^{-1}\)&\footnotesize\(1\)&\footnotesize\(10\)&\footnotesize\(1\)\\

\cline{2-11}


&\multirow{1}{*}{Cosine}& \multicolumn{1}{|c|}{\(\lambda\)}&\footnotesize\(10^{3}\)&\footnotesize\(10^{3}\)&\footnotesize\(10^{3}\)&\footnotesize\(10^{3}\)&\footnotesize\(10^{3}\)&\footnotesize\(10^{3}\)&\footnotesize\(10^{3}\)&\footnotesize\(10^{3}\)\\

\cline{1-11}


\multirow{13}{*}{\(\mathscr{O}_{Sum}\)}&\multirow{1}{*}{Linear}& \multicolumn{1}{|c|}{\(\lambda\)}&\footnotesize\(10^4\)&\footnotesize\(10^4\)&\footnotesize\(10^5\)&\footnotesize\(10^5\)&\footnotesize\(10^5\)&\footnotesize\(10^5\)&\footnotesize\(10^5\)&\footnotesize\(10^5\)\\

\cline{2-11}


&\multirow{4}{*}{Polynomial}&\multicolumn{1}{|c|}{\(\lambda\)}&\footnotesize\(10^{-2}\)&\footnotesize\(10^{-2}\)&\footnotesize\(1\)&\footnotesize\(10\)&\footnotesize\(10\)&\footnotesize\(10^{5}\)&\footnotesize\(10^{4}\)&\footnotesize\(10^{5}\)\\

&&\multicolumn{1}{|c|}{\(\gamma\)}&\footnotesize\(10^{-2}\)&\footnotesize\(10^{-3}\)&\footnotesize\(10^{-3}\)&\footnotesize\(10^{-3}\)&\footnotesize\(10^{-3}\)&\footnotesize\(10^{-3}\)&\footnotesize\(10^{-3}\)&\footnotesize\(10^{-3}\)\\

&&\multicolumn{1}{|c|}{\(c_0\)}&\footnotesize\(1\)&\footnotesize\(1\)&\footnotesize\(1\)&\footnotesize\(1\)&\footnotesize\(10\)&\footnotesize\(10\)&\footnotesize\(10\)&\footnotesize\(10\)\\

&&\multicolumn{1}{|c|}{\(d\)}&\footnotesize\(9\)&\footnotesize\(9\)&\footnotesize\(9\)&\footnotesize\(8\)&\footnotesize\(6\)&\footnotesize\(8\)&\footnotesize\(7\)&\footnotesize\(7\)\\

\cline{2-11}


&\multirow{2}{*}{RBF}&\multicolumn{1}{|c|}{\footnotesize \(\lambda\)}&\footnotesize\(10^{-1}\)&\footnotesize\(10^{-2}\)&\footnotesize\(10^{-2}\)&\footnotesize\(10^{-2}\)&\footnotesize\(10^{-2}\)&\footnotesize\(10^{-2}\)&\footnotesize\(10^{-2}\)&\footnotesize\(10^{-2}\)\\

&&\multicolumn{1}{|c|}{\(\gamma\)}&\footnotesize\(10^{-1}\)&\footnotesize\(10^{-2}\)&\footnotesize\(10^{-2}\)&\footnotesize\(10^{-2}\)&\footnotesize\(10^{-3}\)&\footnotesize\(10^{-3}\)&\footnotesize\(10^{-3}\)&\footnotesize\(10^{-3}\)\\

\cline{2-11}


&\multirow{2}{*}{Laplacian}& \multicolumn{1}{|c|}{\(\lambda\)}&\footnotesize\(0\)&\footnotesize\(0\)&\footnotesize\(10^{-3}\)&\footnotesize\(10^{-3}\)&\footnotesize\(10^{-3}\)&\footnotesize\(10^{-3}\)&\footnotesize\(10^{-8}\)&\footnotesize\(10^{-8}\)\\

&&\multicolumn{1}{|c|}{\(\gamma\)}&\footnotesize\(10^{-1}\)&\footnotesize\(10^{-1}\)&\footnotesize\(10^{-2}\)&\footnotesize\(10^{-2}\)&\footnotesize\(10^{-2}\)&\footnotesize\(10^{-2}\)&\footnotesize\(10^{-2}\)&\footnotesize\(10^{-2}\)\\

\cline{2-11}


&\multirow{3}{*}{Sigmoid}& \multicolumn{1}{|c|}{\(\lambda\)}&\footnotesize\(10^{-6}\)&\footnotesize\(10^{-8}\)&\footnotesize\(10^{-5}\)&\footnotesize\(1\)&\footnotesize\(1\)&\footnotesize\(10\)&\footnotesize\(10^{-7}\)&\footnotesize\(10\)\\

&&\multicolumn{1}{|c|}{\(\gamma\)}&\footnotesize\(10^{-2}\)&\footnotesize\(10^{-3}\)&\footnotesize\(10^{-3}\)&\footnotesize\(10^{-3}\)&\footnotesize\(10^{-3}\)&\footnotesize\(10^{-3}\)&\footnotesize\(10^{-3}\)&\footnotesize\(10^{-3}\)\\

&&\multicolumn{1}{|c|}{\(c_0\)}&\footnotesize\(1\)&\footnotesize\(10^{-1}\)&\footnotesize\(1\)&\footnotesize\(1\)&\footnotesize\(10^{-1}\)&\footnotesize\(1\)&\footnotesize\(10\)&\footnotesize\(1\)\\

\cline{2-11}


&\multirow{1}{*}{Cosine}& \multicolumn{1}{|c|}{\(\lambda\)}&\footnotesize\(10^{3}\)&\footnotesize\(10^{3}\)&\footnotesize\(10^{3}\)&\footnotesize\(10^{3}\)&\footnotesize\(10^{3}\)&\footnotesize\(10^{3}\)&\footnotesize\(10^{3}\)&\footnotesize\(10^{3}\)\\

\cline{1-11}

\end{tabular}
\end{adjustbox}
\caption{The best hyperparameters trialled during the 10-fold cross-validation performed on the training datasets for \(\mathscr{H}_4\). Of all the hyperparameters trialled, the ones listed here resulted in models which achieved the highest average coefficient of determination (\(R^2\)) over all ten folds of the training dataset.}
\label{BestHyperparametersH4}
\end{table}
\newpage
\section{Numerical results}
\label{Tables of numerical results}
\renewcommand\thefigure{D.\arabic{figure}} 
\renewcommand\thetable{D.\arabic{table}}
The tables contained in this section of the appendices contain the majority of the numerical results reported in this work. Here we provide learning performance metrics for the models trained on all 1000 training datapoints with the hyperparameter values listed in the last four tables of Appendix \ref{Tables of hyperparameter values}. Specifically, for every choice of $\mathscr{O}_{XZ}$ or $\mathscr{O}_{Sum}$ with one of the parameterised sets of Hamiltonians $\mathscr{H}_1$, $\mathscr{H}_2$, $\mathscr{H}_3$, or $\mathscr{H}_4$, we provide the performance metrics calculated based on predictions made by the trained models on both the training and testing sets.


\subsection{Results for \(\mathscr{O}_{XZ}\) and \(\mathscr{H}_1\)}

\begin{table}[H]
\begin{adjustbox}{center,max width=\textwidth}
\begin{tabular}{| c | c c c c c c c c c c |}

\cline{4-11}

\multicolumn{3}{c|}{}&\multicolumn{8}{c|}{\textbf{Number of qubits} \((n)\)}\\

\cline{1-11}

\multicolumn{1}{|c|}{\textbf{Dataset}}&\multicolumn{1}{c|}{\textbf{Kernel}}& \multicolumn{1}{c|}{\textbf{Metric}}
& \multicolumn{1}{c|}{5}& \multicolumn{1}{c|}{10}
& \multicolumn{1}{c|}{15}& \multicolumn{1}{c|}{20}
& \multicolumn{1}{c|}{25}& \multicolumn{1}{c|}{30}
& \multicolumn{1}{c|}{35}& \multicolumn{1}{c|}{40}\\

\cline{1-11}


\multirow{18}{*}{Testing}&\multirow{3}{*}{\footnotesize Linear}& \multicolumn{1}{|c|}{\footnotesize \(R^2\)}&\footnotesize-1.9284&\footnotesize-4.1339&\footnotesize-3.5942&\footnotesize-3.7616&\footnotesize-3.4297&\footnotesize-3.3197&\footnotesize-3.0432&\footnotesize-2.8091\\

&&\multicolumn{1}{|c|}{\footnotesize RMSE}&\footnotesize0.7886&\footnotesize0.8122&\footnotesize0.8156&\footnotesize0.8214&\footnotesize0.8156&\footnotesize0.8147&\footnotesize0.8030&\footnotesize0.7941\\

&&\multicolumn{1}{|c|}{\footnotesize MAE}&\footnotesize0.7091&\footnotesize0.7296&\footnotesize0.7275&\footnotesize0.7307&\footnotesize0.7186&\footnotesize0.7159&\footnotesize0.6991&\footnotesize0.6857\\

\cline{2-11}


&\multirow{3}{*}{\footnotesize Polynomial}&\multicolumn{1}{|c|}{\footnotesize \(R^2\)}&\footnotesize0.8604&\footnotesize0.8265&\footnotesize0.7872&\footnotesize0.8276&\footnotesize0.8145&\footnotesize0.8132&\footnotesize0.8154&\footnotesize0.8097\\

&&\multicolumn{1}{|c|}{\footnotesize RMSE}&\footnotesize0.1722&\footnotesize0.1493&\footnotesize0.1755&\footnotesize0.1563&\footnotesize0.1669&\footnotesize0.1694&\footnotesize0.1716&\footnotesize0.1775\\

&&\multicolumn{1}{|c|}{\footnotesize MAE}&\footnotesize0.1228&\footnotesize0.1076&\footnotesize0.1185&\footnotesize0.1121&\footnotesize0.1200&\footnotesize0.1191&\footnotesize0.1249&\footnotesize0.1274\\

\cline{2-11}


&\multirow{3}{*}{\footnotesize RBF}&\multicolumn{1}{|c|}{\footnotesize \(R^2\)}&\footnotesize\textbf{0.9646}&\footnotesize0.9454&\footnotesize0.8910&\footnotesize0.8963&\footnotesize0.8771&\footnotesize0.9038&\footnotesize0.9111&\footnotesize0.9246\\

&&\multicolumn{1}{|c|}{\footnotesize RMSE}&\footnotesize\textbf{0.0867}&\footnotesize0.0837&\footnotesize0.1256&\footnotesize0.1212&\footnotesize0.1359&\footnotesize0.1216&\footnotesize0.1191&\footnotesize0.1117\\

&&\multicolumn{1}{|c|}{\footnotesize MAE}&\footnotesize\textbf{0.0431}&\footnotesize0.0489&\footnotesize0.0586&\footnotesize0.0772&\footnotesize0.0790&\footnotesize0.0739&\footnotesize0.0674&\footnotesize0.0594\\

\cline{2-11}


&\multirow{3}{*}{\footnotesize Laplacian}&\multicolumn{1}{|c|}{\footnotesize \(R^2\)}&\footnotesize0.9391&\footnotesize\textbf{0.9696}&\footnotesize\textbf{0.9228}&\footnotesize\textbf{0.9547}&\footnotesize\textbf{0.9356}&\footnotesize\textbf{0.9362}&\footnotesize\textbf{0.9401}&\footnotesize\textbf{0.9542}\\

&&\multicolumn{1}{|c|}{\footnotesize RMSE}&\footnotesize0.1137&\footnotesize\textbf{0.0625}&\footnotesize\textbf{0.1057}&\footnotesize\textbf{0.0801}&\footnotesize\textbf{0.0983}&\footnotesize\textbf{0.0990}&\footnotesize\textbf{0.0978}&\footnotesize\textbf{0.0871}\\

&&\multicolumn{1}{|c|}{\footnotesize MAE}&\footnotesize0.0763&\footnotesize\textbf{0.0388}&\footnotesize\textbf{0.0495}&\footnotesize\textbf{0.0463}&\footnotesize\textbf{0.0529}&\footnotesize\textbf{0.0660}&\footnotesize\textbf{0.0636}&\footnotesize\textbf{0.0559}\\

\cline{2-11}


&\multirow{3}{*}{\footnotesize Sigmoid}&\multicolumn{1}{|c|}{\footnotesize \(R^2\)}&\footnotesize0.6777&\footnotesize0.8639&\footnotesize0.6475&\footnotesize0.8191&\footnotesize0.8144&\footnotesize0.8461&\footnotesize0.6852&\footnotesize0.6400\\

&&\multicolumn{1}{|c|}{\footnotesize RMSE}&\footnotesize0.2616&\footnotesize0.1322&\footnotesize0.2259&\footnotesize0.1601&\footnotesize0.1669&\footnotesize0.1538&\footnotesize0.2240&\footnotesize0.2441\\

&&\multicolumn{1}{|c|}{\footnotesize MAE}&\footnotesize0.1852&\footnotesize0.0895&\footnotesize0.1617&\footnotesize0.1123&\footnotesize0.1176&\footnotesize0.1081&\footnotesize0.1696&\footnotesize0.1899\\

\cline{2-11}


&\multirow{3}{*}{\footnotesize Cosine}&\multicolumn{1}{|c|}{\footnotesize \(R^2\)}&\footnotesize-1.9508&\footnotesize-4.1287&\footnotesize-3.5710&\footnotesize-3.7659&\footnotesize-3.4271&\footnotesize-3.3113&\footnotesize-3.0796&\footnotesize-2.8382\\

&&\multicolumn{1}{|c|}{\footnotesize RMSE}&\footnotesize0.7916&\footnotesize0.8117&\footnotesize0.8135&\footnotesize0.8218&\footnotesize0.8153&\footnotesize0.8139&\footnotesize0.8066&\footnotesize0.7971\\

&&\multicolumn{1}{|c|}{\footnotesize MAE}&\footnotesize0.7116&\footnotesize0.7285&\footnotesize0.7240&\footnotesize0.7306&\footnotesize0.7175&\footnotesize0.7136&\footnotesize0.7024&\footnotesize0.6882\\

\cline{1-11}


\multirow{18}{*}{Training}&\multirow{3}{*}{\footnotesize Linear}& \multicolumn{1}{|c|}{\footnotesize \(R^2\)}&\footnotesize-1.7147&\footnotesize -3.4479&\footnotesize -3.1980&\footnotesize -3.1767&\footnotesize -2.9977&\footnotesize -2.8207&\footnotesize -2.6211&\footnotesize -2.4198\\

&&\multicolumn{1}{|c|}{\footnotesize RMSE}&\footnotesize0.7692&\footnotesize 0.7919&\footnotesize 0.7971&\footnotesize 0.8052&\footnotesize 0.8025&\footnotesize 0.7979&\footnotesize 0.7870&\footnotesize 0.7745\\

&&\multicolumn{1}{|c|}{\footnotesize MAE}&\footnotesize0.6823&\footnotesize 0.6985&\footnotesize 0.6996&\footnotesize 0.7035&\footnotesize 0.6971&\footnotesize 0.6893&\footnotesize 0.6750&\footnotesize 0.6589\\

\cline{2-11}


&\multirow{3}{*}{\footnotesize Polynomial}&\multicolumn{1}{|c|}{\footnotesize \(R^2\)}&\footnotesize0.8942&\footnotesize 0.8747&\footnotesize 0.8709&\footnotesize 0.8712&\footnotesize 0.8715&\footnotesize 0.8779&\footnotesize 0.8650&\footnotesize 0.8572\\

&&\multicolumn{1}{|c|}{\footnotesize RMSE}&\footnotesize0.1519&\footnotesize 0.1329&\footnotesize 0.1398&\footnotesize 0.1414&\footnotesize 0.1439&\footnotesize 0.1426&\footnotesize 0.1520&\footnotesize 0.1583\\

&&\multicolumn{1}{|c|}{\footnotesize MAE}&\footnotesize0.1071&\footnotesize 0.0939&\footnotesize 0.0992&\footnotesize 0.1006&\footnotesize 0.1035&\footnotesize 0.0997&\footnotesize 0.1078&\footnotesize 0.1110\\

\cline{2-11}


&\multirow{3}{*}{\footnotesize RBF}&\multicolumn{1}{|c|}{\footnotesize \(R^2\)}&\footnotesize0.9997&\footnotesize 0.9921&\footnotesize 0.9987&\footnotesize 0.9578&\footnotesize 0.9724&\footnotesize 0.9686&\footnotesize 0.9803&\footnotesize 0.9871\\

&&\multicolumn{1}{|c|}{\footnotesize RMSE}&\footnotesize0.0077&\footnotesize 0.0334&\footnotesize 0.0138&\footnotesize 0.0809&\footnotesize 0.0667&\footnotesize 0.0723&\footnotesize 0.0580&\footnotesize 0.0476\\

&&\multicolumn{1}{|c|}{\footnotesize MAE}&\footnotesize0.0049&\footnotesize 0.0224&\footnotesize 0.0078&\footnotesize 0.0557&\footnotesize 0.0452&\footnotesize 0.0491&\footnotesize 0.0390&\footnotesize 0.0308\\

\cline{2-11}


&\multirow{3}{*}{\footnotesize Laplacian}& \multicolumn{1}{|c|}{\footnotesize \(R^2\)}&\footnotesize\textbf{1.0000}&\footnotesize \textbf{1.0000}&\footnotesize \textbf{1.0000}&\footnotesize \textbf{1.0000}&\footnotesize \textbf{1.0000}&\footnotesize \textbf{0.9997}&\footnotesize \textbf{0.9998}&\footnotesize \textbf{0.9999}\\

&&\multicolumn{1}{|c|}{\footnotesize RMSE}&\footnotesize\textbf{0.0000}&\footnotesize \textbf{0.0000}&\footnotesize \textbf{0.0000}&\footnotesize \textbf{0.0000}&\footnotesize \textbf{0.0000}&\footnotesize \textbf{0.0075}&\footnotesize \textbf{0.0059}&\footnotesize \textbf{0.0047}\\

&&\multicolumn{1}{|c|}{\footnotesize MAE}&\footnotesize\textbf{0.0000}&\footnotesize \textbf{0.0000}&\footnotesize \textbf{0.0000}&\footnotesize \textbf{0.0000}&\footnotesize \textbf{0.0000}&\footnotesize \textbf{0.0054}&\footnotesize \textbf{0.0042}&\footnotesize \textbf{0.0033}\\

\cline{2-11}


&\multirow{3}{*}{\footnotesize Sigmoid}&\multicolumn{1}{|c|}{\footnotesize \(R^2\)}&\footnotesize0.7042&\footnotesize 0.9410&\footnotesize 0.7206&\footnotesize 0.8773&\footnotesize 0.8743&\footnotesize 0.9130&\footnotesize 0.7229&\footnotesize 0.6310\\

&&\multicolumn{1}{|c|}{\footnotesize RMSE}&\footnotesize0.2539&\footnotesize 0.0912&\footnotesize 0.2056&\footnotesize 0.1380&\footnotesize 0.1423&\footnotesize 0.1204&\footnotesize 0.2177&\footnotesize 0.2544\\

&&\multicolumn{1}{|c|}{\footnotesize MAE}&\footnotesize0.1791&\footnotesize 0.0666&\footnotesize 0.1470&\footnotesize 0.0996&\footnotesize 0.1009&\footnotesize 0.0881&\footnotesize 0.1567&\footnotesize 0.1892\\

\cline{2-11}


&\multirow{3}{*}{\footnotesize Cosine}&\multicolumn{1}{|c|}{\footnotesize \(R^2\)}&\footnotesize-1.7129&\footnotesize-3.4530&\footnotesize-3.2072&\footnotesize-3.1761&\footnotesize-2.9999&\footnotesize-2.8250&\footnotesize-2.6125&\footnotesize-2.4122\\

&&\multicolumn{1}{|c|}{\footnotesize RMSE}&\footnotesize0.7690&\footnotesize 0.7923&\footnotesize 0.7979&\footnotesize 0.8052&\footnotesize 0.8027&\footnotesize 0.7983&\footnotesize 0.7860&\footnotesize 0.7736\\

&&\multicolumn{1}{|c|}{\footnotesize MAE}&\footnotesize0.6816&\footnotesize 0.6983&\footnotesize 0.6983&\footnotesize 0.7030&\footnotesize 0.6961&\footnotesize 0.6873&\footnotesize 0.6746&\footnotesize 0.6578\\

\cline{1-11}

\end{tabular}
\end{adjustbox}
\caption{\textbf{(\(\mathscr{O}_{XZ}\) with \(\mathscr{H}_1\))} Learning performance metrics for the models trained on all 1000 training datapoints using the hyperparameters listed in Table \ref{BestHyperparametersH1}. The metrics are calculated based on the predictions made by the models on the training and testing sets with labels determined by \(\mathscr{O}_{XZ}\) with \(\mathscr{H}_1\). The best values for each qubit number and each metric are in bold text for the training and testing datasets.}
\label{ResultsOXZH1}
\end{table}


\subsection{Results for \(\mathscr{O}_{XZ}\) and \(\mathscr{H}_2\)}

\begin{table}[H]
\begin{adjustbox}{center,max width=\textwidth}
\begin{tabular}{| c | c c c c c c c c c c |}

\cline{4-11}

\multicolumn{3}{c|}{}&\multicolumn{8}{c|}{\textbf{Number of qubits} \((n)\)}\\

\cline{1-11}

\multicolumn{1}{|c|}{\textbf{Dataset}}&\multicolumn{1}{c|}{\textbf{Kernel}}& \multicolumn{1}{c|}{\textbf{Metric}}
& \multicolumn{1}{c|}{5}& \multicolumn{1}{c|}{10}
& \multicolumn{1}{c|}{15}& \multicolumn{1}{c|}{20}
& \multicolumn{1}{c|}{25}& \multicolumn{1}{c|}{30}
& \multicolumn{1}{c|}{35}& \multicolumn{1}{c|}{40}\\

\cline{1-11}


\multirow{18}{*}{Testing}&\multirow{3}{*}{\footnotesize Linear}& \multicolumn{1}{|c|}{\footnotesize \(R^2\)}&\footnotesize-2.1150&\footnotesize-3.8313&\footnotesize-4.6197&\footnotesize-4.9758&\footnotesize-4.8876&\footnotesize-4.6711&\footnotesize-4.4779&\footnotesize-4.2726\\

&&\multicolumn{1}{|c|}{\footnotesize RMSE}&\footnotesize0.7436&\footnotesize0.7882&\footnotesize0.8055&\footnotesize0.8159&\footnotesize0.8166&\footnotesize0.8162&\footnotesize0.8148&\footnotesize0.8127\\

&&\multicolumn{1}{|c|}{\footnotesize MAE}&\footnotesize0.6594&\footnotesize0.7145&\footnotesize0.7385&\footnotesize0.7505&\footnotesize0.7493&\footnotesize0.7454&\footnotesize0.7407&\footnotesize0.7354\\

\cline{2-11}


&\multirow{3}{*}{\footnotesize Polynomial}&\multicolumn{1}{|c|}{\footnotesize \(R^2\)}&\footnotesize0.6585&\footnotesize0.6609&\footnotesize0.6947&\footnotesize0.7395&\footnotesize0.7397&\footnotesize0.7450&\footnotesize0.7500&\footnotesize0.7739\\

&&\multicolumn{1}{|c|}{\footnotesize RMSE}&\footnotesize0.2462&\footnotesize0.2088&\footnotesize0.1877&\footnotesize0.1703&\footnotesize0.1717&\footnotesize0.1731&\footnotesize0.1741&\footnotesize0.1683\\

&&\multicolumn{1}{|c|}{\footnotesize MAE}&\footnotesize0.1804&\footnotesize0.1481&\footnotesize0.1335&\footnotesize0.1248&\footnotesize0.1230&\footnotesize0.1264&\footnotesize0.1278&\footnotesize0.1254\\

\cline{2-11}


&\multirow{3}{*}{\footnotesize RBF}&\multicolumn{1}{|c|}{\footnotesize \(R^2\)}&\footnotesize\textbf{0.9656}&\footnotesize\textbf{0.8670}&\footnotesize0.8306&\footnotesize0.8638&\footnotesize0.8725&\footnotesize0.8869&\footnotesize0.8959&\footnotesize\textbf{0.9095}\\

&&\multicolumn{1}{|c|}{\footnotesize RMSE}&\footnotesize\textbf{0.0782}&\footnotesize\textbf{0.1308}&\footnotesize0.1399&\footnotesize0.1232&\footnotesize0.1202&\footnotesize0.1152&\footnotesize0.1123&\footnotesize\textbf{0.1065}\\

&&\multicolumn{1}{|c|}{\footnotesize MAE}&\footnotesize\textbf{0.0457}&\footnotesize\textbf{0.0848}&\footnotesize0.0894&\footnotesize0.0865&\footnotesize0.0827&\footnotesize0.0804&\footnotesize0.0740&\footnotesize\textbf{0.0695}\\

\cline{2-11}


&\multirow{3}{*}{\footnotesize Laplacian}&\multicolumn{1}{|c|}{\footnotesize \(R^2\)}&\footnotesize0.8965&\footnotesize0.8536&\footnotesize\textbf{0.8452}&\footnotesize\textbf{0.8751}&\footnotesize\textbf{0.8834}&\footnotesize\textbf{0.8970}&\footnotesize\textbf{0.8982}&\footnotesize0.8789\\

&&\multicolumn{1}{|c|}{\footnotesize RMSE}&\footnotesize0.1355&\footnotesize0.1372&\footnotesize\textbf{0.1337}&\footnotesize\textbf{0.1179}&\footnotesize\textbf{0.1149}&\footnotesize\textbf{0.1100}&\footnotesize\textbf{0.1110}&\footnotesize0.1232\\

&&\multicolumn{1}{|c|}{\footnotesize MAE}&\footnotesize0.0839&\footnotesize0.0920&\footnotesize\textbf{0.0858}&\footnotesize\textbf{0.0778}&\footnotesize\textbf{0.0744}&\footnotesize\textbf{0.0709}&\footnotesize\textbf{0.0716}&\footnotesize0.0880\\

\cline{2-11}


&\multirow{3}{*}{\footnotesize Sigmoid}&\multicolumn{1}{|c|}{\footnotesize \(R^2\)}&\footnotesize0.3878&\footnotesize0.7562&\footnotesize0.5072&\footnotesize0.7333&\footnotesize0.7533&\footnotesize0.8132&\footnotesize0.6145&\footnotesize0.5679\\

&&\multicolumn{1}{|c|}{\footnotesize RMSE}&\footnotesize0.3297&\footnotesize0.1771&\footnotesize0.2385&\footnotesize0.1724&\footnotesize0.1671&\footnotesize0.1482&\footnotesize0.2161&\footnotesize0.2327\\

&&\multicolumn{1}{|c|}{\footnotesize MAE}&\footnotesize0.2513&\footnotesize0.1221&\footnotesize0.1712&\footnotesize0.1254&\footnotesize0.1198&\footnotesize0.1063&\footnotesize0.1541&\footnotesize0.1646\\

\cline{2-11}


&\multirow{3}{*}{\footnotesize Cosine}& \multicolumn{1}{|c|}{\footnotesize \(R^2\)}&\footnotesize-2.1225&\footnotesize-3.8386&\footnotesize-4.6360&\footnotesize-4.9752&\footnotesize-4.8860&\footnotesize-4.6687&\footnotesize-4.4737&\footnotesize-4.2664\\

&&\multicolumn{1}{|c|}{\footnotesize RMSE}&\footnotesize0.7445&\footnotesize0.7888&\footnotesize0.8067&\footnotesize0.8158&\footnotesize0.8165&\footnotesize0.8161&\footnotesize0.8144&\footnotesize0.8122\\

&&\multicolumn{1}{|c|}{\footnotesize MAE}&\footnotesize0.6603&\footnotesize0.7146&\footnotesize0.7396&\footnotesize0.7502&\footnotesize0.7488&\footnotesize0.7447&\footnotesize0.7400&\footnotesize0.7343\\

\cline{1-11} 


\multirow{18}{*}{Training}&\multirow{3}{*}{\footnotesize Linear}& \multicolumn{1}{|c|}{\footnotesize \(R^2\)}&\footnotesize-2.0465&\footnotesize -3.9309&\footnotesize -5.1218&\footnotesize -5.1428&\footnotesize -5.2954&\footnotesize -4.9662&\footnotesize -4.8482&\footnotesize -4.5928\\

&&\multicolumn{1}{|c|}{\footnotesize RMSE}&\footnotesize0.7533&\footnotesize 0.7960&\footnotesize 0.8159&\footnotesize 0.8177&\footnotesize 0.8237&\footnotesize 0.8204&\footnotesize 0.8194&\footnotesize 0.8164\\

&&\multicolumn{1}{|c|}{\footnotesize MAE}&\footnotesize0.6691&\footnotesize 0.7256&\footnotesize 0.7532&\footnotesize 0.7537&\footnotesize 0.7588&\footnotesize 0.7515&\footnotesize 0.7479&\footnotesize 0.7414\\

\cline{2-11}


&\multirow{3}{*}{\footnotesize Polynomial}&\multicolumn{1}{|c|}{\footnotesize \(R^2\)}&\footnotesize0.7107&\footnotesize 0.6852&\footnotesize 0.7371&\footnotesize 0.7708&\footnotesize 0.7728&\footnotesize 0.7736&\footnotesize 0.7800&\footnotesize 0.7957\\

&&\multicolumn{1}{|c|}{\footnotesize RMSE}&\footnotesize0.2321&\footnotesize 0.2011&\footnotesize 0.1691&\footnotesize 0.1580&\footnotesize 0.1565&\footnotesize 0.1598&\footnotesize 0.1589&\footnotesize 0.1560\\

&&\multicolumn{1}{|c|}{\footnotesize MAE}&\footnotesize0.1741&\footnotesize 0.1374&\footnotesize 0.1204&\footnotesize 0.1111&\footnotesize 0.1100&\footnotesize 0.1115&\footnotesize 0.1125&\footnotesize 0.1123\\

\cline{2-11}


&\multirow{3}{*}{\footnotesize RBF}&\multicolumn{1}{|c|}{\footnotesize \(R^2\)}&\footnotesize0.9999&\footnotesize 0.9037&\footnotesize 0.9958&\footnotesize 0.9047&\footnotesize 0.9131&\footnotesize 0.9341&\footnotesize 0.9571&\footnotesize 0.9720\\

&&\multicolumn{1}{|c|}{\footnotesize RMSE}&\footnotesize0.0035&\footnotesize 0.1112&\footnotesize 0.0215&\footnotesize 0.1018&\footnotesize 0.0968&\footnotesize 0.0862&\footnotesize 0.0702&\footnotesize 0.0578\\

&&\multicolumn{1}{|c|}{\footnotesize MAE}&\footnotesize0.0020&\footnotesize 0.0663&\footnotesize 0.0124&\footnotesize 0.0649&\footnotesize 0.0642&\footnotesize 0.0556&\footnotesize 0.0457&\footnotesize 0.0368\\

\cline{2-11}


&\multirow{3}{*}{\footnotesize Laplacian}& \multicolumn{1}{|c|}{\footnotesize \(R^2\)}&\footnotesize\textbf{1.0000}&\footnotesize \textbf{0.9544}&\footnotesize \textbf{0.9993}&\footnotesize \textbf{0.9997}&\footnotesize \textbf{1.0000}&\footnotesize \textbf{1.0000}&\footnotesize \textbf{1.0000}&\footnotesize \textbf{0.9894}\\

&&\multicolumn{1}{|c|}{\footnotesize RMSE}&\footnotesize\textbf{0.0000}&\footnotesize \textbf{0.0765}&\footnotesize \textbf{0.0087}&\footnotesize \textbf{0.0060}&\footnotesize \textbf{0.0000}&\footnotesize \textbf{0.0000}&\footnotesize \textbf{0.0000}&\footnotesize \textbf{0.0355}\\

&&\multicolumn{1}{|c|}{\footnotesize MAE}&\footnotesize\textbf{0.0000}&\footnotesize \textbf{0.0485}&\footnotesize \textbf{0.0058}&\footnotesize \textbf{0.0038}&\footnotesize \textbf{0.0000}&\footnotesize \textbf{0.0000}&\footnotesize \textbf{0.0000}&\footnotesize \textbf{0.0259}\\

\cline{2-11}


&\multirow{3}{*}{\footnotesize Sigmoid}& \multicolumn{1}{|c|}{\footnotesize \(R^2\)}&\footnotesize0.4002&\footnotesize 0.8042&\footnotesize 0.5287&\footnotesize 0.7707&\footnotesize 0.7758&\footnotesize 0.8554&\footnotesize 0.6587&\footnotesize 0.5484\\

&&\multicolumn{1}{|c|}{\footnotesize RMSE}&\footnotesize0.3343&\footnotesize 0.1586&\footnotesize 0.2264&\footnotesize 0.1580&\footnotesize 0.1554&\footnotesize 0.1277&\footnotesize 0.1979&\footnotesize 0.2320\\

&&\multicolumn{1}{|c|}{\footnotesize MAE}&\footnotesize0.2581&\footnotesize 0.1065&\footnotesize 0.1629&\footnotesize 0.1132&\footnotesize 0.1097&\footnotesize 0.0874&\footnotesize 0.1414&\footnotesize 0.1670\\

\cline{2-11}


&\multirow{3}{*}{\footnotesize Cosine}& \multicolumn{1}{|c|}{\footnotesize \(R^2\)}&\footnotesize-2.0444&\footnotesize -3.9348&\footnotesize -5.1180&\footnotesize -5.1514&\footnotesize -5.3066&\footnotesize -4.9779&\footnotesize -4.8603&\footnotesize -4.6056\\

&&\multicolumn{1}{|c|}{\footnotesize RMSE}&\footnotesize0.7530&\footnotesize 0.7963&\footnotesize 0.8156&\footnotesize 0.8183&\footnotesize 0.8244&\footnotesize 0.8212&\footnotesize 0.8202&\footnotesize 0.8174\\

&&\multicolumn{1}{|c|}{\footnotesize MAE}&\footnotesize0.6687&\footnotesize 0.7265&\footnotesize 0.7530&\footnotesize 0.7546&\footnotesize 0.7600&\footnotesize 0.7529&\footnotesize 0.7493&\footnotesize 0.7428\\

\cline{1-11}

\end{tabular}
\end{adjustbox}
\caption{\textbf{(\(\mathscr{O}_{XZ}\) with \(\mathscr{H}_2\))} Learning performance metrics for the models trained on all 1000 training datapoints using the hyperparameters listed in Table \ref{BestHyperparametersH2}. The metrics are calculated based on the predictions made by the models on the training and testing sets with labels determined by \(\mathscr{O}_{XZ}\) with \(\mathscr{H}_2\). The best values for each qubit number and each metric are in bold text for the training and testing datasets.}
\label{ResultsOXZH2}
\end{table}


\subsection{Results for \(\mathscr{O}_{XZ}\) and \(\mathscr{H}_3\)}

\begin{table}[H]
\begin{adjustbox}{center,max width=\textwidth}
\begin{tabular}{| c | c c c c c c c c c c |}

\cline{4-11}

\multicolumn{3}{c|}{}&\multicolumn{8}{c|}{Number of qubits \((n)\)}\\

\cline{1-11}

\multicolumn{1}{|c|}{\textbf{Dataset}}&\multicolumn{1}{c|}{\textbf{Kernel}}& \multicolumn{1}{c|}{\textbf{Metric}}
& \multicolumn{1}{c|}{5}& \multicolumn{1}{c|}{10}
& \multicolumn{1}{c|}{15}& \multicolumn{1}{c|}{20}
& \multicolumn{1}{c|}{25}& \multicolumn{1}{c|}{30}
& \multicolumn{1}{c|}{35}& \multicolumn{1}{c|}{40}\\

\cline{1-11}


\multirow{18}{*}{Testing}&\multirow{3}{*}{\footnotesize Linear}& \multicolumn{1}{|c|}{\footnotesize \(R^2\)}&\footnotesize-1.5252&\footnotesize-1.9782&\footnotesize-1.9471&\footnotesize-1.8755&\footnotesize-1.7657&\footnotesize-1.6170&\footnotesize-1.4506&\footnotesize-1.2877\\

&&\multicolumn{1}{|c|}{\footnotesize RMSE}&\footnotesize0.6230&\footnotesize0.6630&\footnotesize0.6810&\footnotesize0.6884&\footnotesize0.6873&\footnotesize0.6775&\footnotesize0.6611&\footnotesize0.6408\\

&&\multicolumn{1}{|c|}{\footnotesize MAE}&\footnotesize0.5066&\footnotesize0.5423&\footnotesize0.5542&\footnotesize0.5567&\footnotesize0.5506&\footnotesize0.5349&\footnotesize0.5117&\footnotesize0.4843\\

\cline{2-11}


&\multirow{3}{*}{\footnotesize Polynomial}&\multicolumn{1}{|c|}{\footnotesize \(R^2\)}&\footnotesize0.7911&\footnotesize0.8681&\footnotesize0.8797&\footnotesize0.8807&\footnotesize0.8866&\footnotesize0.8824&\footnotesize0.8855&\footnotesize0.8780\\

&&\multicolumn{1}{|c|}{\footnotesize RMSE}&\footnotesize0.1792&\footnotesize0.1395&\footnotesize0.1376&\footnotesize0.1402&\footnotesize0.1392&\footnotesize0.1436&\footnotesize0.1429&\footnotesize0.1480\\

&&\multicolumn{1}{|c|}{\footnotesize MAE}&\footnotesize0.1323&\footnotesize0.1031&\footnotesize0.1035&\footnotesize0.1066&\footnotesize0.1054&\footnotesize0.1076&\footnotesize0.1059&\footnotesize0.1080\\

\cline{2-11}


&\multirow{3}{*}{\footnotesize RBF}&\multicolumn{1}{|c|}{\footnotesize \(R^2\)}&\footnotesize\textbf{0.9760}&\footnotesize\textbf{0.9829}&\footnotesize\textbf{0.9866}&\footnotesize0.9682&\footnotesize\textbf{0.9749}&\footnotesize\textbf{0.9772}&\footnotesize\textbf{0.9799}&\footnotesize\textbf{0.9823}\\

&&\multicolumn{1}{|c|}{\footnotesize RMSE}&\footnotesize\textbf{0.0608}&\footnotesize\textbf{0.0502}&\footnotesize\textbf{0.0459}&\footnotesize0.0724&\footnotesize\textbf{0.0655}&\footnotesize\textbf{0.0633}&\footnotesize\textbf{0.0598}&\footnotesize\textbf{0.0564}\\

&&\multicolumn{1}{|c|}{\footnotesize MAE}&\footnotesize\textbf{0.0379}&\footnotesize\textbf{0.0354}&\footnotesize\textbf{0.0290}&\footnotesize0.0511&\footnotesize\textbf{0.0452}&\footnotesize\textbf{0.0430}&\footnotesize\textbf{0.0384}&\footnotesize\textbf{0.0347}\\

\cline{2-11}


&\multirow{3}{*}{\footnotesize Laplacian}& \multicolumn{1}{|c|}{\footnotesize \(R^2\)}&\footnotesize0.9537&\footnotesize0.9588&\footnotesize0.9685&\footnotesize\textbf{0.9720}&\footnotesize0.9736&\footnotesize0.9744&\footnotesize0.9745&\footnotesize0.9736\\

&&\multicolumn{1}{|c|}{\footnotesize RMSE}&\footnotesize0.0843&\footnotesize0.0780&\footnotesize0.0704&\footnotesize\textbf{0.0679}&\footnotesize0.0671&\footnotesize0.0670&\footnotesize0.0675&\footnotesize0.0688\\

&&\multicolumn{1}{|c|}{\footnotesize MAE}&\footnotesize0.0573&\footnotesize0.0570&\footnotesize0.0498&\footnotesize\textbf{0.0470}&\footnotesize0.0456&\footnotesize0.0445&\footnotesize0.0431&\footnotesize0.0424\\

\cline{2-11}


&\multirow{3}{*}{\footnotesize Sigmoid}& \multicolumn{1}{|c|}{\footnotesize \(R^2\)}&\footnotesize0.7489&\footnotesize0.9173&\footnotesize0.8105&\footnotesize0.8782&\footnotesize0.8828&\footnotesize0.9099&\footnotesize0.9340&\footnotesize0.7311\\

&&\multicolumn{1}{|c|}{\footnotesize RMSE}&\footnotesize0.1965&\footnotesize0.1105&\footnotesize0.1727&\footnotesize0.1417&\footnotesize0.1415&\footnotesize0.1257&\footnotesize0.1085&\footnotesize0.2197\\

&&\multicolumn{1}{|c|}{\footnotesize MAE}&\footnotesize0.1504&\footnotesize0.0835&\footnotesize0.1336&\footnotesize0.1067&\footnotesize0.1092&\footnotesize0.0975&\footnotesize0.0818&\footnotesize0.1781\\

\cline{2-11}


&\multirow{3}{*}{\footnotesize Cosine}& \multicolumn{1}{|c|}{\footnotesize \(R^2\)}&\footnotesize-1.5284&\footnotesize-1.9820&\footnotesize-1.9474&\footnotesize-1.8757&\footnotesize-1.7659&\footnotesize-1.6170&\footnotesize-1.4505&\footnotesize-1.2875\\

&&\multicolumn{1}{|c|}{\footnotesize RMSE}&\footnotesize0.6234&\footnotesize0.6635&\footnotesize0.6810&\footnotesize0.6884&\footnotesize0.6873&\footnotesize0.6775&\footnotesize0.6611&\footnotesize0.6407\\

&&\multicolumn{1}{|c|}{\footnotesize MAE}&\footnotesize0.5068&\footnotesize0.5431&\footnotesize0.5539&\footnotesize0.5561&\footnotesize0.5493&\footnotesize0.5326&\footnotesize0.5087&\footnotesize0.4808\\

\cline{1-11}


\multirow{18}{*}{Training}&\multirow{3}{*}{\footnotesize Linear}& \multicolumn{1}{|c|}{\footnotesize \(R^2\)}&\footnotesize-1.5555&\footnotesize -2.0435&\footnotesize -1.9755&\footnotesize -1.8833&\footnotesize -1.7698&\footnotesize -1.6269&\footnotesize -1.4667&\footnotesize -1.3070\\

&&\multicolumn{1}{|c|}{\footnotesize RMSE}&\footnotesize0.6146&\footnotesize 0.6641&\footnotesize 0.6847&\footnotesize 0.6934&\footnotesize 0.6927&\footnotesize 0.6825&\footnotesize 0.6646&\footnotesize 0.6424\\

&&\multicolumn{1}{|c|}{\footnotesize MAE}&\footnotesize0.5005&\footnotesize 0.5452&\footnotesize 0.5581&\footnotesize 0.5607&\footnotesize 0.5547&\footnotesize 0.5390&\footnotesize 0.5154&\footnotesize 0.4871\\

\cline{2-11}


&\multirow{3}{*}{\footnotesize Polynomial}&\multicolumn{1}{|c|}{\footnotesize \(R^2\)}&\footnotesize0.8181&\footnotesize 0.8861&\footnotesize 0.8976&\footnotesize 0.9018&\footnotesize 0.9097&\footnotesize 0.9057&\footnotesize 0.9100&\footnotesize 0.9054\\

&&\multicolumn{1}{|c|}{\footnotesize RMSE}&\footnotesize0.1640&\footnotesize 0.1285&\footnotesize 0.1270&\footnotesize 0.1279&\footnotesize 0.1251&\footnotesize 0.1293&\footnotesize 0.1270&\footnotesize 0.1301\\

&&\multicolumn{1}{|c|}{\footnotesize MAE}&\footnotesize0.1217&\footnotesize 0.0954&\footnotesize 0.0960&\footnotesize 0.0984&\footnotesize 0.0952&\footnotesize 0.0969&\footnotesize 0.0949&\footnotesize 0.0966\\

\cline{2-11}


&\multirow{3}{*}{\footnotesize RBF}&\multicolumn{1}{|c|}{\footnotesize \(R^2\)}&\footnotesize0.9992&\footnotesize 0.9945&\footnotesize 0.9991&\footnotesize 0.9847&\footnotesize 0.9871&\footnotesize 0.9920&\footnotesize 0.9956&\footnotesize 0.9979\\

&&\multicolumn{1}{|c|}{\footnotesize RMSE}&\footnotesize0.0108&\footnotesize 0.0282&\footnotesize 0.0118&\footnotesize 0.0506&\footnotesize 0.0474&\footnotesize 0.0377&\footnotesize 0.0282&\footnotesize 0.0196\\

&&\multicolumn{1}{|c|}{\footnotesize MAE}&\footnotesize0.0070&\footnotesize 0.0190&\footnotesize 0.0069&\footnotesize 0.0355&\footnotesize 0.0322&\footnotesize 0.0255&\footnotesize 0.0185&\footnotesize 0.0123\\

\cline{2-11}


&\multirow{3}{*}{\footnotesize Laplacian}& \multicolumn{1}{|c|}{\footnotesize \(R^2\)}&\footnotesize\textbf{1.0000}&\footnotesize \textbf{1.0000}&\footnotesize \textbf{1.0000}&\footnotesize \textbf{1.0000}&\footnotesize \textbf{1.0000}&\footnotesize \textbf{1.0000}&\footnotesize \textbf{1.0000}&\footnotesize \textbf{1.0000}\\

&&\multicolumn{1}{|c|}{\footnotesize RMSE}&\footnotesize\textbf{0.0000}&\footnotesize \textbf{0.0000}&\footnotesize \textbf{0.0000}&\footnotesize \textbf{0.0000}&\footnotesize \textbf{0.0000}&\footnotesize \textbf{0.0000}&\footnotesize \textbf{0.0000}&\footnotesize \textbf{0.0000}\\

&&\multicolumn{1}{|c|}{\footnotesize MAE}&\footnotesize\textbf{0.0000}&\footnotesize \textbf{0.0000}&\footnotesize \textbf{0.0000}&\footnotesize \textbf{0.0000}&\footnotesize \textbf{0.0000}&\footnotesize \textbf{0.0000}&\footnotesize \textbf{0.0000}&\footnotesize \textbf{0.0000}\\

\cline{2-11}


&\multirow{3}{*}{\footnotesize Sigmoid}& \multicolumn{1}{|c|}{\footnotesize \(R^2\)}&\footnotesize0.6817&\footnotesize 0.9448&\footnotesize 0.7911&\footnotesize 0.9070&\footnotesize 0.8992&\footnotesize 0.9353&\footnotesize 0.9559&\footnotesize 0.6980\\

&&\multicolumn{1}{|c|}{\footnotesize RMSE}&\footnotesize0.2169&\footnotesize 0.0894&\footnotesize 0.1814&\footnotesize 0.1245&\footnotesize 0.1321&\footnotesize 0.1071&\footnotesize 0.0888&\footnotesize 0.2324\\

&&\multicolumn{1}{|c|}{\footnotesize MAE}&\footnotesize0.1649&\footnotesize 0.0693&\footnotesize 0.1367&\footnotesize 0.0965&\footnotesize 0.1015&\footnotesize 0.0839&\footnotesize 0.0691&\footnotesize 0.1849\\

\cline{2-11}


&\multirow{3}{*}{\footnotesize Cosine}& \multicolumn{1}{|c|}{\footnotesize \(R^2\)}&\footnotesize-1.5528&\footnotesize -2.0398&\footnotesize -1.9766&\footnotesize -1.8849&\footnotesize -1.7718&\footnotesize -1.6290&\footnotesize -1.4689&\footnotesize -1.3090\\

&&\multicolumn{1}{|c|}{\footnotesize RMSE}&\footnotesize0.6143&\footnotesize 0.6637&\footnotesize 0.6849&\footnotesize 0.6936&\footnotesize 0.6929&\footnotesize 0.6828&\footnotesize 0.6649&\footnotesize 0.6427\\

&&\multicolumn{1}{|c|}{\footnotesize MAE}&\footnotesize0.5001&\footnotesize 0.5450&\footnotesize 0.5582&\footnotesize 0.5607&\footnotesize 0.5541&\footnotesize 0.5375&\footnotesize 0.5130&\footnotesize 0.4840\\

\cline{1-11}

\end{tabular}
\end{adjustbox}
\caption{\textbf{(\(\mathscr{O}_{XZ}\) with \(\mathscr{H}_3\))} Learning performance metrics for the models trained on all 1000 training datapoints using the hyperparameters listed in Table \ref{BestHyperparametersH3}. The metrics are calculated based on the predictions made by the models on the training and testing sets with labels determined by \(\mathscr{O}_{XZ}\) with \(\mathscr{H}_3\). The best values for each qubit number and each metric are in bold text for the training and testing datasets.}
\label{ResultsOXZH3}
\end{table}


\subsection{Results for \(\mathscr{O}_{XZ}\) and \(\mathscr{H}_4\)}

\begin{table}[H]
\begin{adjustbox}{center,max width=\textwidth}
\begin{tabular}{| c | c c c c c c c c c c |}

\cline{4-11}

\multicolumn{3}{c|}{}&\multicolumn{8}{c|}{Number of qubits \((n)\)}\\

\cline{1-11}

\multicolumn{1}{|c|}{\textbf{Dataset}}&\multicolumn{1}{c|}{\textbf{Kernel}}& \multicolumn{1}{c|}{\textbf{Metric}}
& \multicolumn{1}{c|}{5}& \multicolumn{1}{c|}{10}
& \multicolumn{1}{c|}{15}& \multicolumn{1}{c|}{20}
& \multicolumn{1}{c|}{25}& \multicolumn{1}{c|}{30}
& \multicolumn{1}{c|}{35}& \multicolumn{1}{c|}{40}\\

\cline{1-11}


\multirow{18}{*}{Testing}&\multirow{3}{*}{\footnotesize Linear}& \multicolumn{1}{|c|}{\footnotesize \(R^2\)}&\footnotesize-5.4877&\footnotesize-9.7666&\footnotesize-12.2134&\footnotesize-13.2260&\footnotesize-13.3083&\footnotesize-12.7522&\footnotesize-11.8802&\footnotesize-10.6947\\

&&\multicolumn{1}{|c|}{\footnotesize RMSE}&\footnotesize0.8851&\footnotesize0.9240&\footnotesize0.9337&\footnotesize0.9401&\footnotesize0.9424&\footnotesize0.9408&\footnotesize0.9360&\footnotesize0.9293\\

&&\multicolumn{1}{|c|}{\footnotesize MAE}&\footnotesize0.8375&\footnotesize0.8844&\footnotesize0.8976&\footnotesize0.9061&\footnotesize0.9084&\footnotesize0.9053&\footnotesize0.8982&\footnotesize0.8878\\

\cline{2-11}


&\multirow{3}{*}{\footnotesize Polynomial}&\multicolumn{1}{|c|}{\footnotesize \(R^2\)}&\footnotesize0.7034&\footnotesize0.6119&\footnotesize0.6135&\footnotesize0.6729&\footnotesize0.6911&\footnotesize0.7070&\footnotesize0.7900&\footnotesize0.7817\\

&&\multicolumn{1}{|c|}{\footnotesize RMSE}&\footnotesize0.1892&\footnotesize0.1754&\footnotesize0.1597&\footnotesize0.1425&\footnotesize0.1385&\footnotesize0.1373&\footnotesize0.1195&\footnotesize0.1270\\

&&\multicolumn{1}{|c|}{\footnotesize MAE}&\footnotesize0.1060&\footnotesize0.0914&\footnotesize0.0822&\footnotesize0.0782&\footnotesize0.0745&\footnotesize0.0814&\footnotesize0.0690&\footnotesize0.0742\\

\cline{2-11}


&\multirow{3}{*}{\footnotesize RBF}&\multicolumn{1}{|c|}{\footnotesize \(R^2\)}&\footnotesize0.8360&\footnotesize0.7167&\footnotesize\textbf{0.7751}&\footnotesize\textbf{0.8281}&\footnotesize0.7555&\footnotesize\textbf{0.8177}&\footnotesize\textbf{0.8674}&\footnotesize\textbf{0.8931}\\

&&\multicolumn{1}{|c|}{\footnotesize RMSE}&\footnotesize0.1407&\footnotesize0.1499&\footnotesize\textbf{0.1218}&\footnotesize\textbf{0.1033}&\footnotesize0.1232&\footnotesize\textbf{0.1083}&\footnotesize\textbf{0.0950}&\footnotesize\textbf{0.0888}\\

&&\multicolumn{1}{|c|}{\footnotesize MAE}&\footnotesize0.0724&\footnotesize0.0731&\footnotesize\textbf{0.0522}&\footnotesize\textbf{0.0498}&\footnotesize\textbf{0.0650}&\footnotesize\textbf{0.0577}&\footnotesize\textbf{0.0507}&\footnotesize\textbf{0.0468}\\

\cline{2-11}


&\multirow{3}{*}{\footnotesize Laplacian}& \multicolumn{1}{|c|}{\footnotesize \(R^2\)}&\footnotesize\textbf{0.9214}&\footnotesize\textbf{0.8612}&\footnotesize0.7511&\footnotesize0.7732&\footnotesize\textbf{0.7835}&\footnotesize0.8140&\footnotesize0.8521&\footnotesize0.8817\\

&&\multicolumn{1}{|c|}{\footnotesize RMSE}&\footnotesize\textbf{0.0974}&\footnotesize\textbf{0.1049}&\footnotesize0.1282&\footnotesize0.1187&\footnotesize\textbf{0.1159}&\footnotesize0.1094&\footnotesize0.1003&\footnotesize0.0934\\

&&\multicolumn{1}{|c|}{\footnotesize MAE}&\footnotesize\textbf{0.0514}&\footnotesize\textbf{0.0521}&\footnotesize0.0715&\footnotesize0.0686&\footnotesize0.0658&\footnotesize0.0623&\footnotesize0.0575&\footnotesize0.0536\\

\cline{2-11}


&\multirow{3}{*}{\footnotesize Sigmoid}& \multicolumn{1}{|c|}{\footnotesize \(R^2\)}&\footnotesize0.6985&\footnotesize0.6030&\footnotesize0.6269&\footnotesize0.4990&\footnotesize0.5695&\footnotesize-3.7352&\footnotesize0.3487&\footnotesize-11636.6066\\

&&\multicolumn{1}{|c|}{\footnotesize RMSE}&\footnotesize0.1908&\footnotesize0.1774&\footnotesize0.1569&\footnotesize0.1764&\footnotesize0.1635&\footnotesize0.5520&\footnotesize0.2105&\footnotesize29.3142\\

&&\multicolumn{1}{|c|}{\footnotesize MAE}&\footnotesize0.0981&\footnotesize0.0940&\footnotesize0.0815&\footnotesize0.1014&\footnotesize0.1055&\footnotesize0.3001&\footnotesize0.1254&\footnotesize17.0978\\

\cline{2-11}


&\multirow{3}{*}{\footnotesize Cosine}& \multicolumn{1}{|c|}{\footnotesize \(R^2\)}&\footnotesize-5.4872&\footnotesize-9.7161&\footnotesize-12.2097&\footnotesize-13.1944&\footnotesize-13.2552&\footnotesize-12.6843&\footnotesize-11.8028&\footnotesize-10.6139\\

&&\multicolumn{1}{|c|}{\footnotesize RMSE}&\footnotesize0.8851&\footnotesize0.9218&\footnotesize0.9336&\footnotesize0.9391&\footnotesize0.9407&\footnotesize0.9384&\footnotesize0.9332&\footnotesize0.9261\\

&&\multicolumn{1}{|c|}{\footnotesize MAE}&\footnotesize0.8376&\footnotesize0.8821&\footnotesize0.8973&\footnotesize0.9050&\footnotesize0.9068&\footnotesize0.9033&\footnotesize0.8959&\footnotesize0.8851\\

\cline{1-11}


\multirow{18}{*}{Training}&\multirow{3}{*}{\footnotesize Linear}& \multicolumn{1}{|c|}{\footnotesize \(R^2\)}&\footnotesize-6.5416&\footnotesize -11.1825&\footnotesize -12.0591&\footnotesize -12.1214&\footnotesize -11.7599&\footnotesize -11.1197&\footnotesize -10.3358&\footnotesize -9.4707\\

&&\multicolumn{1}{|c|}{\footnotesize RMSE}&\footnotesize0.8912&\footnotesize 0.9224&\footnotesize 0.9317&\footnotesize 0.9352&\footnotesize 0.9363&\footnotesize 0.9345&\footnotesize 0.9300&\footnotesize 0.9233\\

&&\multicolumn{1}{|c|}{\footnotesize MAE}&\footnotesize0.8470&\footnotesize 0.8860&\footnotesize 0.8965&\footnotesize 0.9001&\footnotesize 0.9002&\footnotesize 0.8968&\footnotesize 0.8903&\footnotesize 0.8810\\

\cline{2-11}


&\multirow{3}{*}{\footnotesize Polynomial}&\multicolumn{1}{|c|}{\footnotesize \(R^2\)}&\footnotesize0.7641&\footnotesize 0.8003&\footnotesize 0.8095&\footnotesize 0.8358&\footnotesize 0.8196&\footnotesize 0.8226&\footnotesize 0.8916&\footnotesize 0.8879\\

&&\multicolumn{1}{|c|}{\footnotesize RMSE}&\footnotesize0.1576&\footnotesize 0.1181&\footnotesize 0.1125&\footnotesize 0.1046&\footnotesize 0.1113&\footnotesize 0.1131&\footnotesize 0.0909&\footnotesize 0.0955\\

&&\multicolumn{1}{|c|}{\footnotesize MAE}&\footnotesize0.0884&\footnotesize 0.0656&\footnotesize 0.0616&\footnotesize 0.0616&\footnotesize 0.0641&\footnotesize 0.0704&\footnotesize 0.0514&\footnotesize 0.0555\\

\cline{2-11}


&\multirow{3}{*}{\footnotesize RBF}&\multicolumn{1}{|c|}{\footnotesize \(R^2\)}&\footnotesize0.9374&\footnotesize 0.9012&\footnotesize 0.9703&\footnotesize 0.9927&\footnotesize 0.8783&\footnotesize 0.9120&\footnotesize 0.9443&\footnotesize 0.9666\\

&&\multicolumn{1}{|c|}{\footnotesize RMSE}&\footnotesize0.0812&\footnotesize 0.0831&\footnotesize 0.0444&\footnotesize 0.0221&\footnotesize 0.0914&\footnotesize 0.0796&\footnotesize 0.0652&\footnotesize 0.0521\\

&&\multicolumn{1}{|c|}{\footnotesize MAE}&\footnotesize0.0434&\footnotesize 0.0440&\footnotesize 0.0226&\footnotesize 0.0103&\footnotesize 0.0510&\footnotesize 0.0429&\footnotesize 0.0349&\footnotesize 0.0275\\

\cline{2-11}


&\multirow{3}{*}{\footnotesize Laplacian}& \multicolumn{1}{|c|}{\footnotesize \(R^2\)}&\footnotesize\textbf{1.0000}&\footnotesize \textbf{1.0000}&\footnotesize \textbf{0.9991}&\footnotesize \textbf{0.9996}&\footnotesize \textbf{0.9998}&\footnotesize \textbf{0.9999}&\footnotesize \textbf{1.0000}&\footnotesize \textbf{1.0000}\\

&&\multicolumn{1}{|c|}{\footnotesize RMSE}&\footnotesize\textbf{0.0000}&\footnotesize \textbf{0.0000}&\footnotesize \textbf{0.0076}&\footnotesize \textbf{0.0048}&\footnotesize \textbf{0.0034}&\footnotesize \textbf{0.0026}&\footnotesize \textbf{0.0019}&\footnotesize \textbf{0.0002}\\

&&\multicolumn{1}{|c|}{\footnotesize MAE}&\footnotesize\textbf{0.0000}&\footnotesize \textbf{0.0000}&\footnotesize \textbf{0.0051}&\footnotesize \textbf{0.0032}&\footnotesize \textbf{0.0022}&\footnotesize \textbf{0.0016}&\footnotesize \textbf{0.0012}&\footnotesize \textbf{0.0001}\\

\cline{2-11}


&\multirow{3}{*}{\footnotesize Sigmoid}& \multicolumn{1}{|c|}{\footnotesize \(R^2\)}&\footnotesize0.8788&\footnotesize 0.7988&\footnotesize 0.8540&\footnotesize 0.5826&\footnotesize 0.6472&\footnotesize -4.3227&\footnotesize 0.4554&\footnotesize -13333.5415\\

&&\multicolumn{1}{|c|}{\footnotesize RMSE}&\footnotesize0.1130&\footnotesize 0.1185&\footnotesize 0.0985&\footnotesize 0.1668&\footnotesize 0.1557&\footnotesize 0.6193&\footnotesize 0.2038&\footnotesize 32.9478\\

&&\multicolumn{1}{|c|}{\footnotesize MAE}&\footnotesize0.0658&\footnotesize 0.0673&\footnotesize 0.0550&\footnotesize 0.0976&\footnotesize 0.0989&\footnotesize 0.3556&\footnotesize 0.1172&\footnotesize 20.1928\\

\cline{2-11}


&\multirow{3}{*}{\footnotesize Cosine}& \multicolumn{1}{|c|}{\footnotesize \(R^2\)}&\footnotesize-6.5589&\footnotesize -11.2173&\footnotesize -12.0804&\footnotesize -12.1504&\footnotesize -11.7911&\footnotesize -11.1499&\footnotesize -10.3634&\footnotesize -9.4954\\

&&\multicolumn{1}{|c|}{\footnotesize RMSE}&\footnotesize0.8922&\footnotesize 0.9237&\footnotesize 0.9325&\footnotesize 0.9363&\footnotesize 0.9374&\footnotesize 0.9357&\footnotesize 0.9311&\footnotesize 0.9243\\

&&\multicolumn{1}{|c|}{\footnotesize MAE}&\footnotesize0.8483&\footnotesize 0.8874&\footnotesize 0.8969&\footnotesize 0.9003&\footnotesize 0.9003&\footnotesize 0.8967&\footnotesize 0.8896&\footnotesize 0.8797\\

\cline{1-11}

\end{tabular}
\end{adjustbox}
\caption{\textbf{(\(\mathscr{O}_{XZ}\) with \(\mathscr{H}_4\))} Learning performance metrics for the models trained on all 1000 training datapoints using the hyperparameters listed in Table \ref{BestHyperparametersH4}. The metrics are calculated based on the predictions made by the models on the training and testing sets with labels determined by \(\mathscr{O}_{XZ}\) with \(\mathscr{H}_4\). The best values for each qubit number and each metric are in bold text for the training and testing datasets.}
\label{ResultsOXZH4}
\end{table}


\subsection{Results for \(\mathscr{O}_{Sum}\) and \(\mathscr{H}_1\)}

\begin{table}[H]
\begin{adjustbox}{center,max width=\textwidth}
\begin{tabular}{| c | c c c c c c c c c c |}

\cline{4-11}

\multicolumn{3}{c|}{}&\multicolumn{8}{c|}{Number of qubits \((n)\)}\\

\cline{1-11}

\multicolumn{1}{|c|}{\textbf{Dataset}}&\multicolumn{1}{c|}{\textbf{Kernel}}& \multicolumn{1}{c|}{\textbf{Metric}}
& \multicolumn{1}{c|}{5}& \multicolumn{1}{c|}{10}
& \multicolumn{1}{c|}{15}& \multicolumn{1}{c|}{20}
& \multicolumn{1}{c|}{25}& \multicolumn{1}{c|}{30}
& \multicolumn{1}{c|}{35}& \multicolumn{1}{c|}{40}\\

\cline{1-11}


\multirow{18}{*}{Testing}&\multirow{3}{*}{\footnotesize Linear}& \multicolumn{1}{|c|}{\footnotesize \(R^2\)}&\footnotesize-2.6121&\footnotesize-3.7405&\footnotesize-3.5677&\footnotesize-3.5157&\footnotesize-3.3017&\footnotesize-3.1259&\footnotesize-2.8970&\footnotesize-2.6603\\

&&\multicolumn{1}{|c|}{\footnotesize RMSE}&\footnotesize6.6797&\footnotesize7.2117&\footnotesize7.2580&\footnotesize7.3482&\footnotesize7.2775&\footnotesize7.2564&\footnotesize7.1498&\footnotesize7.0573\\

&&\multicolumn{1}{|c|}{\footnotesize MAE}&\footnotesize5.8351&\footnotesize6.4148&\footnotesize6.4226&\footnotesize6.5026&\footnotesize6.3839&\footnotesize6.3326&\footnotesize6.1876&\footnotesize6.0515\\

\cline{2-11}


&\multirow{3}{*}{\footnotesize Polynomial}&\multicolumn{1}{|c|}{\footnotesize \(R^2\)}&\footnotesize0.9305&\footnotesize0.8697&\footnotesize0.8734&\footnotesize0.8188&\footnotesize0.8258&\footnotesize0.8168&\footnotesize0.8246&\footnotesize0.8211\\

&&\multicolumn{1}{|c|}{\footnotesize RMSE}&\footnotesize0.9266&\footnotesize1.1954&\footnotesize1.2084&\footnotesize1.4719&\footnotesize1.4644&\footnotesize1.5292&\footnotesize1.5170&\footnotesize1.5602\\

&&\multicolumn{1}{|c|}{\footnotesize MAE}&\footnotesize0.6546&\footnotesize0.8219&\footnotesize0.8558&\footnotesize1.0416&\footnotesize1.0461&\footnotesize1.0767&\footnotesize1.0694&\footnotesize1.1116\\

\cline{2-11}


&\multirow{3}{*}{\footnotesize RBF}&\multicolumn{1}{|c|}{\footnotesize \(R^2\)}&\footnotesize0.9556&\footnotesize0.9289&\footnotesize0.9315&\footnotesize0.8979&\footnotesize0.8985&\footnotesize0.9094&\footnotesize0.9208&\footnotesize0.9283\\

&&\multicolumn{1}{|c|}{\footnotesize RMSE}&\footnotesize0.7402&\footnotesize0.8833&\footnotesize0.8885&\footnotesize1.1049&\footnotesize1.1179&\footnotesize1.0755&\footnotesize1.0196&\footnotesize0.9878\\

&&\multicolumn{1}{|c|}{\footnotesize MAE}&\footnotesize0.5409&\footnotesize0.4966&\footnotesize0.4385&\footnotesize0.6916&\footnotesize0.6624&\footnotesize0.6367&\footnotesize0.5657&\footnotesize0.5190\\

\cline{2-11}


&\multirow{3}{*}{\footnotesize Laplacian}& \multicolumn{1}{|c|}{\footnotesize \(R^2\)}&\footnotesize\textbf{0.9621}&\footnotesize\textbf{0.9652}&\footnotesize\textbf{0.9578}&\footnotesize\textbf{0.9555}&\footnotesize\textbf{0.9532}&\footnotesize\textbf{0.9429}&\footnotesize\textbf{0.9503}&\footnotesize\textbf{0.9596}\\

&&\multicolumn{1}{|c|}{\footnotesize RMSE}&\footnotesize\textbf{0.6841}&\footnotesize\textbf{0.6176}&\footnotesize\textbf{0.6979}&\footnotesize\textbf{0.7291}&\footnotesize\textbf{0.7592}&\footnotesize\textbf{0.8534}&\footnotesize\textbf{0.8078}&\footnotesize\textbf{0.7417}\\

&&\multicolumn{1}{|c|}{\footnotesize MAE}&\footnotesize\textbf{0.4643}&\footnotesize\textbf{0.3579}&\footnotesize\textbf{0.3771}&\footnotesize\textbf{0.3990}&\footnotesize\textbf{0.4168}&\footnotesize\textbf{0.5653}&\footnotesize\textbf{0.5250}&\footnotesize\textbf{0.4746}\\

\cline{2-11}


&\multirow{3}{*}{\footnotesize Sigmoid}&\multicolumn{1}{|c|}{\footnotesize \(R^2\)}&\footnotesize0.7472&\footnotesize0.8715&\footnotesize0.6827&\footnotesize0.8175&\footnotesize0.8290&\footnotesize0.8526&\footnotesize0.6879&\footnotesize0.6413\\

&&\multicolumn{1}{|c|}{\footnotesize RMSE}&\footnotesize1.7673&\footnotesize1.1875&\footnotesize1.9128&\footnotesize1.4774&\footnotesize1.4508&\footnotesize1.3716&\footnotesize2.0234&\footnotesize2.2092\\

&&\multicolumn{1}{|c|}{\footnotesize MAE}&\footnotesize1.3062&\footnotesize0.7939&\footnotesize1.4378&\footnotesize1.0354&\footnotesize1.0400&\footnotesize0.9645&\footnotesize1.5305&\footnotesize1.7456\\

\cline{2-11}


&\multirow{3}{*}{\footnotesize Cosine}& \multicolumn{1}{|c|}{\footnotesize \(R^2\)}&\footnotesize-2.6387&\footnotesize-3.7354&\footnotesize-3.5476&\footnotesize-3.4864&\footnotesize-3.2995&\footnotesize-3.1673&\footnotesize-2.9321&\footnotesize-2.6887\\

&&\multicolumn{1}{|c|}{\footnotesize RMSE}&\footnotesize6.7042&\footnotesize7.2079&\footnotesize7.2421&\footnotesize7.3244&\footnotesize7.2756&\footnotesize7.2927&\footnotesize7.1820&\footnotesize7.0846\\

&&\multicolumn{1}{|c|}{\footnotesize MAE}&\footnotesize5.8541&\footnotesize6.4031&\footnotesize6.3958&\footnotesize6.4589&\footnotesize6.3751&\footnotesize6.3714&\footnotesize6.2191&\footnotesize6.0739\\

\cline{1-11}


\multirow{18}{*}{Training}&\multirow{3}{*}{\footnotesize Linear}& \multicolumn{1}{|c|}{\footnotesize \(R^2\)}&\footnotesize-2.1828&\footnotesize -3.1641&\footnotesize -3.0659&\footnotesize -2.9840&\footnotesize -2.8636&\footnotesize -2.6929&\footnotesize -2.5129&\footnotesize -2.3154\\

&&\multicolumn{1}{|c|}{\footnotesize RMSE}&\footnotesize6.4890&\footnotesize 7.0385&\footnotesize 7.0937&\footnotesize 7.1710&\footnotesize 7.1516&\footnotesize 7.1017&\footnotesize 6.9922&\footnotesize 6.8715\\

&&\multicolumn{1}{|c|}{\footnotesize MAE}&\footnotesize5.5720&\footnotesize 6.1532&\footnotesize 6.1925&\footnotesize 6.2540&\footnotesize 6.1792&\footnotesize 6.1015&\footnotesize 5.9658&\footnotesize 5.8112\\

\cline{2-11}


&\multirow{3}{*}{\footnotesize Polynomial}&\multicolumn{1}{|c|}{\footnotesize \(R^2\)}&\footnotesize0.9596&\footnotesize 0.9364&\footnotesize 0.9343&\footnotesize 0.8788&\footnotesize 0.8862&\footnotesize 0.8801&\footnotesize 0.8902&\footnotesize 0.8845\\

&&\multicolumn{1}{|c|}{\footnotesize RMSE}&\footnotesize0.7307&\footnotesize 0.8701&\footnotesize 0.9015&\footnotesize 1.2510&\footnotesize 1.2271&\footnotesize 1.2795&\footnotesize 1.2360&\footnotesize 1.2827\\

&&\multicolumn{1}{|c|}{\footnotesize MAE}&\footnotesize0.5112&\footnotesize 0.6282&\footnotesize 0.6547&\footnotesize 0.8919&\footnotesize 0.8785&\footnotesize 0.8960&\footnotesize 0.8745&\footnotesize 0.9044\\

\cline{2-11}


&\multirow{3}{*}{\footnotesize RBF}&\multicolumn{1}{|c|}{\footnotesize \(R^2\)}&\footnotesize0.9753&\footnotesize 0.9822&\footnotesize 0.9990&\footnotesize 0.9603&\footnotesize 0.9754&\footnotesize 0.9715&\footnotesize 0.9825&\footnotesize 0.9880\\

&&\multicolumn{1}{|c|}{\footnotesize RMSE}&\footnotesize0.5713&\footnotesize 0.4603&\footnotesize 0.1088&\footnotesize 0.7159&\footnotesize 0.5706&\footnotesize 0.6236&\footnotesize 0.4932&\footnotesize 0.4137\\

&&\multicolumn{1}{|c|}{\footnotesize MAE}&\footnotesize0.4250&\footnotesize 0.3093&\footnotesize 0.0643&\footnotesize 0.4913&\footnotesize 0.3904&\footnotesize 0.4216&\footnotesize 0.3307&\footnotesize 0.2663\\

\cline{2-11}


&\multirow{3}{*}{\footnotesize Laplacian}& \multicolumn{1}{|c|}{\footnotesize \(R^2\)}&\footnotesize\textbf{1.0000}&\footnotesize \textbf{1.0000}&\footnotesize \textbf{1.0000}&\footnotesize \textbf{1.0000}&\footnotesize \textbf{1.0000}&\footnotesize \textbf{0.9997}&\footnotesize \textbf{0.9998}&\footnotesize \textbf{0.9999}\\

&&\multicolumn{1}{|c|}{\footnotesize RMSE}&\footnotesize\textbf{0.0000}&\footnotesize \textbf{0.0000}&\footnotesize \textbf{0.0000}&\footnotesize \textbf{0.0000}&\footnotesize \textbf{0.0000}&\footnotesize \textbf{0.0648}&\footnotesize \textbf{0.0492}&\footnotesize \textbf{0.0397}\\

&&\multicolumn{1}{|c|}{\footnotesize MAE}&\footnotesize\textbf{0.0000}&\footnotesize \textbf{0.0000}&\footnotesize \textbf{0.0000}&\footnotesize \textbf{0.0000}&\footnotesize \textbf{0.0000}&\footnotesize \textbf{0.0463}&\footnotesize \textbf{0.0352}&\footnotesize \textbf{0.0279}\\

\cline{2-11}


&\multirow{3}{*}{\footnotesize Sigmoid}& \multicolumn{1}{|c|}{\footnotesize \(R^2\)}&\footnotesize0.7473&\footnotesize 0.9429&\footnotesize 0.7354&\footnotesize 0.8782&\footnotesize 0.8768&\footnotesize 0.9162&\footnotesize 0.7237&\footnotesize 0.6295\\

&&\multicolumn{1}{|c|}{\footnotesize RMSE}&\footnotesize1.8283&\footnotesize 0.8245&\footnotesize 1.8096&\footnotesize 1.2541&\footnotesize 1.2770&\footnotesize 1.0701&\footnotesize 1.9611&\footnotesize 2.2970\\

&&\multicolumn{1}{|c|}{\footnotesize MAE}&\footnotesize1.3103&\footnotesize 0.5973&\footnotesize 1.3112&\footnotesize 0.9025&\footnotesize 0.9082&\footnotesize 0.7856&\footnotesize 1.4166&\footnotesize 1.7306\\

\cline{2-11}


&\multirow{3}{*}{\footnotesize Cosine}& \multicolumn{1}{|c|}{\footnotesize \(R^2\)}&\footnotesize-2.1793&\footnotesize -3.1696&\footnotesize -3.0751&\footnotesize -2.9934&\footnotesize -2.8659&\footnotesize -2.6827&\footnotesize -2.5044&\footnotesize -2.3079\\

&&\multicolumn{1}{|c|}{\footnotesize RMSE}&\footnotesize6.4854&\footnotesize 7.0431&\footnotesize 7.1016&\footnotesize 7.1794&\footnotesize 7.1538&\footnotesize 7.0918&\footnotesize 6.9837&\footnotesize 6.8637\\

&&\multicolumn{1}{|c|}{\footnotesize MAE}&\footnotesize5.5633&\footnotesize 6.1464&\footnotesize 6.1758&\footnotesize 6.2249&\footnotesize 6.1695&\footnotesize 6.1030&\footnotesize 5.9616&\footnotesize 5.8011\\

\cline{1-11}

\end{tabular}
\end{adjustbox}
\caption{\textbf{(\(\mathscr{O}_{Sum}\) with \(\mathscr{H}_1\))} Learning performance metrics for the models trained on all 1000 training datapoints using the hyperparameters listed in Table \ref{BestHyperparametersH1}. The metrics are calculated based on the predictions made by the models on the training and testing sets with labels determined by \(\mathscr{O}_{Sum}\) with \(\mathscr{H}_1\). The best values for each qubit number and each metric are in bold text for the training and testing datasets.}
\label{ResultsOSumH1}
\end{table}


\subsection{Results for \(\mathscr{O}_{Sum}\) and \(\mathscr{H}_2\)}

\begin{table}[H]
\begin{adjustbox}{center,max width=\textwidth}
\begin{tabular}{| c | c c c c c c c c c c |}

\cline{4-11}

\multicolumn{3}{c|}{}&\multicolumn{8}{c|}{Number of qubits \((n)\)}\\

\cline{1-11}

\multicolumn{1}{|c|}{\textbf{Dataset}}&\multicolumn{1}{c|}{\textbf{Kernel}}& \multicolumn{1}{c|}{\textbf{Metric}}
& \multicolumn{1}{c|}{5}& \multicolumn{1}{c|}{10}
& \multicolumn{1}{c|}{15}& \multicolumn{1}{c|}{20}
& \multicolumn{1}{c|}{25}& \multicolumn{1}{c|}{30}
& \multicolumn{1}{c|}{35}& \multicolumn{1}{c|}{40}\\

\cline{1-11}


\multirow{18}{*}{Testing}&\multirow{3}{*}{\footnotesize Linear}& \multicolumn{1}{|c|}{\footnotesize \(R^2\)}&\footnotesize-2.4632&\footnotesize-4.3090&\footnotesize-5.1549&\footnotesize-5.2246&\footnotesize-5.0555&\footnotesize-4.7658&\footnotesize-4.5452&\footnotesize-4.3089\\

&&\multicolumn{1}{|c|}{\footnotesize RMSE}&\footnotesize6.6084&\footnotesize6.9384&\footnotesize7.1717&\footnotesize7.2117&\footnotesize7.2634&\footnotesize7.2450&\footnotesize7.2467&\footnotesize7.2143\\

&&\multicolumn{1}{|c|}{\footnotesize MAE}&\footnotesize5.7357&\footnotesize6.2669&\footnotesize6.5641&\footnotesize6.6096&\footnotesize6.6403&\footnotesize6.5921&\footnotesize6.5668&\footnotesize6.5059\\

\cline{2-11}


&\multirow{3}{*}{\footnotesize Polynomial}&\multicolumn{1}{|c|}{\footnotesize \(R^2\)}&\footnotesize0.8230&\footnotesize0.8447&\footnotesize0.7799&\footnotesize0.8029&\footnotesize0.7957&\footnotesize0.8045&\footnotesize0.6727&\footnotesize0.7990\\

&&\multicolumn{1}{|c|}{\footnotesize RMSE}&\footnotesize1.4942&\footnotesize1.1867&\footnotesize1.3563&\footnotesize1.2833&\footnotesize1.3342&\footnotesize1.3341&\footnotesize1.7604&\footnotesize1.4036\\

&&\multicolumn{1}{|c|}{\footnotesize MAE}&\footnotesize1.1004&\footnotesize0.8443&\footnotesize1.0029&\footnotesize0.9740&\footnotesize1.0096&\footnotesize1.0189&\footnotesize1.2957&\footnotesize1.0740\\

\cline{2-11}


&\multirow{3}{*}{\footnotesize RBF}&\multicolumn{1}{|c|}{\footnotesize \(R^2\)}&\footnotesize\textbf{0.9706}&\footnotesize\textbf{0.9049}&\footnotesize\textbf{0.9022}&\footnotesize0.8957&\footnotesize0.8967&\footnotesize\textbf{0.8965}&\footnotesize\textbf{0.9053}&\footnotesize\textbf{0.9156}\\

&&\multicolumn{1}{|c|}{\footnotesize RMSE}&\footnotesize\textbf{0.6089}&\footnotesize\textbf{0.9284}&\footnotesize\textbf{0.9039}&\footnotesize0.9337&\footnotesize0.9487&\footnotesize\textbf{0.9708}&\footnotesize\textbf{0.9471}&\footnotesize\textbf{0.9099}\\

&&\multicolumn{1}{|c|}{\footnotesize MAE}&\footnotesize\textbf{0.3170}&\footnotesize\textbf{0.6004}&\footnotesize\textbf{0.5720}&\footnotesize0.6676&\footnotesize0.6700&\footnotesize\textbf{0.6674}&\footnotesize\textbf{0.6390}&\footnotesize\textbf{0.6063}\\

\cline{2-11}


&\multirow{3}{*}{\footnotesize Laplacian}& \multicolumn{1}{|c|}{\footnotesize \(R^2\)}&\footnotesize0.9098&\footnotesize0.8973&\footnotesize0.8990&\footnotesize\textbf{0.9088}&\footnotesize\textbf{0.9108}&\footnotesize0.8862&\footnotesize0.8865&\footnotesize0.8897\\

&&\multicolumn{1}{|c|}{\footnotesize RMSE}&\footnotesize1.0664&\footnotesize0.9648&\footnotesize0.9189&\footnotesize\textbf{0.8730}&\footnotesize\textbf{0.8813}&\footnotesize1.0177&\footnotesize1.0370&\footnotesize1.0398\\

&&\multicolumn{1}{|c|}{\footnotesize MAE}&\footnotesize0.7712&\footnotesize0.6657&\footnotesize0.6309&\footnotesize\textbf{0.6007}&\footnotesize\textbf{0.5915}&\footnotesize0.7304&\footnotesize0.7410&\footnotesize0.7411\\

\cline{2-11}


&\multirow{3}{*}{\footnotesize Sigmoid}& \multicolumn{1}{|c|}{\footnotesize \(R^2\)}&\footnotesize0.4967&\footnotesize0.8416&\footnotesize0.5619&\footnotesize0.8142&\footnotesize0.8178&\footnotesize0.8291&\footnotesize0.6329&\footnotesize0.5820\\

&&\multicolumn{1}{|c|}{\footnotesize RMSE}&\footnotesize2.5193&\footnotesize1.1984&\footnotesize1.9133&\footnotesize1.2459&\footnotesize1.2601&\footnotesize1.2474&\footnotesize1.8645&\footnotesize2.0244\\

&&\multicolumn{1}{|c|}{\footnotesize MAE}&\footnotesize1.9383&\footnotesize0.8575&\footnotesize1.3799&\footnotesize0.9470&\footnotesize0.9429&\footnotesize0.8848&\footnotesize1.3602&\footnotesize1.4461\\

\cline{2-11}


&\multirow{3}{*}{\footnotesize Cosine}& \multicolumn{1}{|c|}{\footnotesize \(R^2\)}&\footnotesize-2.4709&\footnotesize-4.3156&\footnotesize-5.1549&\footnotesize-5.2228&\footnotesize-5.0525&\footnotesize-4.7597&\footnotesize-4.5384&\footnotesize-4.2995\\

&&\multicolumn{1}{|c|}{\footnotesize RMSE}&\footnotesize6.6158&\footnotesize6.9427&\footnotesize7.1717&\footnotesize7.2107&\footnotesize7.2616&\footnotesize7.2411&\footnotesize7.2422&\footnotesize7.2080\\

&&\multicolumn{1}{|c|}{\footnotesize MAE}&\footnotesize5.7417&\footnotesize6.2682&\footnotesize6.5637&\footnotesize6.6062&\footnotesize6.6348&\footnotesize6.5832&\footnotesize6.5569&\footnotesize6.4933\\

\cline{1-11}


\multirow{18}{*}{Training}&\multirow{3}{*}{\footnotesize Linear}& \multicolumn{1}{|c|}{\footnotesize \(R^2\)}&\footnotesize-2.5404&\footnotesize -4.3985&\footnotesize -5.3317&\footnotesize -5.2921&\footnotesize -5.1667&\footnotesize -4.8959&\footnotesize -4.6703&\footnotesize -4.4283\\

&&\multicolumn{1}{|c|}{\footnotesize RMSE}&\footnotesize6.6747&\footnotesize 6.9451&\footnotesize 7.2232&\footnotesize 7.2249&\footnotesize 7.2857&\footnotesize 7.2512&\footnotesize 7.2503&\footnotesize 7.2167\\

&&\multicolumn{1}{|c|}{\footnotesize MAE}&\footnotesize5.7924&\footnotesize 6.2699&\footnotesize 6.6278&\footnotesize 6.6248&\footnotesize 6.6670&\footnotesize 6.6053&\footnotesize 6.5781&\footnotesize 6.5165\\

\cline{2-11}


&\multirow{3}{*}{\footnotesize Polynomial}&\multicolumn{1}{|c|}{\footnotesize \(R^2\)}&\footnotesize0.8817&\footnotesize 0.8715&\footnotesize 0.7817&\footnotesize 0.8027&\footnotesize 0.7974&\footnotesize 0.8080&\footnotesize 0.7078&\footnotesize 0.8052\\

&&\multicolumn{1}{|c|}{\footnotesize RMSE}&\footnotesize1.2201&\footnotesize 1.0716&\footnotesize 1.3412&\footnotesize 1.2795&\footnotesize 1.3206&\footnotesize 1.3084&\footnotesize 1.6460&\footnotesize 1.3672\\

&&\multicolumn{1}{|c|}{\footnotesize MAE}&\footnotesize0.9328&\footnotesize 0.7624&\footnotesize 0.9493&\footnotesize 0.9258&\footnotesize 0.9539&\footnotesize 0.9551&\footnotesize 1.2068&\footnotesize 1.0022\\

\cline{2-11}


&\multirow{3}{*}{\footnotesize RBF}&\multicolumn{1}{|c|}{\footnotesize \(R^2\)}&\footnotesize\textbf{0.9999}&\footnotesize 0.9441&\footnotesize 0.9964&\footnotesize 0.9292&\footnotesize 0.9275&\footnotesize 0.9478&\footnotesize 0.9626&\footnotesize 0.9752\\

&&\multicolumn{1}{|c|}{\footnotesize RMSE}&\footnotesize\textbf{0.0252}&\footnotesize 0.7066&\footnotesize 0.1727&\footnotesize 0.7663&\footnotesize 0.7902&\footnotesize 0.6823&\footnotesize 0.5885&\footnotesize 0.4879\\

&&\multicolumn{1}{|c|}{\footnotesize MAE}&\footnotesize\textbf{0.0142}&\footnotesize 0.4788&\footnotesize 0.1020&\footnotesize 0.5265&\footnotesize 0.5465&\footnotesize 0.4763&\footnotesize 0.4020&\footnotesize 0.3227\\

\cline{2-11}


&\multirow{3}{*}{\footnotesize Laplacian}& \multicolumn{1}{|c|}{\footnotesize \(R^2\)}&\footnotesize0.9539&\footnotesize \textbf{0.9736}&\footnotesize \textbf{0.9995}&\footnotesize \textbf{0.9998}&\footnotesize \textbf{1.0000}&\footnotesize \textbf{0.9860}&\footnotesize \textbf{0.9891}&\footnotesize \textbf{0.9912}\\

&&\multicolumn{1}{|c|}{\footnotesize RMSE}&\footnotesize0.7618&\footnotesize \textbf{0.4858}&\footnotesize \textbf{0.0639}&\footnotesize \textbf{0.0435}&\footnotesize \textbf{0.0000}&\footnotesize \textbf{0.3533}&\footnotesize \textbf{0.3186}&\footnotesize \textbf{0.2906}\\

&&\multicolumn{1}{|c|}{\footnotesize MAE}&\footnotesize0.5515&\footnotesize \textbf{0.3406}&\footnotesize \textbf{0.0450}&\footnotesize \textbf{0.0301}&\footnotesize \textbf{0.0000}&\footnotesize \textbf{0.2641}&\footnotesize \textbf{0.2397}&\footnotesize \textbf{0.2191}\\

\cline{2-11}


&\multirow{3}{*}{\footnotesize Sigmoid}& \multicolumn{1}{|c|}{\footnotesize \(R^2\)}&\footnotesize0.5367&\footnotesize 0.8811&\footnotesize 0.5941&\footnotesize 0.8180&\footnotesize 0.8082&\footnotesize 0.8815&\footnotesize 0.6817&\footnotesize 0.5779\\

&&\multicolumn{1}{|c|}{\footnotesize RMSE}&\footnotesize2.4146&\footnotesize 1.0306&\footnotesize 1.8288&\footnotesize 1.2287&\footnotesize 1.2850&\footnotesize 1.0282&\footnotesize 1.7178&\footnotesize 2.0124\\

&&\multicolumn{1}{|c|}{\footnotesize MAE}&\footnotesize1.8768&\footnotesize 0.7487&\footnotesize 1.3039&\footnotesize 0.8932&\footnotesize 0.9275&\footnotesize 0.7365&\footnotesize 1.2508&\footnotesize 1.4609\\

\cline{2-11}


&\multirow{3}{*}{\footnotesize Cosine}& \multicolumn{1}{|c|}{\footnotesize \(R^2\)}&\footnotesize-2.5378&\footnotesize -4.4031&\footnotesize -5.3380&\footnotesize -5.3013&\footnotesize -5.1766&\footnotesize -4.9077&\footnotesize -4.6811&\footnotesize -4.4397\\

&&\multicolumn{1}{|c|}{\footnotesize RMSE}&\footnotesize6.6723&\footnotesize 6.9480&\footnotesize 7.2267&\footnotesize 7.2301&\footnotesize 7.2915&\footnotesize 7.2585&\footnotesize 7.2572&\footnotesize 7.2243\\

&&\multicolumn{1}{|c|}{\footnotesize MAE}&\footnotesize5.7902&\footnotesize 6.2753&\footnotesize 6.6324&\footnotesize 6.6318&\footnotesize 6.6754&\footnotesize 6.6159&\footnotesize 6.5880&\footnotesize 6.5268\\

\cline{1-11}

\end{tabular}
\end{adjustbox}
\caption{\textbf{(\(\mathscr{O}_{Sum}\) with \(\mathscr{H}_2\))} Learning performance metrics for the models trained on all 1000 training datapoints using the hyperparameters listed in Table \ref{BestHyperparametersH2}. The metrics are calculated based on the predictions made by the models on the training and testing sets with labels determined by \(\mathscr{O}_{Sum}\) with \(\mathscr{H}_2\). The best values for each qubit number and each metric are in bold text for the training and testing datasets.}
\label{ResultsOSumH2}
\end{table}


\subsection{Results for \(\mathscr{O}_{Sum}\) and \(\mathscr{H}_3\)}

\begin{table}[H]
\begin{adjustbox}{center,max width=\textwidth}
\begin{tabular}{| c | c c c c c c c c c c |}

\cline{4-11}

\multicolumn{3}{c|}{}&\multicolumn{8}{c|}{Number of qubits \((n)\)}\\

\cline{1-11}

\multicolumn{1}{|c|}{\textbf{Dataset}}&\multicolumn{1}{c|}{\textbf{Kernel}}& \multicolumn{1}{c|}{\textbf{Metric}}
& \multicolumn{1}{c|}{5}& \multicolumn{1}{c|}{10}
& \multicolumn{1}{c|}{15}& \multicolumn{1}{c|}{20}
& \multicolumn{1}{c|}{25}& \multicolumn{1}{c|}{30}
& \multicolumn{1}{c|}{35}& \multicolumn{1}{c|}{40}\\

\cline{1-11}


\multirow{18}{*}{Testing}&\multirow{3}{*}{\footnotesize Linear}& \multicolumn{1}{|c|}{\footnotesize \(R^2\)}&\footnotesize-1.8253&\footnotesize-2.0274&\footnotesize-1.9619&\footnotesize-1.8855&\footnotesize-1.7631&\footnotesize-1.5998&\footnotesize-1.4225&\footnotesize-1.2536\\

&&\multicolumn{1}{|c|}{\footnotesize RMSE}&\footnotesize5.2940&\footnotesize5.8909&\footnotesize6.0765&\footnotesize6.1494&\footnotesize6.1253&\footnotesize6.0140&\footnotesize5.8459&\footnotesize5.6489\\

&&\multicolumn{1}{|c|}{\footnotesize MAE}&\footnotesize4.2697&\footnotesize4.8224&\footnotesize4.9492&\footnotesize4.9773&\footnotesize4.9052&\footnotesize4.7378&\footnotesize4.5049&\footnotesize4.2420\\

\cline{2-11}


&\multirow{3}{*}{\footnotesize Polynomial}&\multicolumn{1}{|c|}{\footnotesize \(R^2\)}&\footnotesize0.8647&\footnotesize0.9194&\footnotesize0.9221&\footnotesize0.9244&\footnotesize0.9296&\footnotesize0.9222&\footnotesize0.9160&\footnotesize0.9140\\

&&\multicolumn{1}{|c|}{\footnotesize RMSE}&\footnotesize1.1585&\footnotesize0.9611&\footnotesize0.9857&\footnotesize0.9956&\footnotesize0.9778&\footnotesize1.0404&\footnotesize1.0888&\footnotesize1.1037\\

&&\multicolumn{1}{|c|}{\footnotesize MAE}&\footnotesize0.8594&\footnotesize0.7375&\footnotesize0.7432&\footnotesize0.7518&\footnotesize0.7466&\footnotesize0.7980&\footnotesize0.8239&\footnotesize0.8420\\

\cline{2-11}


&\multirow{3}{*}{\footnotesize RBF}&\multicolumn{1}{|c|}{\footnotesize \(R^2\)}&\footnotesize\textbf{0.9852}&\footnotesize\textbf{0.9867}&\footnotesize\textbf{0.9882}&\footnotesize0.9703&\footnotesize\textbf{0.9776}&\footnotesize\textbf{0.9810}&\footnotesize\textbf{0.9836}&\footnotesize\textbf{0.9850}\\

&&\multicolumn{1}{|c|}{\footnotesize RMSE}&\footnotesize\textbf{0.3834}&\footnotesize\textbf{0.3908}&\footnotesize\textbf{0.3831}&\footnotesize0.6243&\footnotesize\textbf{0.5514}&\footnotesize\textbf{0.5142}&\footnotesize\textbf{0.4808}&\footnotesize\textbf{0.4614}\\

&&\multicolumn{1}{|c|}{\footnotesize MAE}&\footnotesize\textbf{0.2247}&\footnotesize\textbf{0.2704}&\footnotesize\textbf{0.2342}&\footnotesize0.4376&\footnotesize\textbf{0.3829}&\footnotesize\textbf{0.3485}&\footnotesize\textbf{0.3089}&\footnotesize\textbf{0.2831}\\

\cline{2-11}


&\multirow{3}{*}{\footnotesize Laplacian}&\multicolumn{1}{|c|}{\footnotesize \(R^2\)}&\footnotesize0.9637&\footnotesize0.9606&\footnotesize0.9705&\footnotesize\textbf{0.9740}&\footnotesize0.9755&\footnotesize0.9763&\footnotesize0.9764&\footnotesize0.9755\\

&&\multicolumn{1}{|c|}{\footnotesize RMSE}&\footnotesize0.6000&\footnotesize0.6719&\footnotesize0.6060&\footnotesize\textbf{0.5842}&\footnotesize0.5770&\footnotesize0.5740&\footnotesize0.5768&\footnotesize0.5884\\

&&\multicolumn{1}{|c|}{\footnotesize MAE}&\footnotesize0.3753&\footnotesize0.4879&\footnotesize0.4285&\footnotesize\textbf{0.4042}&\footnotesize0.3895&\footnotesize0.3767&\footnotesize0.3644&\footnotesize0.3564\\

\cline{2-11}


&\multirow{3}{*}{\footnotesize Sigmoid}&\multicolumn{1}{|c|}{\footnotesize \(R^2\)}&\footnotesize0.7468&\footnotesize0.9181&\footnotesize0.8022&\footnotesize0.8799&\footnotesize0.8832&\footnotesize0.9097&\footnotesize0.9348&\footnotesize0.7250\\

&&\multicolumn{1}{|c|}{\footnotesize RMSE}&\footnotesize1.5850&\footnotesize0.9691&\footnotesize1.5702&\footnotesize1.2545&\footnotesize1.2594&\footnotesize1.1208&\footnotesize0.9594&\footnotesize1.9734\\

&&\multicolumn{1}{|c|}{\footnotesize MAE}&\footnotesize1.1870&\footnotesize0.7475&\footnotesize1.2129&\footnotesize0.9506&\footnotesize0.9707&\footnotesize0.8719&\footnotesize0.7301&\footnotesize1.6113\\

\cline{2-11}


&\multirow{3}{*}{\footnotesize Cosine}& \multicolumn{1}{|c|}{\footnotesize \(R^2\)}&\footnotesize-1.8257&\footnotesize-2.0275&\footnotesize-1.9619&\footnotesize-1.8855&\footnotesize-1.7633&\footnotesize-1.5999&\footnotesize-1.4225&\footnotesize-1.2535\\

&&\multicolumn{1}{|c|}{\footnotesize RMSE}&\footnotesize5.2944&\footnotesize5.8910&\footnotesize6.0766&\footnotesize6.1494&\footnotesize6.1255&\footnotesize6.0142&\footnotesize5.8459&\footnotesize5.6488\\

&&\multicolumn{1}{|c|}{\footnotesize MAE}&\footnotesize4.2685&\footnotesize4.8210&\footnotesize4.9456&\footnotesize4.9710&\footnotesize4.8940&\footnotesize4.7187&\footnotesize4.4806&\footnotesize4.2142\\

\cline{1-11}


\multirow{18}{*}{Training}&\multirow{3}{*}{\footnotesize Linear}& \multicolumn{1}{|c|}{\footnotesize \(R^2\)}&\footnotesize-1.7817&\footnotesize -2.0309&\footnotesize -1.9525&\footnotesize -1.8691&\footnotesize -1.7526&\footnotesize -1.6012&\footnotesize -1.4342&\footnotesize -1.2712\\

&&\multicolumn{1}{|c|}{\footnotesize RMSE}&\footnotesize5.2967&\footnotesize 5.9302&\footnotesize 6.1264&\footnotesize 6.2033&\footnotesize 6.1787&\footnotesize 6.0624&\footnotesize 5.8818&\footnotesize 5.6703\\

&&\multicolumn{1}{|c|}{\footnotesize MAE}&\footnotesize4.2821&\footnotesize 4.8539&\footnotesize 4.9820&\footnotesize 5.0092&\footnotesize 4.9391&\footnotesize 4.7727&\footnotesize 4.5383&\footnotesize 4.2713\\

\cline{2-11}


&\multirow{3}{*}{\footnotesize Polynomial}&\multicolumn{1}{|c|}{\footnotesize \(R^2\)}&\footnotesize0.9067&\footnotesize 0.9451&\footnotesize 0.9455&\footnotesize 0.9510&\footnotesize 0.9561&\footnotesize 0.9446&\footnotesize 0.9416&\footnotesize 0.9450\\

&&\multicolumn{1}{|c|}{\footnotesize RMSE}&\footnotesize0.9698&\footnotesize 0.7981&\footnotesize 0.8325&\footnotesize 0.8109&\footnotesize 0.7800&\footnotesize 0.8850&\footnotesize 0.9109&\footnotesize 0.8821\\

&&\multicolumn{1}{|c|}{\footnotesize MAE}&\footnotesize0.7195&\footnotesize 0.6049&\footnotesize 0.6223&\footnotesize 0.6265&\footnotesize 0.6016&\footnotesize 0.6786&\footnotesize 0.6957&\footnotesize 0.6847\\

\cline{2-11}


&\multirow{3}{*}{\footnotesize RBF}&\multicolumn{1}{|c|}{\footnotesize \(R^2\)}&\footnotesize0.9999&\footnotesize 0.9954&\footnotesize 0.9992&\footnotesize 0.9861&\footnotesize 0.9885&\footnotesize 0.9931&\footnotesize 0.9962&\footnotesize 0.9981\\

&&\multicolumn{1}{|c|}{\footnotesize RMSE}&\footnotesize0.0333&\footnotesize 0.2315&\footnotesize 0.1039&\footnotesize 0.4314&\footnotesize 0.3995&\footnotesize 0.3123&\footnotesize 0.2316&\footnotesize 0.1623\\

&&\multicolumn{1}{|c|}{\footnotesize MAE}&\footnotesize0.0183&\footnotesize 0.1548&\footnotesize 0.0602&\footnotesize 0.3035&\footnotesize 0.2726&\footnotesize 0.2113&\footnotesize 0.1520&\footnotesize 0.1017\\

\cline{2-11}


&\multirow{3}{*}{\footnotesize Laplacian}&\multicolumn{1}{|c|}{\footnotesize \(R^2\)}&\footnotesize\textbf{1.0000}&\footnotesize \textbf{1.0000}&\footnotesize \textbf{1.0000}&\footnotesize \textbf{1.0000}&\footnotesize \textbf{1.0000}&\footnotesize \textbf{1.0000}&\footnotesize \textbf{1.0000}&\footnotesize \textbf{1.0000}\\

&&\multicolumn{1}{|c|}{\footnotesize RMSE}&\footnotesize\textbf{0.0000}&\footnotesize \textbf{0.0000}&\footnotesize \textbf{0.0000}&\footnotesize \textbf{0.0000}&\footnotesize \textbf{0.0000}&\footnotesize \textbf{0.0000}&\footnotesize \textbf{0.0000}&\footnotesize \textbf{0.0000}\\

&&\multicolumn{1}{|c|}{\footnotesize MAE}&\footnotesize\textbf{0.0000}&\footnotesize \textbf{0.0000}&\footnotesize \textbf{0.0000}&\footnotesize \textbf{0.0000}&\footnotesize \textbf{0.0000}&\footnotesize \textbf{0.0000}&\footnotesize \textbf{0.0000}&\footnotesize \textbf{0.0000}\\

\cline{2-11}

\cline{2-11}


&\multirow{3}{*}{\footnotesize Sigmoid}&\multicolumn{1}{|c|}{\footnotesize \(R^2\)}&\footnotesize0.7305&\footnotesize 0.9498&\footnotesize 0.7981&\footnotesize 0.9104&\footnotesize 0.9024&\footnotesize 0.9368&\footnotesize 0.9564&\footnotesize 0.6995\\

&&\multicolumn{1}{|c|}{\footnotesize RMSE}&\footnotesize1.6486&\footnotesize 0.7634&\footnotesize 1.6019&\footnotesize 1.0965&\footnotesize 1.1632&\footnotesize 0.9449&\footnotesize 0.7869&\footnotesize 2.0626\\

&&\multicolumn{1}{|c|}{\footnotesize MAE}&\footnotesize1.2481&\footnotesize 0.5958&\footnotesize 1.2083&\footnotesize 0.8516&\footnotesize 0.8922&\footnotesize 0.7384&\footnotesize 0.6158&\footnotesize 1.6549\\

\cline{2-11}


&\multirow{3}{*}{\footnotesize Cosine}& \multicolumn{1}{|c|}{\footnotesize \(R^2\)}&\footnotesize-1.7830&\footnotesize -2.0316&\footnotesize -1.9539&\footnotesize -1.8711&\footnotesize -1.7545&\footnotesize -1.6031&\footnotesize -1.4362&\footnotesize -1.2730\\

&&\multicolumn{1}{|c|}{\footnotesize RMSE}&\footnotesize5.2979&\footnotesize 5.9310&\footnotesize 6.1279&\footnotesize 6.2054&\footnotesize 6.1808&\footnotesize 6.0647&\footnotesize 5.8841&\footnotesize 5.6725\\

&&\multicolumn{1}{|c|}{\footnotesize MAE}&\footnotesize4.2842&\footnotesize 4.8553&\footnotesize 4.9839&\footnotesize 5.0095&\footnotesize 4.9333&\footnotesize 4.7599&\footnotesize 4.5186&\footnotesize 4.2461\\

\cline{1-11}

\end{tabular}
\end{adjustbox}
\caption{\textbf{(\(\mathscr{O}_{Sum}\) with \(\mathscr{H}_3\))} Learning performance metrics for the models trained on all 1000 training datapoints using the hyperparameters listed in Table \ref{BestHyperparametersH3}. The metrics are calculated based on the predictions made by the models on the training and testing sets with labels determined by \(\mathscr{O}_{Sum}\) with \(\mathscr{H}_3\). The best values for each qubit number and each metric are in bold text for the training and testing datasets.}
\label{ResultsOSumH3}
\end{table}


\subsection{Results for \(\mathscr{O}_{Sum}\) and \(\mathscr{H}_4\)}

\begin{table}[H]
\begin{adjustbox}{center,max width=\textwidth}
\begin{tabular}{| c | c c c c c c c c c c |}

\cline{4-11}

\multicolumn{3}{c|}{}&\multicolumn{8}{c|}{Number of qubits \((n)\)}\\

\cline{1-11}

\multicolumn{1}{|c|}{\textbf{Dataset}}&\multicolumn{1}{c|}{\textbf{Kernel}}& \multicolumn{1}{c|}{\textbf{Metric}}
& \multicolumn{1}{c|}{5}& \multicolumn{1}{c|}{10}
& \multicolumn{1}{c|}{15}& \multicolumn{1}{c|}{20}
& \multicolumn{1}{c|}{25}& \multicolumn{1}{c|}{30}
& \multicolumn{1}{c|}{35}& \multicolumn{1}{c|}{40}\\

\cline{1-11}


\multirow{18}{*}{Testing}&\multirow{3}{*}{\footnotesize Linear}& \multicolumn{1}{|c|}{\footnotesize \(R^2\)}&\footnotesize-6.0930&\footnotesize-10.4791&\footnotesize-12.5001&\footnotesize-13.3118&\footnotesize-13.1500&\footnotesize-12.4112&\footnotesize-11.3274&\footnotesize-10.0421\\

&&\multicolumn{1}{|c|}{\footnotesize RMSE}&\footnotesize7.8717&\footnotesize8.2751&\footnotesize8.3789&\footnotesize8.4404&\footnotesize8.4590&\footnotesize8.4367&\footnotesize8.3851&\footnotesize8.3171\\

&&\multicolumn{1}{|c|}{\footnotesize MAE}&\footnotesize7.3590&\footnotesize7.8947&\footnotesize8.0577&\footnotesize8.1362&\footnotesize8.1498&\footnotesize8.1097&\footnotesize8.0297&\footnotesize7.9231\\

\cline{2-11}


&\multirow{3}{*}{\footnotesize Polynomial}&\multicolumn{1}{|c|}{\footnotesize \(R^2\)}&\footnotesize0.8593&\footnotesize0.6920&\footnotesize0.6842&\footnotesize0.6860&\footnotesize0.5918&\footnotesize0.7476&\footnotesize0.7734&\footnotesize0.7838\\

&&\multicolumn{1}{|c|}{\footnotesize RMSE}&\footnotesize1.1086&\footnotesize1.3555&\footnotesize1.2815&\footnotesize1.2502&\footnotesize1.4368&\footnotesize1.1573&\footnotesize1.1370&\footnotesize1.1639\\

&&\multicolumn{1}{|c|}{\footnotesize MAE}&\footnotesize0.7069&\footnotesize0.7836&\footnotesize0.6983&\footnotesize0.7379&\footnotesize0.8742&\footnotesize0.6691&\footnotesize0.6764&\footnotesize0.7258\\

\cline{2-11}


&\multirow{3}{*}{\footnotesize RBF}&\multicolumn{1}{|c|}{\footnotesize \(R^2\)}&\footnotesize0.9091&\footnotesize0.7997&\footnotesize\textbf{0.8563}&\footnotesize\textbf{0.8734}&\footnotesize0.7908&\footnotesize0.8480&\footnotesize\textbf{0.8862}&\footnotesize0.9075\\

&&\multicolumn{1}{|c|}{\footnotesize RMSE}&\footnotesize0.8913&\footnotesize1.0930&\footnotesize\textbf{0.8644}&\footnotesize\textbf{0.7940}&\footnotesize1.0286&\footnotesize0.8982&\footnotesize\textbf{0.8055}&\footnotesize0.7614\\

&&\multicolumn{1}{|c|}{\footnotesize MAE}&\footnotesize0.5528&\footnotesize0.5934&\footnotesize\textbf{0.4079}&\footnotesize\textbf{0.3943}&\footnotesize0.5672&\footnotesize\textbf{0.5084}&\footnotesize\textbf{0.4461}&\footnotesize\textbf{0.4183}\\

\cline{2-11}


&\multirow{3}{*}{\footnotesize Laplacian}& \multicolumn{1}{|c|}{\footnotesize \(R^2\)}&\footnotesize\textbf{0.9607}&\footnotesize\textbf{0.9038}&\footnotesize0.8151&\footnotesize0.8103&\footnotesize\textbf{0.8209}&\footnotesize\textbf{0.8532}&\footnotesize0.8848&\footnotesize\textbf{0.9090}\\

&&\multicolumn{1}{|c|}{\footnotesize RMSE}&\footnotesize\textbf{0.5859}&\footnotesize\textbf{0.7574}&\footnotesize0.9806&\footnotesize0.9716&\footnotesize\textbf{0.9516}&\footnotesize\textbf{0.8828}&\footnotesize0.8107&\footnotesize\textbf{0.7550}\\

&&\multicolumn{1}{|c|}{\footnotesize MAE}&\footnotesize\textbf{0.3555}&\footnotesize\textbf{0.4011}&\footnotesize0.5750&\footnotesize0.5742&\footnotesize\textbf{0.5497}&\footnotesize0.5137&\footnotesize0.4757&\footnotesize0.4492\\

\cline{2-11}


&\multirow{3}{*}{\footnotesize Sigmoid}& \multicolumn{1}{|c|}{\footnotesize \(R^2\)}&\footnotesize0.8475&\footnotesize0.6959&\footnotesize0.7261&\footnotesize0.5394&\footnotesize0.6076&\footnotesize-4.1101&\footnotesize0.3548&\footnotesize-7943.9105\\

&&\multicolumn{1}{|c|}{\footnotesize RMSE}&\footnotesize1.1541&\footnotesize1.3469&\footnotesize1.1936&\footnotesize1.5142&\footnotesize1.4086&\footnotesize5.2078&\footnotesize1.9184&\footnotesize223.0957\\

&&\multicolumn{1}{|c|}{\footnotesize MAE}&\footnotesize0.7069&\footnotesize0.7720&\footnotesize0.6515&\footnotesize0.8967&\footnotesize0.9337&\footnotesize2.8464&\footnotesize1.1720&\footnotesize130.1390\\

\cline{2-11}


&\multirow{3}{*}{\footnotesize Cosine}& \multicolumn{1}{|c|}{\footnotesize \(R^2\)}&\footnotesize-6.0862&\footnotesize-10.4191&\footnotesize-12.4909&\footnotesize-13.2760&\footnotesize-13.0931&\footnotesize-12.3407&\footnotesize-11.2493&\footnotesize-9.9617\\

&&\multicolumn{1}{|c|}{\footnotesize RMSE}&\footnotesize7.8679&\footnotesize8.2535&\footnotesize8.3761&\footnotesize8.4298&\footnotesize8.4419&\footnotesize8.4145&\footnotesize8.3585&\footnotesize8.2868\\

&&\multicolumn{1}{|c|}{\footnotesize MAE}&\footnotesize7.3541&\footnotesize7.8779&\footnotesize8.0550&\footnotesize8.1268&\footnotesize8.1352&\footnotesize8.0915&\footnotesize8.0086&\footnotesize7.8983\\

\cline{1-11}


\multirow{18}{*}{Training}&\multirow{3}{*}{\footnotesize Linear}& \multicolumn{1}{|c|}{\footnotesize \(R^2\)}&\footnotesize-6.6620&\footnotesize -10.8265&\footnotesize -11.6110&\footnotesize -11.5868&\footnotesize -11.1922&\footnotesize -10.5742&\footnotesize -9.7730&\footnotesize -8.8826\\

&&\multicolumn{1}{|c|}{\footnotesize RMSE}&\footnotesize7.9094&\footnotesize 8.2592&\footnotesize 8.3579&\footnotesize 8.3957&\footnotesize 8.4016&\footnotesize 8.3794&\footnotesize 8.3295&\footnotesize 8.2601\\

&&\multicolumn{1}{|c|}{\footnotesize MAE}&\footnotesize7.4386&\footnotesize 7.9118&\footnotesize 8.0255&\footnotesize 8.0652&\footnotesize 8.0626&\footnotesize 8.0271&\footnotesize 7.9567&\footnotesize 7.8595\\

\cline{2-11}


&\multirow{3}{*}{\footnotesize Polynomial}&\multicolumn{1}{|c|}{\footnotesize \(R^2\)}&\footnotesize0.9262&\footnotesize 0.8383&\footnotesize 0.8414&\footnotesize 0.8347&\footnotesize 0.7679&\footnotesize 0.8541&\footnotesize 0.8699&\footnotesize 0.8783\\

&&\multicolumn{1}{|c|}{\footnotesize RMSE}&\footnotesize0.7762&\footnotesize 0.9656&\footnotesize 0.9374&\footnotesize 0.9621&\footnotesize 1.1591&\footnotesize 0.9408&\footnotesize 0.9154&\footnotesize 0.9165\\

&&\multicolumn{1}{|c|}{\footnotesize MAE}&\footnotesize0.4871&\footnotesize 0.5608&\footnotesize 0.5223&\footnotesize 0.5788&\footnotesize 0.6837&\footnotesize 0.5476&\footnotesize 0.5393&\footnotesize 0.5459\\

\cline{2-11}


&\multirow{3}{*}{\footnotesize RBF}&\multicolumn{1}{|c|}{\footnotesize \(R^2\)}&\footnotesize0.9681&\footnotesize 0.9275&\footnotesize 0.9772&\footnotesize 0.9941&\footnotesize 0.8921&\footnotesize 0.9228&\footnotesize 0.9514&\footnotesize 0.9698\\

&&\multicolumn{1}{|c|}{\footnotesize RMSE}&\footnotesize0.5106&\footnotesize 0.6467&\footnotesize 0.3556&\footnotesize 0.1823&\footnotesize 0.7904&\footnotesize 0.6843&\footnotesize 0.5594&\footnotesize 0.4568\\

&&\multicolumn{1}{|c|}{\footnotesize MAE}&\footnotesize0.3166&\footnotesize 0.3636&\footnotesize 0.1849&\footnotesize 0.0864&\footnotesize 0.4471&\footnotesize 0.3794&\footnotesize 0.3089&\footnotesize 0.2443\\

\cline{2-11}


&\multirow{3}{*}{\footnotesize Laplacian}& \multicolumn{1}{|c|}{\footnotesize \(R^2\)}&\footnotesize\textbf{1.0000}&\footnotesize \textbf{1.0000}&\footnotesize \textbf{0.9993}&\footnotesize \textbf{0.9997}&\footnotesize \textbf{0.9999}&\footnotesize \textbf{0.9999}&\footnotesize \textbf{1.0000}&\footnotesize \textbf{1.0000}\\

&&\multicolumn{1}{|c|}{\footnotesize RMSE}&\footnotesize\textbf{0.0000}&\footnotesize \textbf{0.0000}&\footnotesize \textbf{0.0602}&\footnotesize \textbf{0.0398}&\footnotesize \textbf{0.0281}&\footnotesize \textbf{0.0204}&\footnotesize \textbf{0.0000}&\footnotesize \textbf{0.0000}\\

&&\multicolumn{1}{|c|}{\footnotesize MAE}&\footnotesize\textbf{0.0000}&\footnotesize \textbf{0.0000}&\footnotesize \textbf{0.0412}&\footnotesize \textbf{0.0266}&\footnotesize \textbf{0.0184}&\footnotesize \textbf{0.0132}&\footnotesize \textbf{0.0000}&\footnotesize \textbf{0.0000}\\

\cline{2-11}


&\multirow{3}{*}{\footnotesize Sigmoid}& \multicolumn{1}{|c|}{\footnotesize \(R^2\)}&\footnotesize0.9294&\footnotesize 0.8485&\footnotesize 0.8796&\footnotesize 0.6078&\footnotesize 0.6826&\footnotesize -4.6875&\footnotesize 0.4620&\footnotesize -9101.2820\\

&&\multicolumn{1}{|c|}{\footnotesize RMSE}&\footnotesize0.7593&\footnotesize 0.9347&\footnotesize 0.8167&\footnotesize 1.4820&\footnotesize 1.3557&\footnotesize 5.8739&\footnotesize 1.8615&\footnotesize 250.6840\\

&&\multicolumn{1}{|c|}{\footnotesize MAE}&\footnotesize0.4922&\footnotesize 0.5471&\footnotesize 0.4609&\footnotesize 0.8870&\footnotesize 0.8745&\footnotesize 3.3699&\footnotesize 1.0980&\footnotesize 153.6522\\

\cline{2-11}


&\multirow{3}{*}{\footnotesize Cosine}& \multicolumn{1}{|c|}{\footnotesize \(R^2\)}&\footnotesize-6.6762&\footnotesize -10.8590&\footnotesize -11.6306&\footnotesize -11.6133&\footnotesize -11.2211&\footnotesize -10.6022&\footnotesize -9.7986&\footnotesize -8.9055\\

&&\multicolumn{1}{|c|}{\footnotesize RMSE}&\footnotesize7.9167&\footnotesize 8.2705&\footnotesize 8.3644&\footnotesize 8.4045&\footnotesize 8.4116&\footnotesize 8.3895&\footnotesize 8.3394&\footnotesize 8.2697\\

&&\multicolumn{1}{|c|}{\footnotesize MAE}&\footnotesize7.4446&\footnotesize 7.9173&\footnotesize 8.0274&\footnotesize 8.0670&\footnotesize 8.0634&\footnotesize 8.0236&\footnotesize 7.9482&\footnotesize 7.8462\\

\cline{1-11}

\end{tabular}
\end{adjustbox}
\caption{\textbf{(\(\mathscr{O}_{Sum}\) with \(\mathscr{H}_4\))} Learning performance metrics for the models trained on all 1000 training datapoints using the hyperparameters listed in Table \ref{BestHyperparametersH4}. The metrics are calculated based on the predictions made by the models on the training and testing sets with labels determined by \(\mathscr{O}_{Sum}\) with \(\mathscr{H}_4\). The best values for each qubit number and each metric are in bold text for the training and testing datasets.}
\label{ResultsOSumH4}
\end{table}

\end{document}